\documentclass[10pt,letterpaper]{article}

\usepackage{latexsym}
\usepackage{amsthm}
\usepackage{amssymb}
\usepackage{amsmath}
\usepackage{amsfonts}
\usepackage{mathtools}
\usepackage{nicefrac}
\usepackage[mathscr]{euscript}
\usepackage{todonotes}

\usepackage[letterpaper, top=.95in, bottom=1in, left=1.2in, right=1.2in, marginparwidth=2.3cm, marginparsep=1mm, includefoot]{geometry}

\usepackage{color,url,bm,cite,microtype}
\usepackage{verbatim,listings}
\usepackage{todonotes}

\usepackage[font=small, labelfont=bf]{caption}
\usepackage{setspace}
\usepackage{subcaption}
\usepackage{setspace}

\usepackage[
		pdftex, plainpages = false, pdfpagelabels, 
                 bookmarks = true,
                 bookmarksopen = true,
                 bookmarksnumbered = true,
                 breaklinks = true,
                 linktocpage,
                 pagebackref,
                 hyperindex = true,
                 hyperfigures
                  ]
                {hyperref}
                     
\usepackage[inline]{enumitem}
\usepackage{dsfont}
\usepackage{tikz}
\usepackage[ruled,linesnumbered]{algorithm2e}

\RequirePackage[T1]{fontenc}
\RequirePackage[utf8]{inputenc}
\usepackage[english]{babel}
\usepackage{alphabeta}

\makeatletter
\input{colordvi}

\usepackage{aliascnt}

\definecolor{MidnightBlack}{rgb}{0.1,0.1,.34}
\definecolor{MidnightBlue}{rgb}{0.1,0.1,0.43}
\definecolor{Black}{rgb}{0,0, 0}
\definecolor{Blue}{rgb}{0, 0 ,1}
\definecolor{Red}{rgb}{1, 0 ,0}
\definecolor{White}{rgb}{1, 1, 1}
\definecolor{grey}{rgb}{.6, .6, .6}
\definecolor{Mygreen}{rgb}{.0, .7, .0}
\definecolor{Yellow}{rgb}{.55,.55,0}
\definecolor{Mustard}{rgb}{1.0, 0.86, 0.35}
\definecolor{applegreen}{rgb}{0.55, 0.71, 0.0}
\definecolor{darkturquoise}{rgb}{0.0, 0.81, 0.82}
\definecolor{celestialblue}{rgb}{0.29, 0.59, 0.82}
\definecolor{green_yellow}{rgb}{0.68, 1.0, 0.18}
\definecolor{crimsonglory}{rgb}{0.75, 0.0, 0.2}
\definecolor{darkmagenta}{rgb}{0.30, 0.0, 0.30}
\definecolor{magenta}{rgb}{0.50, 0.0, 0.50}
\definecolor{internationalorange}{rgb}{1.0, 0.31, 0.0}
\definecolor{darkorange}{rgb}{1.0, 0.55, 0.0}
\definecolor{ao}{rgb}{0.0, 0.5, 0.0}
\definecolor{awesome}{rgb}{1.0, 0.13, 0.32}
\definecolor{darkcyan}{rgb}{0.0, 0.50, 0.50}
\definecolor{violet}{rgb}{0.93, 0.51, 0.93}
\definecolor{brown}{rgb}{0.65, 0.16, 0.16}
\definecolor{orange}{rgb}{1.0, 0.65, 0.0}
\definecolor{cornflowerblue}{rgb}{0.39, 0.58, 0.93}
\definecolor{purpleMadness}{rgb}{0.41, 0.28, 0.69}

\newcommand{\rev}[1]

\definecolor{BrilliantRose}{rgb}{1.0, 0.33, 0.64}
\usepackage{todonotes}

\newcommand{\remove}[1]{}

\newcounter{func}
\newcommand{\funref}[1]{\hyperref[#1]{f_{\ref*{#1}}}} 

\usetikzlibrary{math}
\tikzset{black node/.style={draw, circle, fill = black, minimum size = 5pt, inner sep = 0pt}}
\tikzset{white node/.style={draw, circlternary_treese, fill = white, minimum size = 5pt, inner sep = 0pt}}
\tikzset{normal/.style = {draw=none, fill = none}}
\tikzset{lean/.style = {draw=none, rectangle, fill = none, minimum size = 0pt, inner sep = 0pt}}
\usetikzlibrary{decorations.pathreplacing}
\usetikzlibrary{arrows.meta}
\usetikzlibrary{calc}

\usetikzlibrary{shapes}
\tikzset{diam/.style={draw, diamond, fill = black, minimum size = 7pt, inner sep = 0pt}}

\tikzset{
	position/.style args={#1:#2 from #3}{
		at=($(#3)+(#1:#2)$)
	}
}

\tikzset{
  v:main/.style = {draw, circle, scale=0.8, thick,fill=black,inner sep=0.7mm},
  v:ghost/.style = {inner sep=0pt,scale=1},
  v:marked/.style = {circle, scale=1.3, fill=DarkGoldenrod,opacity=0.4},
  >={latex},
  e:main/.style = {line width=1pt}
}

\newcommand{\Ccal}{\mathcal{C}}

\newcommand{\Ecal}{\mathcal{E}}

\newcommand{\Gcal}{\mathcal{G}}
\newcommand{\Hcal}{\mathcal{H}}

\newcommand{\Ocal}{\mathcal{O}}
\newcommand{\Pcal}{\mathcal{P}}

\newcommand\Rcal{\mathcal{R}}

\newcommand{\Tcal}{\mathcal{T}}

\newcommand{\Xcal}{\mathcal{X}}

\newcommand{\Zcal}{\mathcal{Z}}

\newcommand{\Nbbb}{\mathbb{N}}

\newcommand{\Sbbb}{\mathbb{S}}

\RequirePackage{stmaryrd}
\usepackage{textcomp}
\DeclareUnicodeCharacter{2286}{\subseteq}
\DeclareUnicodeCharacter{2192}{\ifmmode\to\else\textrightarrow\fi}
\DeclareUnicodeCharacter{2203}{\ensuremath\exists}
\DeclareUnicodeCharacter{183}{\cdot}
\DeclareUnicodeCharacter{2200}{\forall}
\DeclareUnicodeCharacter{2264}{\leq}
\DeclareUnicodeCharacter{2265}{\geq}
\DeclareUnicodeCharacter{8614}{\mathbin{\mapsto}}
\DeclareUnicodeCharacter{8656}{\Leftarrow}
\DeclareUnicodeCharacter{8657}{\Uparrow}
\DeclareUnicodeCharacter{8658}{\Rightarrow}
\DeclareUnicodeCharacter{8659}{\Downarrow}
\DeclareUnicodeCharacter{8669}{\rightsquigarrow}
\newcommand{\eqdef}{\stackrel{{\scriptsize\rm def}}{=}}
\DeclareUnicodeCharacter{8797}{\eqdef}
\DeclareUnicodeCharacter{8870}{\vdash}
\DeclareUnicodeCharacter{8873}{\Vdash}
\DeclareUnicodeCharacter{22A7}{\models}
\DeclareUnicodeCharacter{9121}{\lceil}
\DeclareUnicodeCharacter{9123}{\lfloor}
\DeclareUnicodeCharacter{9124}{\rceil}
\DeclareUnicodeCharacter{2208}{\in}
\DeclareUnicodeCharacter{9126}{\rfloor}
\DeclareUnicodeCharacter{9655}{\triangleright}
\DeclareUnicodeCharacter{9665}{\triangleleft}
\DeclareUnicodeCharacter{9671}{\diamond}
\DeclareUnicodeCharacter{9675}{\circ}
\DeclareUnicodeCharacter{10178}{\bot}
\DeclareUnicodeCharacter{10214}{} 
\DeclareUnicodeCharacter{10215}{} 
\DeclareUnicodeCharacter{10229}{\longleftarrow}
\DeclareUnicodeCharacter{10230}{\longrightarrow}
\DeclareUnicodeCharacter{10231}{\longleftrightarrow}
\DeclareUnicodeCharacter{10232}{\Longleftarrow}
\DeclareUnicodeCharacter{10233}{\Longrightarrow}
\DeclareUnicodeCharacter{10234}{\Longleftrightarrow}
\DeclareUnicodeCharacter{10236}{\longmapsto}
\DeclareUnicodeCharacter{10238}{\Longmapsto} 
\DeclareUnicodeCharacter{10503}{\Mapsto}    
\DeclareUnicodeCharacter{10971}{\mathrel{\not\hspace{-0.2em}\cap}}
\DeclareUnicodeCharacter{65294}{\ldotp}
\DeclareUnicodeCharacter{65372}{\mid}

\definecolor{Red}{rgb}{1, 0 ,0}
\definecolor{Blue}{rgb}{0, 0 ,1}

\newtheorem{theorem}{Theorem}[section]

\newaliascnt{question}{theorem}

\aliascntresetthe{question}

\newaliascnt{lemma}{theorem}
\newtheorem{lemma}[lemma]{Lemma}
\aliascntresetthe{lemma}

\newaliascnt{claim}{theorem}

\aliascntresetthe{claim}

\newaliascnt{invariant}{theorem}

\aliascntresetthe{invariant}

\newaliascnt{proposition}{theorem}
\newtheorem{proposition}[proposition]{Proposition}
\aliascntresetthe{proposition}

\newaliascnt{observation}{theorem}
\newtheorem{observation}[observation]{Observation}
\aliascntresetthe{observation}

\newaliascnt{corollary}{theorem}

\aliascntresetthe{corollary}

\newaliascnt{definition}{theorem}
\newtheorem{definition}[definition]{Definition}
\aliascntresetthe{definition}

\newaliascnt{conjecture}{theorem}

\aliascntresetthe{conjecture}

\newaliascnt{counterexample}{theorem}

\aliascntresetthe{counterexample}

\newcommand{\hh}{\end{document}}

\newcommand{\p}{{\sf p}}
\newcommand{\etw}{{\sf etw}\xspace}
\newcommand{\rw}{{\sf rw}\xspace}
\newcommand{\cw}{{\sf cw}\xspace}%
\newcommand{\tpw}{{\sf tpw}\xspace}
\newcommand{\cvw}{{\sf cvw}\xspace}

\newcommand{\maxleaf}{{\sf maxleaf}\xspace}
\newcommand{\eadm}{{\sf eadm}\xspace}
\newcommand{\stcw}{{\sf stcw}\xspace}

\newcommand{\obs}{{\sf obs}}

\newcommand{\excl}{{\sf excl}}

\newcommand{\dom}{\textsf{dom}}

\newcommand{\cobs}{\mbox{\rm \textsf{cobs}}}
\newcommand{\pobs}{\mbox{\rm \textsf{pobs}}}

\newcommand{\down}[2]{\mathsf{Down}_{#1}(#2)}

\newcommand{\scrossing}{\Gcal_{{\text{\rm  \textsf{singly-crossing}}}}}
\newcommand{\gall}{\Gcal_{{\text{\rm  \textsf{all}}}}}
\newcommand{\gplanar}{\Gcal_{\text{\rm  \textsf{planar}}}}

\newcommand{\gforest}{\Gcal_{\text{\rm  \textsf{forest}}}}
\newcommand{\gtforest}{\Gcal_{\text{\rm  \textsf{subcubic forest}}}}
\newcommand{\glforest}{\Gcal_{\text{\rm  \textsf{linear forest}}}}

\newcommand{\gouterplanar}{\Gcal_{\text{\rm \textsf{outerplanar}}}}
\newcommand{\gfapex}{\Gcal_{\text{\rm \textsf{apex forest}}}}

\newcommand{\gcircle}{\Gcal_{\text{\rm \textsf{circle}}}}

\newcommand{\minorsoftriangles}{\Gcal_{\text{\rm \textsf{triangles minors}}}}
\newcommand{\caterpillars}{\Gcal_{\text{\rm \textsf{caterpillars}}}}
\newcommand{\ladderminors}{\Gcal_{\text{\rm \textsf{ladder minors}}}}
\newcommand{\vpaths}{\Gcal_{\text{\rm \textsf{path vminors}}}}
\newcommand{\shallowvortexminors}{\Gcal_{\text{\rm \textsf{svminors}}}}
\newcommand{\nton}{\mathbb{N}\to\mathbb{N}\xspace}
\newcommand{\hw}{{\sf hw}\xspace}
\newcommand{\sctw}{{\sf sc\mbox{-}tw}\xspace}
\newcommand{\svtw}{{\sf sv\mbox{-}tw}\xspace}
\newcommand{\tw}{{\sf tw}\xspace}
\newcommand{\size}{{\sf size}\xspace}
\newcommand{\esize}{{\sf esize}\xspace}
\newcommand{\psize}{{\sf psize}\xspace}
\newcommand{\bipw}{{\sf bi\mbox{-}pw}\xspace}
\newcommand{\pw}{{\sf pw}\xspace}
\newcommand{\btd}{{\sf btd}\xspace}
\newcommand{\td}{{\sf td}\xspace}
\newcommand{\fvs}{{\sf fvs}\xspace}
\newcommand{\apex}{{\sf apex}\xspace}
\newcommand{\barrier}{{\sf barrier}\xspace}
\newcommand{\bed}{{\sf bed}\xspace}
\newcommand{\conn}{{\sf conn}\xspace}
\newcommand{\vc}{{\sf vc}\xspace}
\newcommand{\lrw}{{\sf lrw}\xspace}
\newcommand{\rkd}{{\sf rkd}\xspace}
\newcommand{\tcw}{{\sf tcw}\xspace}
\newcommand{\wn}{{\sf wn}\xspace}
\newcommand{\edforest}{{\sf edforest}\xspace}
\newcommand{\edpforest}{{\sf edpforest}\xspace}
\newcommand{\apexouter}{{\sf apexouter}\xspace}

\newcommand{\cupall}{{\pmb{\bigcup}}}

\newcommand{\adh}{\ensuremath{\mathsf{adh}}}
\newcommand{\width}{\ensuremath{\mathsf{width}}}

\newcommand*\samethanks[1][\value{footnote}]{\footnotemark[#1]}

\newcommand{\lin}[1]{\langle #1\rangle}

\newcommand{\numen}[1]{\ifthenelse{\not\equal{#1}{1}}{#1}{}}
\newcommand{\nn}[2]{\Nbbb^{\numen{#1}}\to\Nbbb^{\numen{#2}}}

\definecolor{vagelisColour}{RGB}{0, 65, 130}

\newcommand{\closure}[2]{{\downarrow}_{#1}{#2}}

\newcommand{\UNB}{\mathbb{U}}

\usepackage{marginnote}
\usepackage{titlesec}
\usepackage{cleveref}
\usepackage{chngcntr}

\hypersetup{
    colorlinks,
    linkcolor   = purpleMadness,
    citecolor   = purpleMadness,
    urlcolor    = purpleMadness,
    anchorcolor = purpleMadness
}

\usepackage{tikzscale}

\usepackage{graphicx}

\begin{document}

\title{An Overview of Universal Obstructions for Graph Parameters\thanks{All authors are supported by the French-German Collaboration ANR/DFG Project UTMA ANR-20-CE92-0027 and  by the ANR project GODASse ANR-24-CE48-4377. The second author is also supported by the ERC project BUKA (n°\! 101126229). The third author is also  supported  by the MEAE and the MESR, via the Franco-Norwegian project PHC Aurora project n. 51260WL (2024-5) and  the France 2030 grant reference number
ANR-24-RRII-0002 operated by the Inria Quadrant Program.}}

\author{ 
Christophe Paul\thanks{LIRMM, Univ Montpellier, CNRS, Montpellier, France.}
\and 
Evangelos Protopapas\thanks{Faculty of Mathematics, Informatics and Mechanics, University of Warsaw, Poland.}
\and
Dimitrios  M. Thilikos\samethanks[2]
}
\date{}

\maketitle
\thispagestyle{empty}

\begin{abstract}
\noindent
In a recent work, we introduced a parametric framework for obtaining obstruction characterizations of graph parameters with respect to a quasi-ordering $\leqslant$ on graphs.
Towards this, we proposed the concepts of \textsl{class obstruction},  \textsl{parametric obstruction}, and  \textsl{universal obstruction} as combinatorial objects that determine the approximate behaviour of a graph parameter.
In this work, we explore its potential as a unifying framework for classifying graph parameters.
Under this framework, we survey existing graph-theoretic results on many known graph parameters.
Additionally, we provide some unifying results on their classification.
\end{abstract}


\newpage
\thispagestyle{empty}
\tableofcontents

\newpage

\section{Introduction}

A \emph{graph parameter} is a function $\p$ mapping graphs to integers.
Graph parameters aim at capturing the combinatorial structure of graphs and, as revealed by the successful development of  parameterized complexity theory~\cite{DowneyF99Parameterized}, are central to the design and complexity analysis of algorithms.  
A large variety of graph parameters can be defined, ranging from more natural ones, including the number of vertices or edges, to more elaborate ones such as width parameters~\cite{HlinenyOSG07Width} that are associated to diverse types of graph decompositions.
Some parameters capture local properties of graphs, as for example the maximum degree of a graph, while some others express global properties, as for example the size of a maximum matching of a graph.
Graph parameters may be designed to measure the resemblance or the distance to some combinatorial structure which allows to efficiently tackle problems that are \textsf{NP}-hard in general. 
Among the wide zoo of graph parameters, one may wonder what makes a parameter important or why a parameter is more significant than another one.
To answer this legitimate and natural question, we need tools to compare graph parameters.
In turn, comparing graph parameters heavily relies on the parameter descriptions/definitions, that may vary significantly.
These questions – \textsl{evaluating a parameter}, \textsl{comparing parameters}, and \textsl{characterizing a parameter} –  are central to the understanding of graph parameters and certainly do not admit a simple nor a single answer (see \cite{EibenGHJK22AUnifying} for such an attempt, on  branch decomposition based parameters).

\paragraph{Importance of a parameter.}

Following~\cite{GanianHKMORS16Are}, the importance of a parameter can be appreciated by considering: 1) its algorithmic utility; and 2) its {combinatorial} properties. 
Typically, being algorithmically useful for a given parameter means that a broad set of \textsf{NP}-hard problems are fixed parameter tractable with respect to this parameter. Having nice structural properties can be reflected by the central role a parameter plays in proving structural theorems. 
Under this setting, \textsl{treewidth}~\cite{BerteleB72nonse,RobertsonS84GMIII} can be considered as an important parameter for both criteria (see \autoref{treewidth_sec} for a definition of treewidth).
On one hand, Courcelle's theorem~\cite{Courcelle90them,Courcelle92,Courcelle97} establishes that every graph property that is expressible in monadic second order logic can be checked in linear time on graphs of bounded treewidth.
On the other hand, treewidth has been a cornerstone parameter of the Graph Minors series of Robertson and Seymour towards proving Wagner's conjecture~\cite{Wagner37Uber,RobertsonS04GMXX}: \textsl{the class of all graphs is well-quasi-ordered (well-quasi-ordering) by the minor relation} (now known as the Robertson and Seymour theorem). 

A common trait of many graph parameters that are considered as being ``important'', is to be \textsl{monotone} with respect to some quasi-ordering relation $\leqslant$ on the set of all graphs.
We say that a parameter $\p$ is \emph{$\leqslant$-monotone} if, for every two graphs $H$ and $G$, $H \leqslant G$ implies that  $\p(H) \leqslant \p(G)$.
For instance, treewidth is known to be minor-monotone (see \autoref{rfisondlclsidop}) and \textsl{cutwidth} is known to be immersion-monotone (see \autoref{ommsers}).
Since both the minor and the immersion relation are well-quasi-orderings on the set $\gall$ of all graphs~\cite{RobertsonS04GMXX, RobertsonS10GMXXIII}, this implies that for every minor/immersion-monotone parameter $\p$ and every $k \in \mathbb{N}$, the class 
\begin{eqnarray}
\Gcal_{\p,k} &  \coloneqq & \{G \in \gall \mid \p(G) \leq k \} \label{def_class_k}
\end{eqnarray}
is minor/immersion-closed and therefore it is characterized by the exclusion of a \textsl{finite set} of obstructions (minor-minimal, respectively immersion-minimal, graphs that do not belong to $\Gcal_{\p,k}$) as a minor/immersion.
Moreover, if checking whether a graph $G$ contains some fixed graph as a minor/immersion can be done in polynomial time, then being a well-quasi-ordering automatically implies a polynomial membership decision algorithm for $\Gcal_{\p,k}$.
This is the case for both the minor and the immersion relation~\cite{Wollan16Finding, GroheKMW11Finding, RobertsonS95GMXIII, KawarabayashiKR12Thedisjoint}.
Whether or when such an approach can lead to the proof of the existence of polynomial time decision algorithms for other parameters (by choosing appropriately the relation $\leqslant$) is a running open project in algorithmic graph theory.

 \paragraph{Comparing graph parameters.}
 
 Given two graph parameters $\p$ and $\p'$, we say that $\p$ is \emph{(approximately) smaller} than $\p'$, which we denote $\p \preceq \p'$, if there exists a function $f \colon \mathbb{N} \rightarrow \mathbb{N}$ such that for every graph $G$, $\p(G) \leq f(\p'(G))$.
Observe that if $\p \preceq \p'$, then the graph classes where $\p$ is bounded are ``(approximately) more general'' than those where $\p'$ is bounded, in the sense that $\Gcal_{\p',k} \subseteq \Gcal_{\p,f(k)}$, for some $f \colon \mathbb{N}\to \mathbb{N}$.
This means that, if some problem can be solved efficiently when $\p$ is bounded, then it can also be solved efficiently when $\p'$ is bounded.
In that sense, the algorithmic applicability of $\p$ is wider than the one of $\p'$ (for this particular problem).
On the other side, one may expect (and it is frequently the case) that more problems can be solved efficiently when $\p'$ is bounded: while $\p'$-bounded graphs are more restricted, they may also be ``more structured'' than the class of $\p$-bounded graphs, which may give rise to more algorithmic applications.

We say that $\p$ and $\p'$ are \emph{equivalent} if $\p \preceq \p'$ and $\p' \preceq \p$.
Consider the following example.
The parameter \textsl{cliquewidth}~\cite{CourcelleMR00Linear} is known to be monotone under the induced subgraph relation.
Unfortunately, the induced subgraph relation does not share properties such as those discussed above for the minor or immersion relation (see \autoref{induced_parameters_subsec} for more on induced subgraph-monotone parameters).
In fact, for now, no efficient algorithm is known to decide whether the input graph has cliquewidth at most $k$, for a fixed non-negative integer $k$.
This motivated the introduction of the \textsl{rankwidth} parameter~\cite{OumS06appro}, that is equivalent to cliquewidth and is monotone under the \textsl{vertex minor} relation (see \autoref{vesjrs_misos} for the definitions).
That way, rankwidth offered an alternative to cliquewidth that, in many aspects, enjoys better combinatorial and algorithmic properties \cite{HlinenyO08FindingBranch}.

\medskip
From the discussion above, we conclude that comparing graph parameters is important.
However, this is not always a trivial task.
For example, it is clear from their respective definitions that treewidth is (approximately) smaller than the parameter \textsl{pathwidth} (see \autoref{pathwidth_def} for a definition of pathwidth).
On the other side, the equivalence of  \textsl{rankwidth} and \textsl{cliquewidth}  is not a direct 
consequence of  their definitions.

The question is then to identify conditions under which parameters can be described in a common framework to facilitate their comparison and understand their relative relationships.

\paragraph{Parameter description.}

Again, the case of treewidth is very interesting to discuss.
First of all there are many ways to define it, including graph decompositions, vertex orderings, extremal graphs, and obstructions.
In~\cite{HarveyW17Parameters}, no less than 15 parameters are identified to be either equal or equivalent to treewidth.
Each of the above definitions is putting light on different aspects of the treewidth parameter and its relatives.
Some of them, such as \textsl{branchwidth}~\cite{RobertsonS91GMX} may be preferred to treewidth for the sake of algorithm design~\cite{DornT09Semi}, some others, such as the \textsl{(strict) bramble number}~\cite{LardasPTZ23OnStrict, SeymourT93graph, BellenbaumD02Two, Diestel10grap} or the \textsl{largest grid-minor} in the graph~\cite{ChuzhoyT21Towards, GuT12ImprovedBounds}, are more convenient to certify large treewidth.
A similar situation holds for pathwidth: it enjoys many alternative definitions and there are many parameters that are equivalent to it~\cite{Kinnersley92TheVertex, Moehring90grap, KirousisP85inte}.
This plethora of viewpoints for the same parameter is a strong sign of importance and certainly ease the comparison with other parameters.
However, such a diversity of definitions is not the case for other parameters.
For this reason, it is important to have a \textsl{universal framework}  for dealing with the approximate behaviour of parameters.

In \cite{PaulPT2023GraphParameters} such a framework was proposed.
The first step in this direction is to group parameters together with respect to some quasi-ordering $\leqslant$ on graphs under which they are monotone.
When focussing on $\leqslant$-monotone parameters, one can observe that many results in the literature fulfill the following  pattern.

 \begin{quote}
 \textsl{Let $\p$ be a graph parameter that is $\leqslant$-monotone with respect to a quasi-ordering relation $\leqslant$.
 There exists a finite set $\mathbb{Z}$ of $\leqslant$-closed classes such that for every $\leqslant$-closed class $\Gcal$, there exists a non-negative integer $c_{\Gcal}$ so that $\p(G) \leq c_{\Gcal},$ for every graph $G \in \Gcal$ (i.e. $\p$ is bounded in $\Gcal$) if and only if $\Zcal \not\subseteq \Gcal,$ for every $\leqslant$-closed class $\Zcal \in \mathbb{Z}$.}
\end{quote}

Hereafter, we refer to such a set $\mathbb{Z}$ (when it exists) as the \emph{$\leqslant$-class obstruction} of $\p,$ denoted by $\cobs_{\leqslant}(\p)$ (formally defined in \autoref{diidnidialdedfs}).
Let us review a few such results.
First consider the minor relation, which we denote by $\leqslant_{\mathsf{m}}$.
It follows that for treewidth, $\cobs_{\leqslant_{\mathsf{m}}}(\tw)$ is composed of the class of planar graphs~\cite{RobertsonS86GMV,RobertsonST94Quickly}, while for pathwidth, $\cobs_{\leqslant_{\mathsf{m}}}(\pw)$ contains the class of forests~\cite{RobertsonS83GMI,BienstockRST91Quickly}, and for treedepth~\cite{NesetrilO06Treedepth}, $\cobs_{\leqslant_{\mathsf{m}}}(\td)$ is the set containing the class of linear forests\footnote{A linear forest is the disjoint union of paths.}.
In these first three examples, the class obstruction contains a \textsl{unique} class of graphs.
This is not always the case.
According to~\cite{HuynhJMW20Seymour,dang2018minors}, for the \textsl{biconnected pathwidth}, $\cobs_{\leqslant_{\mathsf{m}}}(\bipw)$ contains two graph classes, namely the apex-forests and the outerplanar graphs  (see~\autoref{bicpw}).
If we now consider the immersion relation, denoted by $\leqslant_{\mathsf{i}}$, and the cutwidth parameter, it holds that $\cobs_{\leqslant_{\mathsf{i}}}(\cw)$ is formed by three classes of graphs (see~\autoref{cutwidth_sec}).
In the context of the vertex-minor relation, denoted by $\leqslant_{\mathsf{vm}}$, it has recently been proved that the class obstruction of rankwidth $\cobs_{\leqslant_{\mathsf{vm}}}(\rw)$ consists of the set of circle graphs~\cite{GeelenKMW23Thegrid}, while is it conjectured that for linear rankwidth, $\cobs_{\leqslant_{\mathsf{vm}}}(\lrw)$ consists of the set of trees~\cite{KanteK18Linear}.
As we shall see in this work, the above statement pattern holds for many other $\leqslant$-monotone parameters.

\paragraph{Capturing the asymptotic behavior of monotone parameters.}

Interestingly, the relationship between the $\leqslant$-class obstructions of $\leqslant$-monotone parameters is telling something about the relationship of these parameters.
The fact that forests are planar graphs clearly indicates that pathwidth is (approximately) bigger than treewidth.
It is worth to observe that in every known result fulfilling the above pattern, the class obstruction contains a finite (typically small) number of graph classes.
But still, these classes contain an infinite number of graphs.
Moreover, class obstructions do not constitute ``universal patterns'' of the approximate behaviour of the relative parameters.
As proposed in~\cite{PaulPT2023GraphParameters}, using the concepts of \emph{$\leqslant$-parametric obstructions} and \emph{$\leqslant$-universal obstructions}, it is possible to resolve these two issues.

Let $\p$ be a $\leqslant$-monotone parameter having a $\leqslant$-class obstruction $\cobs_{\leqslant}(\p)$.
We observe that every class of graph $\Gcal \in \cobs_{\leqslant}(\p)$ is $\leqslant$-closed.
This implies that if $\leqslant$ is a well-quasi-ordering, $\Gcal$ admits a \textsl{finite} number of obstructions, denoted by $\obs_{\leqslant}(\Gcal)$, that are the $\leqslant$-minimal graphs not belonging to $\Gcal$.
We can then define the notion of \emph{$\leqslant$-parametric obstruction} of $\p$, denoted by $\pobs_{\leqslant}(\p)$, as the set $\big\{\obs_{\leqslant}(\Gcal) \mid \Gcal \in \cobs_{\leqslant}(\p) \big\}$.

We proceed with an example: in \cite{BienstockRST91Quickly} it was proven that, for every forest $F$ on $k$ vertices, every graph excluding $F$ as a minor has pathwidth at most $k-2$.
As $\{K_3\}$ is the minor-obstruction of forests, this implies that $\cobs_{\leqslant_{\mathsf{m}}}(\pw)$ consists of the graph class of forests and that $\pobs_{\leqslant_{\mathsf{m}}}(\pw) = \big\{ \{ K_3 \} \big\}$.
Similarly, as we see in the coming sections, all the aforementioned examples of monotone parameters (treewidth, treedepth, biconnected pathwidth, cutwidth, rankwidth) and many others enjoy a \textsl{finite description} by means of parametric obstructions.
However, as noticed in \cite{PaulPT2023GraphParameters}, whether such a characterization always exists for $\leqslant$-monotone parameters is open and is linked to important conjectures of order theory that go further than the well-quasi-ordering of $\leqslant$.

\medskip
As demonstrated in \cite{PaulPT2023GraphParameters} the concept of $\leqslant$-universal obstruction of a $\leqslant$-monotone parameter $\p$ is closely related to the classes in $\cobs_{\leqslant}(\p).$
A \emph{$\leqslant$-parametric graph} is any non $\leqslant$-decreasing sequence of graphs $\mathscr{H} = \langle \mathscr{H}_t \rangle_{t \in \mathbb{N}}$, i.e., such that for every $i \leq j$, $\mathscr{H}_i \leqslant \mathscr{H}_j$.
We say that a $\leqslant$-parametric graph $\mathscr{H}$ is a \emph{$\leqslant$-omnivore} of the graph class $\Gcal$ if every graph in $\mathscr{H}$ belongs to $\Gcal$ and, moreover, for every $G \in \Gcal$, there exists $k \in \mathbb{N}$ such that $G \leqslant \mathscr{H}_k$.
We say that a finite set of $\leqslant$-parametric graphs $\mathfrak{H}$ is a \emph{$\leqslant$-universal obstruction} for $\p$ if and only if it contains a $\leqslant$-omnivore for every class of graphs contained in $\cobs_{\leqslant}(\p)$.
As explained in \autoref{diidnidialdedfs}, to a $\leqslant$-universal obstruction for a parameter $\p$, one can naturally associate a parameter $\p_{\mathfrak{H}}$ that is equivalent to $\p$.
 
Let us again illustrate these concepts with pathwidth.
It is well-known that every forest is a minor of a complete ternary tree\footnote{A \emph{complete ternary tree} is a tree where all internal vertices have degree three and where each leaf belongs to some anti-diametrical path.} (depicted in \autoref{fig_ternary_trees}).
It follows that the sequence $\mathscr{T} = \langle \mathscr{T}_t \rangle_{t \geq 1}$ of ternary trees is a $\leqslant_{\mathsf{m}}$-omnivore of forests.
This implies that $\{\mathscr{T}\}$ is a $\leqslant_{\mathsf{m}}$-universal obstruction for pathwidth.
This, combined with the results of \cite{BienstockRST91Quickly}, implies that the pathwidth of a graph is equivalent to the biggest size of a complete ternary tree minor it contains.
Therefore, $\p_{\{\mathscr{T}\}}$ is equivalent to pathwidth and the set $\{\mathscr{T}\}$ is a $\leqslant_{\mathsf{m}}$-universal obstruction for pathwidth.
Also, note that there might exist many different universal obstructions for the same parameter.
For instance,  one may also consider the set consisting of the sequence of  ``quaternary trees'' that is also a $\leqslant_{\mathsf{m}}$-universal obstruction for pathwidth.

\medskip
The value of universal obstructions is that they provide a ``common for all'' framework yielding ``regular'' alternative definitions of graph parameters (subject to parameter equivalence).
We believe that the notions of class obstruction, parametric obstruction and universal obstruction offer a unifying and concise framework to describe, compare, and evaluate $\leqslant$-monotone parameters.
These concepts were underlying many known results but were only recently formalized in~\cite{PaulPT2023GraphParameters}  (see also \autoref{diidnidialdedfs}).
They offer a global view on the hierarchy of the $\leqslant$-monotone parameters.
Moreover, these concepts form a framework that allows one to define very natural and yet unexplored parameters.
As pointed out in~\cite{PaulPT2023GraphParameters}, under certain circumstances, fixed parameterized approximation algorithms can be automatically derived using this framework.
This work offers a systematic review of known minor-monotone parameters (see \autoref{rfisondlclsidop}), immersion-monotone parameters (see \autoref{ommsers}), vertex-minor-monotone parameters (see \autoref{vesjrs_misos}), and a few examples of parameters that are monotone with respect to other relations that are not well-quasi-orderings (see \autoref{obsjjdjdother_a}), and present them in the unified framework discussed above.
\autoref{diidnidialdedfs} formally defines all the concepts introduced in the discussion above.

\section{Universal obstructions: basic definitions and results}
\label{diidnidialdedfs}

All graphs in this paper are finite and undirected.
Unless the opposite is explicitly mentioned (see \autoref{ommsers} and \autoref{multedge}), graphs are assumed to be simple.
In this section, we introduce the basic concepts that we use in the rest of the paper.

\subsection{General concepts}

Each ordering relation $\leqslant$ on the set of all graphs $\gall$ that we consider is a quasi-ordering, that is, it is reflexive and transitive.
Given a graph class $\Gcal \subseteq \gall$, we say that $\leqslant$ is \emph{well-founded} on $\Gcal$, if for every non-empty $\Hcal \subseteq \Gcal$, $\Hcal$ has a $\leqslant$-minimal element.
All relations on graphs that we consider are well-founded.
A set $\Zcal$ of graphs is a \emph{$\leqslant$-antichain} if its elements are pairwise non $\leqslant$-comparable.

Given some graph class $\Gcal \subseteq \gall$, the quasi-ordering relation $\leqslant$ is a \emph{well-quasi-ordering} on $\Gcal$ if $\leqslant$ is well-founded on $\Gcal$ and all $\leqslant$-antichains on $\Gcal$ are finite.
We say that a graph class $\Gcal \subseteq \gall$ is \emph{$\leqslant$-closed} if for every graph $G \in \Gcal$, if $H \leqslant G$, then $H \in \Gcal$.
We use the notation $\down{\leqslant}{\gall}$ in order to denote the collection of all $\leqslant$-closed classes.

Given a class $\Gcal \in \down{\leqslant}{\gall},$ we define its \emph{$\leqslant$-obstruction set} as the set of $\leqslant$-minimal graphs in $\gall \setminus \Gcal$ and we denote it by $\obs_{\leqslant}(\Gcal)$.
Clearly $\obs_{\leqslant}(\Gcal)$ is always a $\leqslant$-antichain.
Note that if $\leqslant$ is a well-quasi-ordering on $\gall$, then for every $\Gcal \in \down{\leqslant}{\gall}$, $\obs_{\leqslant}(\Gcal)$ is a finite set of graphs.
Given a set of graphs $\Zcal \subseteq \gall,$ we define
\begin{align*}
\excl_{\leqslant }(\Zcal) \ \coloneqq \ \big\{G\in\gall\mid \forall Z\in\Zcal: Z\nleq G\big\}.
\end{align*}

If $\Zcal$ is a $\leqslant $-antichain, the graphs of $\Zcal$ are the \emph{$\leqslant $-obstructions} of $\excl_{\leqslant }(\Zcal)$.
Clearly, for every  $\Gcal \in \down{\leqslant}{\gall}$, it holds that $\excl_{\leqslant }(\obs_{\leqslant }(\Gcal)) = \Gcal$.

\subsection{Graph parameters}
\label{sliubeabaueike}

A \emph{graph parameter} is a function mapping graphs to non-negative integers and infinity, i.e. $\p \colon \gall \to \mathbb{N} \cup \{ \infty \}$.
We insist that $\p$ is an invariant under graph isomorphism.
For a graph parameter $\p$, we write $\dom(\p) = \{ G \in \gall \mid \p(G) \in \mathbb{N}\}$.
We say that $\p$ is \emph{$\leqslant $-monotone} if, for every graph $G$ and every $H \leqslant G$, we have that $\p(H) \leqslant \p(G)$.
 
Let $\p$ and $\p'$ be two graph parameters.
We define $\p + \p'$ as the parameter where $(\p+\p')(G) = \p(G) + \p'(G).$
We also define $\max\{ \p, \p' \}$ as the parameter where $(\max\{\p,\p'\})(G) = \max\{\p(G), \p'(G)\}$.
Similarly, for $c \in \mathbb{N}$, $\p+ c$ is defined so that $(\p+ c)(G) = \p(G)+ c$.
 
We write $\p \preceq \p'$ if $\dom(\p') \subseteq \dom(\p)$ and there exists a function $f \colon \mathbb{N} \to \mathbb{N}$ such that, for every graph $G \in \dom(\p')$ it holds that $\p(G) \leqslant  f(\p'(G))$.
We also say that $\p$ and $\p'$ are \emph{equivalent} which we denote by $\p \sim \p'$, if $\p \preceq \p'$ and $\p' \preceq \p$. 
Note that $\p \sim \p'$ holds if and only if $\dom(\p) = \dom(\p')$ and there exists a function $f \colon \mathbb{N}\to\mathbb{N}$ such that for every graph $G$, $\p(G) \leqslant  f(\p'(G))$ and $\p'(G) \leqslant  f(\p(G))$.
We call the function $f$ \emph{the gap function} (or just \emph{the gap}) of the equivalence between $\p$ and $\p'$.
If the gap function $f$ is polynomial (resp. linear), then we say that $\p$ and $\p'$ are \emph{polynomially (resp. linearly) equivalent} and we denote this by $\p \sim_{\mathsf{P}}\p'$ (resp. $\p \sim_{\mathsf{L}}\p'$).

\paragraph{Class and parametric obstructions.}

We say that a collection of graph classes $\mathbb{W} \subseteq \down{\leqslant}{\gall}$ is \emph{upward-closed}, if for every $\Hcal \in \mathbb{W}$ and $\Gcal \in \down{\leqslant}{\gall}$, $\Hcal \subseteq \Gcal$ implies that $\Gcal \in \mathbb{W}$.
Let $\Gcal \in \down{\leqslant}{\gall}$.
We say that a $\leqslant$-monotone parameter $\p$ is \emph{unbounded} in $\Gcal$ if there is no $c \in \mathbb{N}$ such that $\p(G) \leqslant c$ for every $G \in \Gcal$.
For every $\leqslant$-monotone parameter $\p$, we define
\begin{align*}
\UNB_{\leqslant}(\p) \ := \ \big\{{\Gcal \subseteq \gall} \mid \text{$\Gcal$ is a $\leqslant $-closed class such that $\p$ is unbounded in $\Gcal$}\big\}.
\end{align*}

Observe that, by definition, if $\p$ is unbounded in $\Gcal$ and $\Hcal$ is a $\leqslant$-closed class where 
$\Gcal \subseteq \Hcal$, then $\p$ is also unbounded in $\Hcal$.
This implies that $\UNB_{\leqslant}(\p)$ is upward-closed.

\begin{definition}[Class obstruction]
Let  $\p$ be a $\leqslant$-monotone parameter.
If the set of $\subseteq$-minimal elements of $\UNB_{\leqslant}(\p)$ exists, then we call it the \emph{$\leqslant$-class obstruction} of $\p$ and we denote it by $\cobs_{\leqslant}(\p)$.
\end{definition}

\begin{definition}[Parametric obstruction]
Let $\p$ be a $\leqslant$-monotone parameter having a $\leqslant$-class obstruction.
We define the \emph{$\leqslant$-parametric obstruction} of $\p$ as the set of $\leqslant$-obstruction sets $\{ \obs_{\leqslant}(\Gcal) \mid \Gcal \in \cobs_{\leqslant }(\p) \}$ and we denote it by $\pobs_{\leqslant}(\p)$.
\end{definition}

The next proposition follows easily from the definitions (see \cite{PaulPT2023GraphParameters}).

\begin{proposition}\label{prop_wqo_cobs_pobs}
If a quasi-ordering $\leqslant$ is a well-quasi-ordering on $\gall$, then for every $\leqslant$-monotone parameter, the set $\cobs_\leqslant(\p)$ (and therefore also $\pobs_\leqslant(\p)$) exists.
\end{proposition}

Note that the well-quasi-ordering assumption does not imply that $\pobs_\leqslant(\p)$ is a finite set.
However it indeed implies that each of the sets it contains is a finite set of graphs. 

\subsection{Parametric graphs}

A \emph{$\leqslant$-parametric graph} is any \emph{$\leqslant$-increasing} sequence of graphs $\mathscr{H} := \langle \mathscr{H}_{t} \rangle_{t \in \mathbb{N}}$ indexed by non-negative integers, i.e., such that for every $i \leq j,$ $\mathscr{H}_{i} \leqslant \mathscr{H}_{j}.$
We define the \emph{$\leqslant$-closure} of a $\leqslant$-parametric graph $\mathscr{H}$ as follows.
\begin{align*}
\closure{\leqslant}{\mathscr{H}}\ \coloneqq \ \big\{ G \in \gall \mid \exists k \in \mathbb{N} : G \leqslant \mathscr{H}_{k} \big\}.
\end{align*}

In what follows, when we introduce some parametric graph, it is frequently convenient to start the indexing from numbers larger than zero.
In these cases, we may assume that the omitted graphs are all the empty graph.

Given two $\leqslant$-parametric graphs $\mathscr{H}$ and $\mathscr{F}$, we write $\mathscr{H} \lesssim \mathscr{F}$ if there exists a function $f \colon \mathbb{N} \to \mathbb{N}$ such that $\mathscr{H}_k \leqslant  \mathscr{F}_{f(k)}$ and $\mathscr{H} \approx \mathscr{F}$ if $\mathscr{H} \lesssim \mathscr{F}$ and $\mathscr{F} \lesssim \mathscr{H}$. 
As we did for $\sim$, we call $f$ \emph{the gap} of the equivalence $\approx$ and we define relations $\mathscr{H} \approx_{\mathsf{P}} \mathscr{F}$ and $\mathscr{H} \approx_{\mathsf{L}} \mathscr{F}$ in a way analogous to the definitions of $\sim_{\mathsf{P}}$ and $\sim_{\mathsf{L}}$.

To every $\leqslant$-parametric graph $\mathscr{H}$, we associate the $\leqslant$-monotone parameter $\p_{\mathscr{H}}$ defined, so that for every graph $G,$
\begin{eqnarray}
\p_{\mathscr{H},\leqslant}(G) \ \coloneqq \ \inf\{ k \in \mathbb{N} \mid \mathscr{H}_{k} \not\leqslant G \}.\label{def_uioej}
\end{eqnarray}

We continue with the definition of a $\leqslant$-omnivore of a $\leqslant$-closed class.

\begin{definition}
Let $\Gcal$ be a $\leqslant$-closed class.
A $\leqslant$-parametric graph $\mathscr{H}$ is a \emph{$\leqslant$-omnivore} of $\Gcal$ if $\Gcal = \closure{\leqslant}{\mathscr{H}}$. 
\end{definition}

We call any \textsl{finite} set $\mathfrak{H}$ of pairwise $\lesssim$-non-comparable $\leqslant$-parametric graphs a \emph{$\leqslant$-parametric family}.
Let $\mathfrak{H}$ and $\mathfrak{F}$ be $\leqslant$-parametric families.
We write $\mathfrak{H} \equiv \mathfrak{F}$ if there exists a bijection $\sigma \colon \mathfrak{H} \to \mathfrak{F}$ such that for every $\mathscr{H} \in \mathfrak{H}$ it holds that $\mathscr{H} \approx \sigma(\mathscr{H}).$

To every $\leqslant$-parametric family $\mathfrak{H}$, we associate the $\leqslant$-monotone parameter $\p_{\mathfrak{H},\leqslant}$ defined, so that for every graph $G,$
\begin{eqnarray}
\p_{\mathfrak{H},\leqslant}(G) \ \coloneqq \ \begin{cases} \sup\{ \p_{\mathscr{H},\leqslant}(G) \mid \mathscr{H} \in \mathfrak{H} \}, &\text{if $\mathfrak{H} \neq \emptyset$}\\
0, &\text{otherwise}.
\end{cases}\label{def_olkol}
\end{eqnarray}

Note that $\p_{\mathfrak{H},\leqslant}(G) \sim \sum_{\mathscr{H} \in \mathfrak{H}} \p_{\mathscr{H},\leqslant}(G)$, which means that, from the point of view of equivalence, the sum of two parameters is the same as taking the supremum of them.

\begin{definition}[Universal obstruction] We say that a $\leqslant$-parametric family $\mathfrak{H}$ is a \emph{$\leqslant$-universal obstruction} for the graph parameter $\p$ if $\p_{\mathfrak{H},\leqslant} \sim \p$.
\end{definition}

In the definitions \eqref{def_uioej}  and  \eqref{def_olkol}, when the quasi-ordering relation is clear from the context we use the simpler notation $\p_{\mathfrak{H}}$ and $\p_{\mathscr{H}}$.

The following proposition is due to results of \cite{PaulPT2023GraphParameters}.
The first statement follows from Theorems 3.41 and 3.48 in \cite{PaulPT2023GraphParameters}.
The second follows from corollary 3.24 and theorem 3.48 in \cite{PaulPT2023GraphParameters}.

\begin{proposition}[\!\cite{PaulPT2023GraphParameters}]\label{cobs_uobs}
Let $\p$ be a $\leqslant$-monotone parameter.
If $\p$ has a $\leqslant$-universal obstruction $\mathfrak{H}$ then $\cobs_{\leqslant}(\p)$ exists, it is finite, and there exists a bijection $\sigma \colon \mathfrak{H} \to \cobs_{\leqslant}(\p)$ such that every $\mathscr{A} \in \mathfrak{H}$ is a $\leqslant$-omnivore of $\sigma(\mathscr{A})$.
Moreover, if $\leqslant$ is a well-quasi-ordering on $\gall,$ then $\cobs_{\leqslant}(\mathsf{p}
)$ exists, and if it is finite, $\p$ has a $\leqslant$-universal obstruction.
\end{proposition} 

We can now present the equivalence between the four concepts that we defined so far.
We say that a $\leqslant $-monotone parameter $\p$ is \emph{bounded} in some $\leqslant $-closed graph class $\Gcal$ if 
there exists a $c \in \mathbb{N}$ such that $\p(G) \leqslant c$, for every $G \in \Gcal$.

\begin{proposition}[Corollary 3.45 in~\cite{PaulPT2023GraphParameters}] \label{th_equiv}
Let $\p$ and $\p'$ be $\leqslant$-monotone parameters.
If both $\p$ and $\p'$ have $\leqslant$-universal obstructions and both $\cobs_{\leqslant}(\p)$ and $\cobs_{\leqslant}(\p')$ exist, then the following statements are equivalent:
\begin{enumerate}
\item $\p \sim \p'$,
\item $\cobs_{\leqslant}(\p) = \cobs_{\leqslant}(\p')$,
\item $\pobs_{\leqslant}(\p) = \pobs_{\leqslant}(\p')$,
\item for every $\leqslant$-universal obstruction $\mathfrak{H}$ for $\p$ and every $\leqslant$-universal obstruction $\mathfrak{F}$ for $\p'$, $\mathfrak{H} \equiv \mathfrak{F}$, and
\item for every $\leqslant $-closed class $\Gcal$, $\p$ is bounded in $\Gcal$ if and only if $\p'$ is bounded in $\Gcal$.
\end{enumerate}
\end{proposition}

\subsection{Comparing parameters using Smyth extensions}
\label{conbncpsjdsusion}

Let $\bm{\leqslant }$ be a quasi-ordering relation on some set $\mathbf{A}$.
The \emph{Smyth extension} of $\bm{\leqslant}$ (see e.g., \cite{Jancar99ANote}) is the quasi-ordering relation $\bm{\leqslant}^{*}$ defined on $2^{\mathbf{A}}$ such that, for $\mathbf{X}, \mathbf{Y} \in 2^{\mathbf{A}}$, $\mathbf{X}\bm{\leqslant }^*\mathbf{Y}$ iff $\forall \mathbf{y}\in \mathbf{Y}, \exists \mathbf{x}\in \mathbf{X}$ such that $\mathbf{x}\bm{\leqslant }\mathbf{y}$.

\medskip
The proof of \autoref{th_equiv} follows from the fact, shown in \cite[Theorem 3.44]{PaulPT2023GraphParameters}, that the following statements are equivalent:
\begin{eqnarray}
\p & \succeq & \p' \label{first}\\
\cobs_{\leqslant }(\p) & \subseteq^* &\cobs_{\leqslant }(\p'),\label{second}\\ 
\pobs_{\leqslant }(\p) &  \leqslant ^{**} & \pobs_{\leqslant }(\p'),\label{third}\\ 
\mathfrak{H} & \lesssim^*& \mathfrak{F}.\label{fourth}
\end{eqnarray}

Also observe that $\pobs_{\leqslant }(\p)$, when it exists, is a $\leqslant^*$-antichain and that each $\Ocal\in \pobs_{\leqslant}(\p)$ is a $\leqslant$-antichain.
Similarly, $\cobs_{\leqslant}(\p)$, when it exists, is a $\subseteq$-antichain.
Finally if $\mathfrak{H}$ is a $\leqslant$-universal obstruction for $\p$, then $\mathfrak{H}$ is a $\lesssim$-antichain.

\medskip
In the next sections, for each of the parameters that we deal with, we provide all equivalent descriptions corresponding to \autoref{th_equiv}. 
The relations \eqref{first}, \eqref{second}, \eqref{third}, and \eqref{fourth} may serve as useful tools for making approximate comparisons between parameters. 

\subsection{Basic definitions} \label{some_definitions}

Given $a,b\in\mathbb{N}$ we denote the set $\{z\in\mathbb{N} \mid a\leq z\leq b\}$ by $[a,b].$
In case $a>b$ the set $[a,b]$ is empty. For an integer $p\geq 1,$ we set $[p]=[1,p]$ and $\mathbb{N}_{\geq p}=\mathbb{N}\setminus [0,p-1].$ 
Given that $\mathbf{A}$ is a set of objects where the operation $\cup$  is defined, we denote $\cupall \mathbf{A}=\bigcup_{A\in\mathbf{A}}A$.

Given a graph $G$, we use $V(G)$ and $E(G)$ for the vertex and the edge set of $G$, respectively.
We also denote $|G|=|V(G)|$. Given a set $S\subseteq V(G)$, we denote by $G - S$ the graph obtained if we remove the vertices of $S$ from $G$.
Likewise, given a set $F \subseteq E(G)$, we denote by $G - F$ the graph obtained if we remove the edges of $F$ from $G.$
If $a \in V(G)\cup E(G)$ we write $G-a$ instead of $G-\{a\}$.
Also the \emph{subgraph of $G$ induced by $S\subseteq V(G)$} is defined as $G[S] \coloneqq G-(V(G)\setminus S)$.
Given two graphs $G_{1}$ and $G_{2}$ we define $G_{1}\cup G_{2}=(V(G_{1})\cup V(G_{2}),E(G_{1})\cup E(G_{2}))$. Given $v\in V(G)$ we define the \emph{neighbourhood} of $v$ in $G$ as the set $N_{G}(v)=\{u\mid \{v,u\}\in E(G)\}$.
The \emph{degree} of a vertex $v \in V(G)$ is defined as the number of edges that are incident to $v.$
  
Let $e=\{x,y\}$ be an edge in a graph $G$.
The result of the \emph{subdivision} of $e$ in $G$ is the graph obtained from the graph $G-e$ if we add  to it a new vertex $v_{x,y}$ and the two edges $\{x,v_{x,y}\}$ and $\{v_{x,y},y\}$.
A \emph{subdivision} of a graph $H$ is any graph that is obtained from $H$ after a  (possibly empty) sequence of subdivisions.

\section{Obstructions of minor-monotone parameters}
\label{rfisondlclsidop}

We commence our presentation with one of the most studied quasi-ordering relations on graphs, that is the \textsl{minor} relation: a graph $H$ is a \emph{minor} of a graph $G$ if $H$ can be obtained from a subgraph of $G$ after contracting\footnote{The action of \emph{contracting} an edge is the identification of its vertices.
If multiple edges appear after this operation, then we reduce their multiplicity to one.} edges.
We denote this relation by $H \leqslant_{\mathsf{m}} G$.
According to the Robertson and Seymour theorem \cite{RobertsonS04GMXX} the relation $\leqslant_{\mathsf{m}}$ is a well-quasi-ordering on $\gall$.

\subsection{Warm up}
\label{wadfdjgjingngup}

To enhance the reader's familiarity with the concepts tied to the parametric framework introduced in \autoref{diidnidialdedfs}, we  begin our discussion with a few elementary minor-monotone parameters.

\paragraph{Number of vertices.}

The simplest graph parameter of interest that one could consider is perhaps the number of vertices of a graph.
We denote this parameter by $\size$, that is for every graph $G \in \gall,$
\begin{align}
\size(G) \ &\coloneqq \ |G|.\label{size_parameter}
\end{align}

An easy observation is that a minor-universal obstruction for $\size$ is the singleton minor-parametric family $\{\mathscr{K}^{\mathsf{d}}\}$  where $\mathscr{K}^{\mathsf{d}} = \langle k\cdot K_{1} \rangle_{k\in\mathbb{N}}$ consists of the edgeless graphs.
Also note that $\cobs_{\leqslant_{\mathsf{m}}}(\size) = \{ \excl_{\leqslant_{\mathsf{m}}}(K_{2}) \}$ and $\pobs_{\leqslant_{\mathsf{m}}}(\size) = \big\{\{K_{2}\}\big\}.$

\paragraph{Number of edges.}

We may next consider the parameter counting the number of edges of a graph, which we denote $\esize$, that is for every graph $G \in \gall$,
\begin{align}
\esize(G) \ &\coloneqq \ |E(G)|.\label{esize_parameter}
\end{align}

A minor-universal obstruction for $\esize$ is $\{\mathscr{K}^{\mathsf{m}}, \mathscr{K}^{\mathsf{s}}\}$ where 
$\mathscr{K}^{\mathsf{m}} = \langle k\cdot K_{2} \rangle_{k\in\mathbb{N}}$ and 
$\mathscr{K}^{\mathsf{s}} = \langle K_{1,k} \rangle_{k\in\mathbb{N}}$, where $K_{1, 0}$ is defined to be $K_{1}.$
Indeed, a graph has many edges if and only it contains either a large matching as a minor, or a star with many leaves.
Let $\Ccal^{M}$ be the set of forests of maximum degree one and $\Ccal^{S}$ be the set of star forests, i.e., the forests whose connected components contain at most one vertex of degree larger than one, with at most one component on at least two vertices.

It is easy to verify that $\cobs_{\leqslant_{\mathsf{m}}}(\esize) = \{\Ccal^{M}, \Ccal^{S}\}$.
Consequently, $\pobs_{\leqslant_{\mathsf{m}}}(\esize) = \big\{ \{P_{3}\}, \{K_{3}, 2\cdot K_{2}\} \big\}$.

\medskip
So, as we just discussed, $\esize$ is a parameter characterized by a universal obstruction containing two parametric graphs.
It is natural to wonder to which graph parameter each of these two graph sequences correspond.
As a preliminary remark, let us observe that since $\pobs_{\leqslant_{\mathsf{m}}}(\esize) \leqslant^{**} \big\{\{P_{3}\}\big\}$ and $\pobs_{\leqslant_{\mathsf{m}}}(\esize)\leqslant^{**} \big\{\{K_{3},2\cdot K_{2}\}\big\}$, we have that $\p_{\{\mathscr{K}^{\mathsf{s}}\}} \preceq \esize$ and $\p_{\{\mathscr{K}^{\mathsf{m}}\}} \preceq\esize$.

\paragraph{Maximum number of leaves of a tree subgraph.}

Let us first consider the parameter $\p_{\{\mathscr{K}^{\mathsf{s}}\}}$.
As discussed above, we have that $\cobs(\p_{\{\mathscr{K}^{\mathsf{s}}\}})=\{\Ccal^{S}\}$ and $\pobs(\p_{\{\mathscr{K}^{\mathsf{s}}\}})=\big\{\{K_{3},2\cdot K_{2}\}\big\}$. 
Observe now that for every graph $G$, $\p_{\{\mathscr{K}^{\mathsf{s}}\}}(G)$ is the maximum number of leaves of a tree subgraph of $G.$
We denote this graph parameter by $\mathsf{maxstar}(G)$\footnote{If $G$ is connected then $\mathsf{maxstar}(G)$ corresponds to the 
$\maxleaf$ parameter \cite{GareyJ79Computers,FellowsL92Onwell}.} that is for every graph $G \in \gall,$
\begin{align}
\nonumber\mathsf{maxstar}(G) \ \coloneqq \ \max\{ d \in \mathbb{N} \mid \ &T \subseteq G\text{ is a tree and}\\
&d\text{ is the number of leaves of }T\}.\label{maxstar_parameter}
\end{align}

\paragraph{Vertex cover.}

We now turn our attention to $\p_{\{\mathscr{K}^{\mathsf{m}}\}}$, which captures the maximum size of a matching of a graph.
Again, by the above discussion, we have that $\cobs_{\leqslant_{\mathsf{m}}}(\p_{\{\mathscr{K}^{\mathsf{m}}\}}) = \{\Ccal^{M}\}$ and $\pobs_{\leqslant_{\mathsf{m}}}(\p_{\{\mathscr{K}^{\mathsf{m}}\}}) = \big\{\{P_{3}\}\big\}$.
For a graph $G$, let $\vc(G)$ denote the vertex cover of $G$, that is
\begin{align}
\vc(G) \ \coloneqq \ \min\{ |S| \mid S \subseteq V(G)\text{ and }G - S\text{ is edgeless} \}.
\end{align}

It is well known that $\p_{\{\mathscr{K}^{\mathsf{m}}\}}$ and $\vc$ are equivalent parameters.

\paragraph{The two extremes.}

Last but not least, one may wonder which parameters correspond to the two ``extremes'' (up to equivalence) of the quasi-ordering induced by $\preceq$ on the set of all graph parameters.

The first graph parameter\footnote{Formally, we may also consider the graph parameter that is the empty function $\p_{\varnothing} = \varnothing$, i.e., is an empty set of pairs. It follows that $\cobs(\p_{\varnothing}) = \{\emptyset\}$, and $\pobs(\p_{\varnothing}) = \big\{\{K_0\}\big\}$.} we consider is the parameter that associates every graph to infinity which we denote by $\p_{\infty}$ and is defined so that for every non-null graph $G \in \gall,$
\begin{align}
\p_{\infty}(G) \ &\coloneqq \ \infty.\label{infinity_parameter}
\end{align}

It is straightforward to observe that $\p_{\infty}$ corresponds to a maximal element of the quasi-ordering $\preceq$ on the set of all parameters as trivially, for every graph parameter $\p,$ $\p \preceq \p_{\infty}.$
To understand what are the class obstruction and the parametric obstruction of $\mathsf{p}_\infty$, we consider the empty graph\footnote{The \emph{empty graph} $K_{0}$ is the graph whose vertex set is the empty set. We may sometimes refer to this as the empty graph or the empty clique.} $K_0$ and define the minor-parametric graph $\mathscr{K}^{\emptyset} = \langle k\cdot K_{0}\rangle_{k\in\mathbb{N}}$, where for every $k \in \Nbbb$, $k \cdot K_{0} = K_{0}.$
One can observe that, for every non-null graph $G$, $\p_{\{\mathscr{K}^{\emptyset}\}}(G) = \infty$.
It follows that $\{\mathscr{K}^{\emptyset}\}$ is a minor-universal obstruction for $\p_{\infty}$, $\cobs_{\leqslant_{\mathsf{m}}}(\p_{\infty}) = \{\{K_{0}\}\}$, and $\pobs_{\leqslant_{\mathsf{m}}}(\p_{\infty}) = \big\{\{ K_{1} \} \big\}$.

The second graph parameter we consider is the constant zero parameter which we denote by $\p_{0}$ and is defined so that for every graph $G \in \gall,$
\begin{align}
\p_{0}(G) \ &\coloneqq \ 0.\label{zero_parameter}
\end{align}

The parameter $\p_{0}(G)$ defines the other extreme of the relation $\preceq$ as it is a minimal element since for every graph parameter $\p_{0} \preceq \p.$
It is straightforward to see that the empty set $\emptyset$ is a $\leqslant$-universal obstruction of $\p_{0}$ and as a result $\cobs_{\leqslant}(\p_{0}) = \emptyset$ and $\pobs_{\leqslant}(\p_{0}) = \emptyset,$ for any quasi-ordering relation $\leqslant$ on $\gall.$
It is worthwhile to mention that with respect to equivalence of parameters every constant $c$ graph parameter for any non-negative integer $c$ is equivalent to $\p_{0}.$

\medskip
The discussion above on these few simple parameters highlights the versatility of the concepts introduced in \autoref{diidnidialdedfs} and their interrelationships. As we have seen, these concepts provide various perspectives and approaches to defining relevant parameters. In the next subsection, we explore several classic minor-monotone graph parameters, along with some natural yet less studied parameters. This examination will lead to the definition of a rich hierarchy of parameters that, in our view, merits more systematic investigation.

\medskip
We proceed with our presentation of minor-monotone parameters by considering the three arguably most famous ones: treewidth, pathwidth, and treedepth.

\subsection{Treewidth}
\label{treewidth_sec}

Our first guest is treewidth.
A traditional approach is to see treewidth as a measure of the topological resemblance of a graph to a tree and is an omnipresent graph parameter due to its numerous applications both in combinatorics and in algorithms \cite{Bodlaender98}.

\paragraph{Clique sum closure and treewidth.}

To understand what it means to measure the topological resemblance of a graph to a tree, the concept of the \textsl{clique sum} operation comes in handy.
Clique sums play an important role in graph minor theory as they allow to build more complicated graphs in a controlled way by ``gluing'' structurally simpler ones.

\medskip
Formally, the clique sum operation applies to two graphs $G$ and $G'$, each containing a clique of size $c \in \mathbb{N}$, denoted $K$ and $K'$ respectively.
Given a bijection $\sigma \colon V(K)\to V(K')$, the \emph{clique sum} operation takes the disjoint union of $G$ and $G'$; identifies, each vertex $v \in V(K)$ with the vertex $\sigma(v) \in V(K')$ and, in the resulting graph, possibly removes some of the edges between the identified vertices.
Notice that the clique sum operation is a \textsl{one-to-many} operation as the final result depends on the choice of $\sigma$ and on the choice of the edges that are eventually removed.
The \emph{clique sum closure} of a set of graphs $\Gcal$  is the set of all graphs that  can be constructed by a sequence of clique sum operations from the graphs in $\Gcal$.
As we shall see in \autoref{params_smaller_tw}, the clique sum closure is useful for defining more ``complex'' parameters from ``simpler'' ones.

Recall that $\Gcal_{\size, k}$ be the class of graphs on at most $k$ vertices.  Consider now the clique sum closure of  $\Gcal_{\size, k+1}$. In case $k = 1$ one can easily verify that the resulting graph is a forest.
When we allow values of $k \geq 2$, we obtain a version of a ``tree like'' graph which leads to the definition of treewidth.

More precisely, for every graph $G \in \gall$, the treewidth of a graph is defined as follows.
\begin{align}
\tw(G) \ &\coloneqq \ \min\{ k \in \Nbbb \mid G\text{ is in the clique sum closure of } \Gcal_{\size, k+1} \}.\label{treewidth_parameter}
\end{align}

Treewidth is important in algorithmic graph theory 
due to the fact that many problems on graphs can be solved efficiently when restricted to graphs of bounded treewidth.
This fact has been formalized by the celebrated Courcelle's theorem stating that 
every problem on graphs that is expressible in Monadic Second Order Logic (MSOL) can be solved in linear time 
on graphs with bounded treewidth \cite{Courcelle90them, Courcelle97, Courcelle92}.

\medskip
We let $\mathscr{A} = \langle \mathscr{A}_t \rangle_{t\in \mathbb{N}_{\geq 2}}$ denote the minor-parametric graph such that for every $t\geq 2$, $\mathscr{A}_t$ is the $(t \times 4t)$-annulus grid (see \autoref{fig_grid}).
It is easy to verify that both $\mathscr{A}$ and the minor-parametric graph $\Gamma \coloneqq \langle \Gamma_{t} \rangle_{t \in \Nbbb}$, where $\Gamma_{t}$ is the \textsl{$(t \times t)$-grid}\footnote{The $(t \times t)$-grid is the graph obtained as the Cartesian product of two paths on $t$ vertices.}, are equivalent with linear gap with respect to $\approx$.

\begin{figure}[htbp]
\centering
\includegraphics[width=0.9\linewidth]{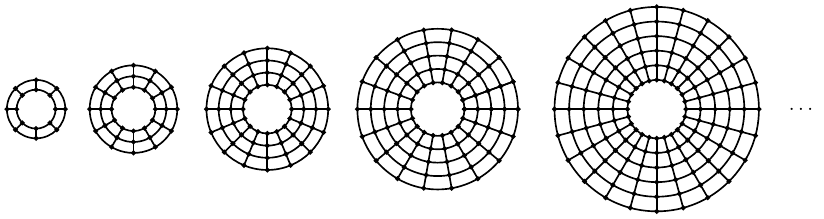}
\caption{The minor-parametric graph $\mathscr{A}=\langle \mathscr{A}_2, \mathscr{A}_3, \mathscr{A}_4, \mathscr{A}_5, \mathscr{A}_6 \ldots \rangle$ of annulus grids.}\label{fig_grid}
\end{figure}

The next proposition sums up the approximate behaviour 
of treewidth in terms of universal obstructions, class obstructions, and parametric obstructions.
Its proof, being the first of its kind in this work, is presented in full detail.

\begin{proposition}\label{prop_tw}
The set $\{\mathscr{A}\}$ consisting of the minor-parametric graph $\mathscr{A}$ of annulus grids is a minor-universal obstruction for $\tw$ with polynomial gap.
Moreover, $\cobs_{\leqslant_{\mathsf{m}}}(\tw) = \{\gplanar\}$ and $\pobs_{\leqslant_{\mathsf{m}}}(\tw) = \big\{ \{K_5, K_{3,3}\} \big\}$.
\end{proposition}
\begin{proof}
We prove that $\tw \sim \p_{\mathscr{A}}$, implying that $\mathscr{A}$ is a minor-universal obstruction for $\tw$.
By definition, if a graph $G$ satisfies $\p_{\mathscr{A}}(G) \geq k$, then $\mathscr{A}_k$, and therefore also $Γ_{k}$, is a minor of $G$.
As $\tw(Γ_k) \geq k$ (see e.g., \cite[Lemma 88]{Bodlaender98}), this implies that $\tw(G) \geq k$ and so $\p_{\mathscr{A}} \preceq \tw$.
The fact that $\tw \preceq \p_{\mathscr{A}}$ follows from the grid theorem~\cite{RobertsonS86GMV} stating that there exists a function $f \colon \mathbb{N} \to \mathbb{N}$, such that for every positive integer $k$, every graph of treewidth at least $f(k)$ contains $Γ_k$ as a minor.
Since $\mathscr{A}_{k} \leqslant_{\mathsf{m}} Γ_{3k}$, it follows that $\tw \preceq \p_{\mathscr{A}}$.
Therefore, $\{ \mathscr{A} \}$ is a universal obstruction for $\tw$.
Moreover, it has been shown in~\cite{ChekuriC16Polynomial,ChuzhoyT21Towards} that the gap function $f$ is polynomial, in particular $f(k) = \Ocal(k^9(\log k)^{O(1)})$.
Therefore, $\tw \sim_{\mathsf{P}} \p_{\mathscr{A}}$.
 
It is easy to observe, see e.g., \cite[(1.5)]{RobertsonST94Quickly}, that grids   are planar and that every planar graph is a minor of some large enough grid. 
As $\mathscr{A}_{k} \leqslant_{\mathsf{m}} Γ_{3k} \leqslant_{\mathsf{m}} \mathscr{A}_{3k}$, it follows that the sequence $\mathscr{A}=\langle \mathscr{A}_k\rangle_{k\in\mathbb{N}_{\geq 3}}$ is 
a minor-omnivore of $\gplanar$.
As $\{\mathscr{A}\}$ is a minor-universal obstruction for $\tw$, \autoref{cobs_uobs} implies that $\cobs_{\leqslant_{\mathsf{m}}}(\tw)=\{\gplanar\}$.
Finally, as $\{K_5,K_{3,3}\}$ is the minor-obstruction set of planar graphs~\cite{Wagner37Uber}, we have $\pobs_{\leqslant_{\mathsf{m}}}(\tw) = \big\{\{K_5,K_{3,3}\}\big\}$.
\end{proof}

An interesting question is to what point the $\Ocal(k^9(\log k)^{\Ocal(1)})$ parametric gap, implied by the result of Chuzhoy and Tan in \cite{ChuzhoyT21Towards}, can be improved.
By considering expander graphs one may easily prove that this gap cannot become better than $\Ocal(k^2 \log k)$ (see e.g., \cite{RobertsonST94Quickly, Thilikos12GraphMinors, DemaineHK09AlgorithmicGraph}).

\medskip
Concluding the presentation on treewidth, we should mention an interesting generalization of treewidth that was given in \cite{geelen2016generalization}.
Given a non-negative integer $t$, we define the \emph{$t$-treewidth} of a graph by using the clique sum definition of treewidth with the additional constraint that all clique sums concern cliques on at most $t$ vertices.
The main result of \cite{geelen2016generalization} can be reformulated as a minor-universal obstruction characterization for $t$-treewidth, for every possible value of $t$.

\subsection{Pathwidth}
\label{pathwidth_def}

In the same manner that treewidth is a measure of the topological resemblance of a graph to a tree, pathwidth serves as a measure of its topological resemblance to a path.
As in the case of treewidth, pathwidth has several equivalent definitions of different flavour.
The classic definition is the one given by Robertson and Seymour in \cite{RobertsonS83GMI} that uses \textsl{path decompositions}, a special case of the tree decomposition, where the underlying tree of the decomposition is in fact a path.
Formally, a \emph{path decomposition} of a graph $G$ is a sequence $\Pcal= \{X_{1}, \ldots, X_{q} \}$ of subsets of $V(G)$ where $\cupall \{G[X_{1}],\ldots,G[X_{q}]\} = G$ and such that for every $1 \leq i  \leq j \leq h \leq q$, it holds that $X_{i} \cap X_{h} \subseteq X_{j}$.
The \emph{width} of $\Pcal$ is the maximum size of a set in $\Pcal$ minus one.
The \emph{pathwidth} of $G$, denoted by $\pw(G)$, is defined as follows.
\begin{align}
\pw(G) \ &\coloneqq \ \min\{ k \in \Nbbb \mid G\text{ has a path decomposition of width at most }k \}.
\end{align}

Alternatively, the pathwidth of a graph $G$ can also be defined as the minimum $k$ for which there exists an ordering $v_1, \ldots, v_n$ of $V(G)$ such that for each $i \in [n]$ there are at most $k$ vertices in $\{ v_{1}, \ldots, v_{i-1} \}$ that are adjacent with vertices in $\{v_{i}, \ldots, v_{n}\}$ (see also \cite{Kinnersley92TheVertex, Moehring90grap, KirousisP85inte}).

We define $\mathscr{T} = \langle \mathscr{T}_t \rangle_{t \in \mathbb{N}_{\geq 1}}$ as the minor-parametric graph of complete ternary trees of depth $k$ illustrated in \autoref{fig_ternary_trees}.
We also denote by $\gforest$ the set of all acyclic graphs.

\begin{figure}[htbp]
\centering
\includegraphics[width=0.85\linewidth]{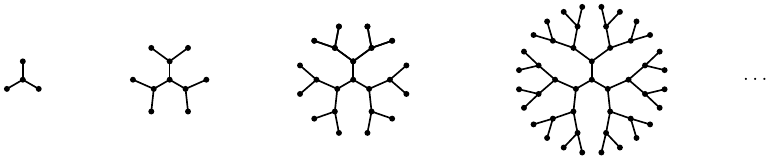}
\caption{\label{fig_ternary_trees} The minor-parametric graph $\mathscr{T} = \langle \mathscr{T}_1, \mathscr{T}_2, \mathscr{T}_3, \mathscr{T}_4, \ldots \rangle$ of complete ternary trees.}
\end{figure}

\begin{proposition}
The set $\mathfrak{T}\coloneqq \{ \mathscr{T} \}$ consisting of the minor-parametric graph $\mathscr{T}$ of complete ternary trees is a minor-universal obstruction for $\pw$.
Moreover, $\cobs_{\leqslant_{\mathsf{m}}}(\pw) = \{\gforest\}$ and $\pobs_{\leqslant_{\mathsf{m}}}(\pw) = \big\{\{K_3\}\big\}$.
\end{proposition}
\begin{proof}
It is easy to see that $\pw(\mathscr{T}_k) = Ω(k)$ (see e.g., \cite[Proposition 3.2]{KirousisP86Searching}), therefore $\p_{\{\mathscr{T}\}} \preceq \pw$.
The direction $\pw \preceq \p_{\{\mathscr{T}\}}$ follows from the fact that, for every forest $F$ every graph in $\excl_{\leqslant_{\mathsf{m}}}(F)$ has pathwidth at most $|F| - 2$, as proven in \cite{BienstockRST91Quickly} (see also \cite{RobertsonS83GMI}).
As $\mathscr{T}$ is a minor-omnivore of $\gforest$ and $\obs_{\leqslant_{\mathsf{m}}}(\gforest) = \{K_{3}\}$,  from \autoref{cobs_uobs}, we obtain that $\cobs_{\leqslant_{\mathsf{m}}}(\pw) = \{\gforest\}$ and $\pobs_{\leqslant_{\mathsf{m}}}(\pw) = \big\{ \{K_3\} \big\}$.
\end{proof}

\subsection{Treedepth}
\label{treedepth_par}

A third graph parameter of arguable significance that has received considerable attention is \textsl{treedepth}.
Given a graph $G,$ let $\mathsf{tcl}(G)$ denote the transitive closure of $G.$
The \emph{treedepth} of a graph $G$, denoted $\td(G)$, is defined as follows \cite{NesetrilO06Treedepth}.
\begin{align}
\td(G) \ &\coloneqq \ \min\{ d \in \Nbbb \mid F\text{ is a rooted forest of depth }d\text{ and }G \subseteq \mathsf{tcl}(F)\}.\label{treedepth_parameter}
\end{align}

As it is the case for both treewidth and pathwidth, there are several equivalent definitions of treedepth.
An alternative definition is using vertex colorings.
The treedepth of a graph is the minimum $k$ for which there is a proper coloring  $χ \coloneqq V(G) \to [k]$ of $G$ such that every path whose endpoints have the same color contains a vertex of larger color.

A \emph{linear forest} is a graph whose connected components are paths.
We use $\glforest$ in order to denote the set of all linear forests.
Let $\mathscr{P} = \langle \mathscr{P}_t \rangle_{t \in \mathbb{N}_{1}}$ denote the minor-parametric graph of paths on $t$ vertices (see \autoref{fig_paths}).

\begin{figure}[htbp]
\centering
\includegraphics[width=0.6\linewidth]{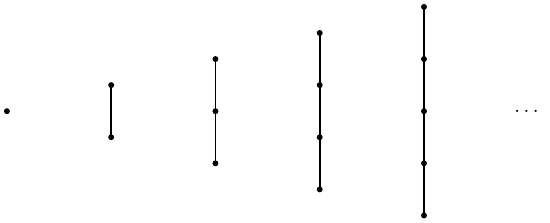}
\caption{\label{fig_paths} The minor-parametric graph $\mathscr{P} = \langle \mathscr{P}_1, \mathscr{P}_2, \mathscr{P}_3, \mathscr{P}_4, \mathscr{P}_5, \ldots \rangle$ of paths.}
\end{figure}

\begin{theorem}\label{prop_td}
The set $\{\mathscr{P}\}$ consisting of the minor-parametric graph $\mathscr{P}$ of paths is a minor-universal obstruction for $\td$.
Moreover, $\cobs_{\leqslant_{\mathsf{m}}}(\td) = \{\glforest\}$ and $\pobs_{\leqslant_{\mathsf{m}}}(\td) = \big\{ \{K_3, K_{1,3}\} \big\}$.
\end{theorem}
\begin{proof}
Notice that $\td(P_{k}) \leq \td(P_{\lceil \nicefrac{k}{2} \rceil}) + 1$, for every $k \in \Nbbb$.
This immediately implies that $\p_{\mathscr{P}} \preceq \td$.
To see that $\td \preceq \p_{\mathscr{P}}$ assume that $G$ does not have any path of length $> k$ and observe that any BFS-tree in $G$ certifies that $\td(G) \leq k$.
We conclude that $\td \sim \p_{\mathscr{P}}$.
As $\mathscr{P}$ is a minor-omnivore of $\glforest$, we have that $\cobs_{\leqslant_{\mathsf{m}}}(\td) = \{\glforest\}$ and it is easy to see that $\obs_{\leqslant_{\mathsf{m}}}(\glforest) = \{K_3, K_{1,3} \}$.
\end{proof}

We stress that the statement of \autoref{prop_td} can be rephrased in the context of the subgraph relation $\leqslant_{\mathsf{sg}}$ instead of the minor one. 
The difference is that $\pobs_{\leqslant_\mathsf{sg}}(\td) = \big\{ \{ K_{1,3} \} \cup \{ C_{k} \mid k \geq 3 \} \big\}$, where $C_{k}$ is the cycle on $k$ vertices, for every $k \in \Nbbb_{\geq 3}.$
The reason for this is that containing a path as a minor is the same as containing it as a subgraph.
Observe that the subgraph-parametric obstruction of $\td$ is a finite set containing only one (infinite) subgraph-antichain.
In general, as we shall discuss in \autoref{obsjjdjdother_a}, for parameters that are monotone under relations which are not well-quasi-ordered (such as the subgraph relation), class obstructions may not be finite or even exist.

\medskip
It is also interesting to note that the proof above does not provide a polynomial gap function as $\td(P_{k}) = \Theta(\log k)$.
Actually, this is unavoidable when considering any minor-universal obstruction consisting of a single minor-parametric graph $\mathscr{H} = \langle \mathscr{H}_{k} \rangle_{k\in\mathbb{N}}$, because $\mathscr{H}_{k}$, being a minor-omnivore of $\glforest$, has to consist of linear forests.
Interestingly, were we to permit universal obstructions to not necessarily be $\lesssim$-antichains, then it may be possible to attain improved gaps.
Indeed, as shown by the following theorem, a polynomial gap for treedepth can be achieved by considering a collection of three minor-parametric graphs $\mathfrak{P}$ such that $\p_\mathfrak{P} \sim_{\mathsf{P}} \td$ as follows.

\begin{theorem}[\!\cite{KawarabayashiRossman2022Polynomial}]\label{mystery_theorem}
There exists a constant $c$ such that every graph of treedepth at least $k^c$ contains at least one of the following graphs as a minor.
\begin{itemize}
\item The $(k\times k)$-grid;
\item the complete ternary tree of height $k$;
\item the path of order $2^k$.
\end{itemize}
\end{theorem}

\subsubsection{Parameter hierarchy to the ``right'' of treewidth}

At this point we would like to discuss the approximate relationship of the parameters we have introduced so far.
We have that
\begin{align}
\tw \ \preceq \ \pw \ \preceq \ \td \ \preceq \ \vc \ \preceq \ \esize \ \preceq \ \size \ \preceq \ \mathsf{p}_\infty.\label{right_part}
\end{align}

Reflecting on the interrelationship between the minor-parametric obstructions of the previously introduced parameters, the former hierarchy translates to the following hierarchy via the equivalence between \eqref{first} and \eqref{third}.
\begin{align}
\nonumber\big\{\{K_{1}\}\big\} \ &\leqslant_{\mathsf{m}}^{**} \ \big\{\{K_{2}\}\big\} \ \leqslant_{\mathsf{m}}^{**} \ \big\{\{P_{3}\},\{K_{3},2\cdot K_{2}\}\big\} \ \leqslant_{\mathsf{m}}^{**} \ \big\{\{P_{3}\}\big\}\\
&\leqslant_{\mathsf{m}}^{**} \ \big\{\{K_3,K_{1,3}\}\big\} \ \leqslant_{\mathsf{m}}^{**} \ \big\{\{K_{3}\}\big\} \ \leqslant_{\mathsf{m}}^{**} \ \big\{ \{K_{5}, K_{3,3}\} \big\}.
\end{align}

\subsection{Parameters smaller than treewidth}\label{params_smaller_tw}

A natural question to ask is what resides on ``the left'' of this hierarchy, i.e., what kind of minor-monotone parameters are (approximately) smaller than treewidth.
With this question we enter the realm of Graph Minors and specifically of the Graph Minors Structure Theorem of Robertson and Seymour \cite{RobertsonS03GMXVI}.

\paragraph{Graph parameters and their clique sum closure.}

The clique sum operation is particularly useful for defining new graph parameters.
Given a graph parameter $\p$ we define its \emph{clique sum closure $\p^{*}$} so that for every graph $G,$
\begin{align}
\p^{*}(G) \ &\coloneqq \ \min\{ k \in \Nbbb \mid G\text{ is in  the clique sum closure of }\Gcal_{\p,k}\}.
\end{align}

By the clique sum closure definition of treewidth \eqref{treewidth_parameter}, one may easily observe that $\tw(G) = \size^{*}(G) + 1,$ for every graph $G.$

\subsubsection{Hadwiger number}

The \emph{Hadwiger number} of a graph $G$, denoted $\hw(G)$, is defined as follows \cite{Hadwiger43Uber}.
\begin{align}
\hw(G) \ &\coloneqq \ \max\{ k \in \Nbbb \mid K_{k} \leqslant_{\mathsf{m}} G\}.
\end{align}

Its definition immediately] implies that $\hw$ is the parameter associated to the minor-parametric graph $\langle K_{t} \rangle_{t\in\mathbb{N}}$.
We may alternatively replace cliques by an equivalent minor-parametric graph that is obtained by enhancing the $(t \times 4t)$-annulus grid as follows.
Let $\mathscr{K} = \langle \mathscr{K}_t \rangle_{t\in\mathbb{N}}$ be the minor-parametric graph \emph{clique grids} as illustrated in \autoref{fig_hadwiger_grid}.
It is not difficult to verify that $\mathscr{K} \approx \langle K_{t}\rangle_{t\in\mathbb{N}}$ with linear gap and therefore $\hw \sim \p_{\{\mathscr{K}\}}$.

\begin{figure}[htbp]
\centering
\includegraphics[width=0.9\linewidth]{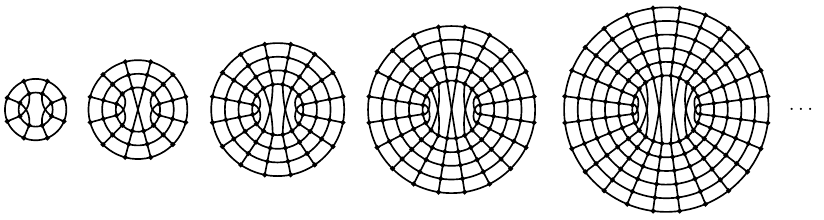}
\caption{\label{fig_hadwiger_grid}The minor-parametric graph $\mathscr{K} = \langle \mathscr{K}_2, \mathscr{K}_3, \mathscr{K}_4, \mathscr{K}_5, \mathscr{K}_6 \ldots \rangle$ of clique grids.}
\end{figure}

A min-max analogue of the Hadwiger number is given by the Graph Minors Structure Theorem (GMST), proven by Robertson and Seymour in \cite{RobertsonS03GMXVI}.
Much like treewidth, the GMST provides insight on the structure of graphs excluding some fixed graph $H$ as a minor.
There are four components to its definition.
The first, as in treewidth, is the clique sum operation that allows us to build more complicated graphs from simpler ones.
However, the notion of ``structurally simple'' now differs significantly.
There are three components to this: 1) \textsl{apices}, 2) \textsl{surfaces}, and 3) \textsl{vortices}.
Essentially, a graph in this setting is considered as simple if after removing a bounded number of ``apex'' vertices, it is embeddable in a surface of bounded Euler genus with the addition of a bounded number of bounded ``depth'' vortices, where ``bounded'' here means a function in the size of the excluded minor $H.$
We do not wish to present a formal definition of the GMST here, therefore we choose to give a more compact statement, proved in \cite{thilikos2023excluding}.

Let $\Ecal_{k}$ denote the class of graphs embeddable in a surface of Euler genus at most $k,$ for every non-negative integer $k.$
Moreover, given a graph $H$ and a set $X \subseteq V(G),$ we say that $H$ is an \emph{$X$-minor} of $G$ if there exists a collection $\mathcal{S} = \{S_{v} \mid v \in V(H))\}$ of pairwise vertex-disjoint connected\footnote{A set $X \subseteq V(G)$ is \emph{connected} in $G$ if the induced subgraph $G[X]$ is a connected graph.} subsets of $V(G)$ each containing at least one vertex of $X$ and such that, for every edge $xy \in E(H),$ the set $S_{x} \cup S_{y}$ is connected in $G.$
An \emph{annotated graph} is a pair $(G,X)$ where $X\subseteq V(G)$,
Given an annotated graph $(G,X),$ we define $\tw(G, X)$ as the maximum treewidth of an $X$-minor of $G.$
We define an extension of treewidth, denoted by $\tw^{*}_{\Ecal}$ as follows.
\begin{align}
\tw^{*}_{\Ecal}(G) \ \coloneqq \ \big( \min\{ k \in \Nbbb \mid \exists X \subseteq V(G) : \tw(G,X) \leq k\text{ and }G - X \in \Ecal_{k}\} \big)^{*}.\label{gmst_annotated_version}
\end{align}

A streamlined way to restate the GMST is the following.

\begin{proposition}[\!\cite{thilikos2023excluding}]\label{@mutilation}
There exists a function $f \colon \Nbbb\to \Nbbb$ such that
\begin{align*}
\tw^{*}_{\Ecal} \ \preceq \ \hw
\end{align*}
with gap function $f.$
\end{proposition}

In fact, $\hw \sim \tw^{*}_{\Ecal}$ (see e.g. Section 1.3, Theorem 1.5 \cite{kawarabayashi2020quickly}).
To conclude, as every graph is a minor of a large enough clique, the following statement holds.

\begin{proposition} 
The set $\{\mathscr{K}\}$ consisting of the minor-parametric graph $\mathscr{K}$ of clique grids is a minor-universal obstruction for $\hw$.
Moreover, $\cobs_{\leqslant_{\mathsf{m}}}(\hw) = \{ \gall \}$ and $\pobs_{\leqslant_{\mathsf{m}}}(\hw) = \{ \emptyset \}$.
\end{proposition}

\subsubsection{Singly-crossing treewidth}

We commence with the definition of an auxiliary parameter $\psize$ as follows.
\begin{align*}
\psize(G) \ \coloneqq \ \begin{cases} 0, &G\text{ is planar}\\
                                     |G|, &\text{otherwise}
                        \end{cases}
\end{align*}

The \emph{singly-crossing treewidth} which we denote by $\sctw$, is defined as follows.
\begin{align}
\sctw \ &\coloneqq \ \psize^{*}\label{sctw_parameter}
\end{align}

Singly-crossing treewidth has been considered in algorithmic applications where it is possible to combine the structure of a planar graph with the tree decomposition formed  by the clique sums of the bounded-size non-planar parts.
In \cite{Kaminski12Maxcut} for instance, Kamiński showed that there exists a function $f \colon \nton$ such that \textsc{Max-Cut} can be solved in time $\Ocal(|G|^{f(\sctw(G))})$, which gives a polynomial time algorithm for every graph class where $\sctw(G)$ is bounded.
For other algorithmic applications of $\sctw$, see \cite{DemaineHT05exponential, DemaineHNRT04approximation}.

\begin{figure}[htbp]
\centering
\includegraphics[width=0.9\linewidth]{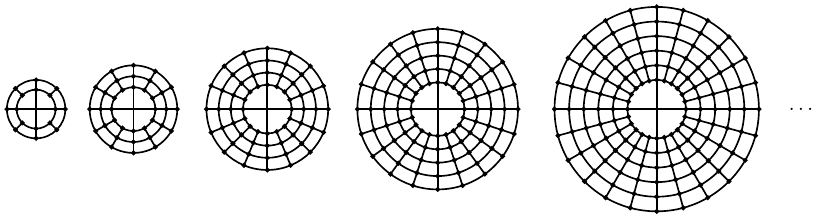}
\caption{\label{fig_singly_crossing}The minor-parametric graph $\mathscr{S} = \langle \mathscr{S}_2, \mathscr{S}_3, \mathscr{S}_4, \mathscr{S}_5, \mathscr{S}_6 \ldots \rangle$ of singly-crossing grids.}
\end{figure}

We define the minor-parametric graph $\mathscr{S} = \langle \mathscr{S}_k \rangle_{k\in\mathbb{N}_{\geq 2}}$ where $\mathscr{S}_k$ is the  \emph{singly-crossing grid} depicted in  \autoref{fig_singly_crossing}.
We also define $\scrossing$ to be the class of \emph{singly-crossing graphs} that are graphs that can be drawn in the sphere so that at most one edge $e$ intersects at most once the drawing of $G - e$.
In \cite{RobertsonS93Excluding}, Robertson and Seymour proved that there exists a function $f \colon \mathbb{N}\to\mathbb{N}$ such that if a graph excludes a singly-crossing graph $H$ on $k$ vertices as a minor, then it is the clique sum closure of set of graphs that are either planar or have  treewidth at most $f(k)$, or equivalently, the clique sum closure of the union of planar graphs and of  graphs of size at most $f(k)+1$.
Also, note that $\mathscr{S}$ is a minor-omnivore of singly-crossing graphs.
Moreover, according to Robertson and Seymour \cite{RobertsonS93Excluding}, $\obs_{\leqslant_{\mathsf{m}}}(\scrossing) = \Ocal_{\mathsf{projective}} \cup \Ocal_{\mathsf{linkless}}$
where $\Ocal_{\mathsf{projective}}$ is the minor-obstruction set of the class of graphs that are embeddable in the projective plane and $\Ocal_{\mathsf{linkless}}$ is the minor-obstruction set of the class of \textsl{linklessly embeddable}\footnote{A graph is \emph{linklessly embeddable} if it has an embedding in the 3-dimensional space so that no two cycles of this embedding are linked \cite{RobertsonST95sachs}.} graphs.
According to \cite{ArchdeaconH89akuratowski, Archdeacon81AKuratowski, GloverHW79graphsthat}, $\Ocal_{\mathsf{projective}}$ contains 35 graphs and, according to \cite{RobertsonST95sachs}, $\Ocal_{\mathsf{linkless}}$ contains the $7$ graphs of the \textit{Petersen family} {that are} the graphs that can be obtained from $K_{6}$ by applying combinations of $\Delta$-$Y$ or $Y$-$\Delta$ transformations  (see \autoref{petrpet}).
As $\Ocal_{\mathsf{projective}} \cap \Ocal_{\mathsf{linkless}} = \{K_{4,4}^-\}$, where $K_{4,4}^-$ is $K_{4,4}$ minus an edge, it follows that $\Ocal_{\mathsf{projective}} \cup \Ocal_{\mathsf{linkless}}$ contains 41 graphs.
We summarize the above facts to the following statement.

\begin{figure}[htbp]
\centering
\includegraphics[width=0.35\linewidth]{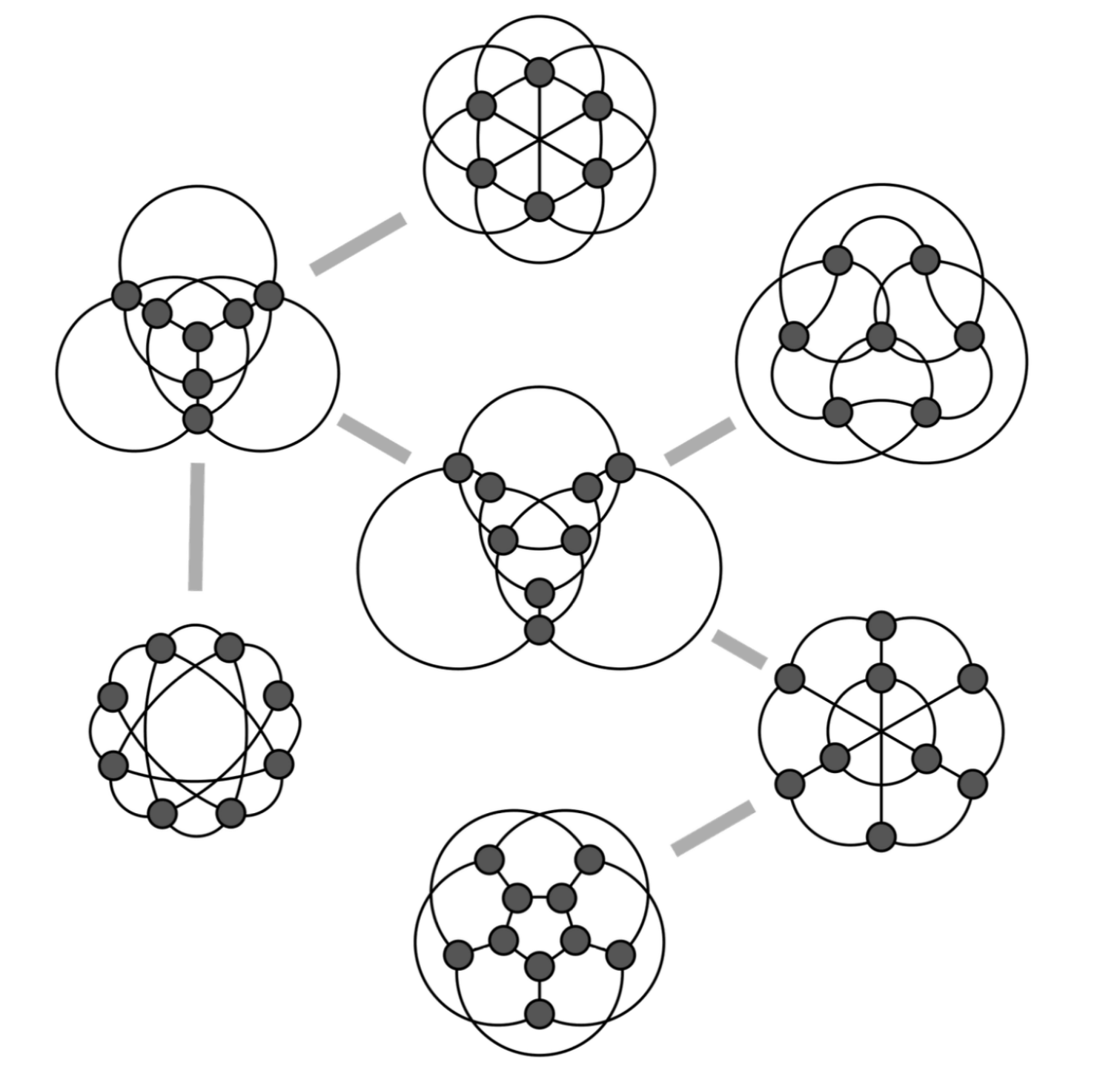}
\caption{\label{petrpet}The Petersen family (image taken from \href{https://en.wikipedia.org/wiki/Petersen_graph}{\sf Wikipedia}).}
\end{figure}

\begin{proposition}
The set $\{\mathscr{S}\}$ consisting of the minor-parametric graph $\mathscr{S}$ of singly-crossing grids is a minor-universal obstruction for $\sctw$.
Moreover, $\cobs_{\leqslant_{\mathsf{m}}}(\sctw) = \{\scrossing\}$ and $\pobs_{\leqslant_{\mathsf{m}}}(\sctw) = \big\{ \Ocal_{\mathsf{projective}} \cup \Ocal_{\mathsf{linkless}} \big\}$.
\end{proposition}

We remark that $\sctw$ sits at the bottom of the ``$\mathsf{vga}$-lattice''\footnote{The term ``$\mathsf{vga}$-lattice'' in \cite{ThilikosW22Killing}. ``$\mathsf{v}$'' stands for vortex, ``$\mathsf{g}$'' stands for genus, and ``$\mathsf{a}$'' stands for apices. The $\mathsf{vga}$-lattice is used to define refined versions of the GMST by considering as building blocks all possible subsets of the three components used in the GMST: vortices, genus, and apices.}  as it corresponds to a version of the GMST where neither apices nor vortices appear, and where the genus of the surface in which the graph embeds is zero.

\subsubsection{Shallow-vortex treewidth}

The next graph parameter we introduce originates from the results in \cite{ThilikosW22Killing}.

The \emph{shallow-vortex treewidth} of a graph $G$, which we denote by $\svtw(G)$ is defined as the minimum non-negative integer $k$ for which $G$ is the clique sum closure of the set of all graphs containing a set of at most $k$ vertices, whose removal yields a graph in $\Ecal_{k}$.
In other words,
\begin{align}
\svtw \ &\coloneqq \ \big( \min\big\{k\in\mathbb{N}\mid\text{$\exists X \subseteq V(G) : |X|\leq k\text{~and~}G-X\in \Ecal_{k}$}\big\} \big)^*.\label{svminor}
\end{align}

Note that the definition above adopts the \textsl{modulator vs target} scheme of graph modification problems.
The set $X$ is the modulator that, when removed, we obtain a graph in the target class $\Ecal_{k}$.
Moreover, the (annotated) parameter $\size$ acts as the measure on the modulator, that in \eqref{svminor} is $|X|$.

The parameter $\svtw$ has been introduced in \cite{ThilikosW22Killing} in the context of the study of the problem \textsc{\#Perfect Matching} asking for the number of perfect matchings in a graph.
The main result of \cite{ThilikosW22Killing} is that there exists a function $f \colon \nton$ such that \textsc{\#Perfect Matching} can be solved in time $\Ocal(|G|^{f(\sctw(G))})$.
This makes \textsc{\#Perfect Matching} polynomial time solvable on every minor-closed graph class where $\sctw$ is bounded.
Moreover, it was proved in \cite{ThilikosW22Killing} that for minor-closed graph classes where $\sctw$ is unbounded, \textsc{\#Perfect Matching} is \#\textsf{P}-complete.
This means that $\sctw$ precisely captures the transition between polynomially solvable and hard with respect to the counting complexity of \textsc{\#Perfect Matching}.

\begin{figure}[htbp]
\centering
\includegraphics[width=0.9\linewidth]{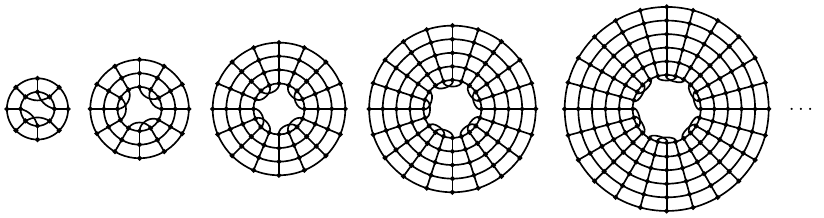}
\caption{\label{svds_fig}The minor-parametric graph of shallow-vortex grids $\mathscr{V} = \langle \mathscr{V}_2, \mathscr{V}_3, \mathscr{V}_4, \mathscr{V}_5, \mathscr{V}_6 \ldots \rangle$.}
\end{figure}

We define the minor-parametric graph $\mathscr{V} \coloneqq \langle \mathscr{V}_{t} \rangle_{t \in \Nbbb_{\geq 2}}$ of shallow-vortex grids, where for every $t \geq 2,$ $\mathscr{V}_{t}$ is the graph obtained from a $(t \times 4t)$-annulus grid by adding a matching $M$ of size $2t$ in its inner cycle, as illustrated in \autoref{svds_fig}.
The class of \emph{shallow-vortex minors} is defined as $\shallowvortexminors = \closure{\leqslant_{\mathsf{m}}}{\mathscr{V}}.$

Using our terminology, in \cite{ThilikosW22Killing}, the following is proved.

\begin{proposition}[\!\cite{ThilikosW22Killing}]
The set $\{\mathscr{V}\}$ consisting of the minor-parametric graph $\mathscr{V}$ of shallow-vortex grids is a minor-universal obstruction for $\svtw$.
Moreover, $\cobs_{\leqslant_{\mathsf{m}}}(\sctw) = \{ \shallowvortexminors \}$ and $\pobs_{\leqslant_{\mathsf{m}}}(\sctw) = \big\{ \obs_{\leqslant_{\mathsf{m}}}(\shallowvortexminors) \big\}$.
\end{proposition}

We should remark that the minor-obstruction set of the class of shallow-vortex minors is unknown, and most likely not easy to identify.
In addition, with respect to the $\mathsf{vga}$-lattice, $\svtw$ corresponds to the version of the GMST without vortices.

\subsubsection{Excluding surfaces as minors}

In \cite{thilikos2023excluding}, a refinement of the GMST was proved, that gives a structural characterization depending on the Euler genus of the excluded minor.
This leads to a definition of an extension of treewidth, for every non-negative integer $g,$ denoted by $\tw^{*}_{\Ecal_{g}},$ that ``fixes'' the Euler genus of the surface in \eqref{gmst_annotated_version}.
Formally, for every graph $G,$
\begin{align}
\tw^{*}_{\Ecal_{g}}(G) \ \coloneqq \ \big( \min\{ k \in \Nbbb \mid \exists X \subseteq V(G) : \tw(G,X) \leq k\text{ and }G - X \in \Ecal_{g}\} \big)^{*}.\label{surfex_treewidth}
\end{align}

For every non-negative integer $g \in \Nbbb,$ consider the set $\Sbbb_{g}$ consisting of the surfaces of minimal Euler genus larger than $g.$
By the surface classification theorem \cite{MoharT01Graphs}, one may observe that there are at most two such surfaces, one orientable and one non-orientable (see \cite{thilikos2023excluding}).
To each surface $\Sigma \in \Sbbb_{g}$ we may associate a ``wall-like'' minor-parametric graph $\mathscr{D}^{\Sigma} = \langle \mathscr{D}^{\Sigma}_{t} \rangle_{t \in \Nbbb_{\geq 2}},$ that precisely captures $\Sigma,$ i.e., such that the minor-closure of $\mathscr{D}^{\Sigma}$ corresponds to the class of graphs embeddable in $\Sigma.$
To define $\mathscr{D}^{\Sigma}$ we require three ingredients.
The minor-parametric graph $\mathscr{A}$ of annulus grids, the minor-parametric graph $\mathscr{D}^{(1, 0)} = \lin{\mathscr{D}^{(1, 0)}_{t}}_{t \geq 2},$ where $\mathscr{D}^{(1, 0)}_{t}$ is obtained from the $(t \times 4t)$-annulus grid by adding a linkage of order $t$ joining vertices of its inner cycle as illustrated in \autoref{fig_torus_grid}, and the minor-parametric graph $\mathscr{D}^{(0, 1)} = \lin{\mathscr{D}^{(0, 1)}_{t}}_{t \geq 2},$ where $\mathscr{D}^{(0, 1)}_{t}$ is once again obtained from the $(t \times 4t)$-annulus grid by adding a linkage of order $2t$ joining vertices of its inner cycle as illustrated in \autoref{fig_projective_grid}.

\begin{figure}[htbp]
\centering
\includegraphics[width=0.9\linewidth]{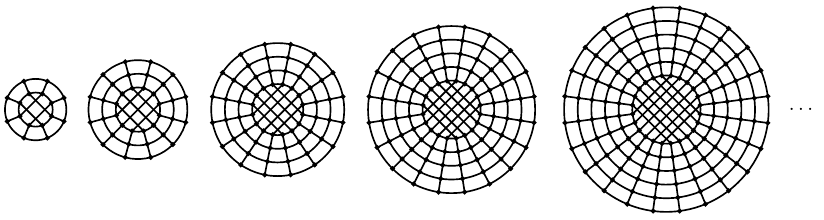}
\caption{\label{fig_torus_grid}The minor-parametric graph $\mathscr{D}^{(1, 0)} = \langle \mathscr{D}^{(1, 0)}_2, \mathscr{D}^{(1, 0)}_3, \mathscr{D}^{(1, 0)}_4, \mathscr{D}^{(1, 0)}_5, \mathscr{D}^{(1, 0)}_6 \ldots \rangle$ of torus grids.}
\end{figure}

\begin{figure}[htbp]
\centering
\includegraphics[width=0.9\linewidth]{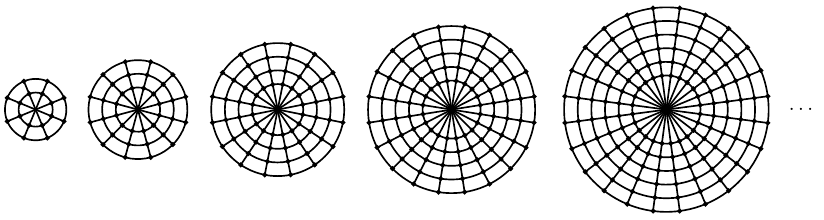}
\caption{\label{fig_projective_grid}The minor-parametric graph $\mathscr{D}^{(0, 1)} = \langle \mathscr{D}^{(0, 1)}_2, \mathscr{D}^{(0, 1)}_3, \mathscr{D}^{(0, 1)}_4, \mathscr{D}^{(0, 1)}_5, \mathscr{D}^{(0, 1)}_6 \ldots \rangle$ of projective grids.}
\end{figure}

It can be seen \cite{thilikos2023excluding} that $\mathscr{D}^{(1, 0)} = \mathscr{D}^{\Sigma^{(1, 0)}}$ captures the torus $\Sigma^{(1, 0)}$, i.e., the orientable surface homeomorphic to the sphere with the addition of a handle, while $\mathscr{D}^{(1, 0)} = \mathscr{D}^{\Sigma^{(0, 1)}}$ captures the projective plane $\Sigma^{(0, 1)}$, i.e., the non-orientable surface homeomorphic to the sphere with the addition of a crosscap \cite{MoharT01Graphs}.
Now, we define $\mathscr{D}^{\Sigma}_{t},$ $t \geq 2,$ for a surface $\Sigma$ homeomorphic to the sphere with the addition of $h \in \Nbbb$ many handles and $c \in [0, 2]$\footnote{By Dyck's Theorem \cite{Dyck1888Beitrage, Francis99ConwayZIP} we may assume that $c \in [0,2]$ since any two crosscaps are equivalent (with respect to homeomorphisms) to a handle in the presence of a (third) crosscap.} many crosscaps,  by joining together one copy of $\mathscr{A}_{t},$ $h$ copies of $\mathscr{D}^{(1,0)}_{t},$ and $c$ copies of $\mathscr{D}^{(0,1)}_{t},$ in the cyclic order $\mathscr{A}_{t},\mathscr{D}^{(1,0)}_{t}, \ldots, \mathscr{D}^{(1,0)}_{t}, \mathscr{D}^{(0,1)}_{t},\ldots, \mathscr{D}^{(0,1)}_{t},$ as indicated in \autoref{@freemasonry}.

\begin{figure}[htbp]
\centering
\includegraphics[width=0.95\linewidth]{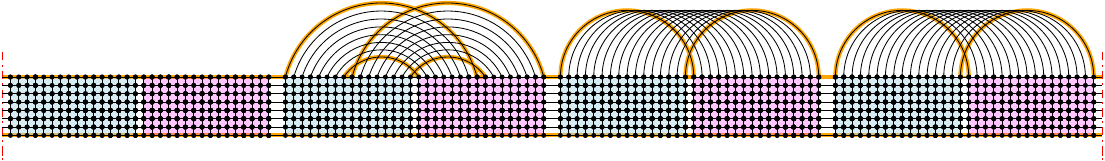}
\caption{\label{@freemasonry}The instance $\mathscr{D}_{8}^{(1,2)}$ of the minor-parametric graph $\mathscr{D}^{\Sigma^{(1,2)}}$.}
\end{figure}

The main result of \cite{thilikos2023excluding} is a minor-universal obstruction for $\tw^{*}_{\Ecal_{g}}$ which is stated in the following proposition.

\begin{proposition}\label{surfex_univ_obs}
For every non-negative integer $g,$ the set $\{ \mathscr{D}^{\Sigma} \mid \Sigma \in \Sbbb_{g} \}$ is a minor-universal obstruction for $\tw^{*}_{\Ecal_{g}}.$
\end{proposition}

With respect to the $\mathsf{vga}$-lattice, $\tw^{*}_{\Ecal_{g}}$ corresponds to a parameterization of the GMST, for each possible value of the genus component.

\subsubsection{From small treewidth to small size modulators}

A consequence of the results in \cite{PaulPTW24Delineating} is a minor-universal obstruction for a special case of $\tw^{*}_{\Ecal_{g}}$, that transforms the small annotated treewidth modulators in \eqref{surfex_treewidth} to modulators of small size.
This result is obtained by combining the results in \cite{thilikos2023excluding} and \cite{ThilikosW22Killing}.
Towards this, for every non-negative integer $g,$ we define the parameter $\apex^{*}_{\Ecal_{g}},$ where for every graph $G,$
\begin{align}
\apex^{*}_{\Ecal_{g}} \ \coloneqq \ \big( \min\{ k \in \Nbbb \mid \exists X \subseteq V(G) : |X| \leq k\text{ and }G - X \in \Ecal_{g}\} \big)^{*}.\label{killing_surfex_treewidth}
\end{align}

They prove that by combining the minor-parametric graphs in the set $\{ \mathscr{D}^{\Sigma} \mid \Sigma \in \Sbbb_{g} \}$ with the minor-parametric graph $\mathscr{V}$ of shallow-vortex grids we obtain a minor-universal obstruction for $\apex^{*}_{\Ecal_{g}}.$

\begin{proposition}\label{killing_surfex_univ_obs}
The set $\{ \mathscr{D}^{\Sigma} \mid \Sigma \in \Sbbb_{g} \} \cup \{ \mathscr{V} \}$ is a minor-universal obstruction for $\apex^{*}_{\Ecal_{g}}.$
\end{proposition}

We remark that in the $\mathsf{vga}$-lattice, $\apex^{*}_{\Ecal_{g}}$ corresponds to a parameterization of the GMST for every possible value of the genus, where additionally the vortex component is completely eliminated.

\subsubsection{Two particular special cases of $\tw^{*}_{\Ecal_{g}}$ and $\apex^{*}_{\Ecal_{g}}$}

Two rather simple and interesting instantiations of $\tw^{*}_{\Ecal_{g}}$ and $\apex^{*}_{\Ecal_{g}}$ are in the setting of planarity.

\medskip
Recall that we use $\gplanar$ to denote the class of all planar graphs.
Planar graphs are the graphs embeddable in the sphere which is the single surface of Euler genus zero.
Therefore, a reader familiar with these concepts may quickly observe that, the set $\Sbbb_{0}$ is equal to $\{ \Sigma^{(1, 0)}, \Sigma^{(0, 1)} \}$.
We denote by $\Gcal_{\mathsf{toroidal}}$ the class of all graphs embeddable in the torus and by $\Gcal_{\mathsf{projective}}$ the class of all graphs embeddable in the projective plane.
Note that, the minor-obstruction set for $\Gcal_{\mathsf{toroidal}}$ has yet to be fully found.

We obtain the following two statements on the minor-universal obstruction for $\tw^{*}_{\gplanar}$ and $\apex^{*}_{\gplanar}$, which are directly implied from \autoref{surfex_univ_obs} and \autoref{killing_surfex_univ_obs} respectively.

\begin{proposition}
The set $\{ \mathscr{D}^{(1, 0)}, \mathscr{D}^{(0, 1)} \}$ is a minor-universal obstruction for $\tw^{*}_{\gplanar}.$
Moreover,
\begin{align*}
\cobs_{\leqslant_{\mathsf{m}}}(\tw^{*}_{\gplanar}) \ &= \ \big\{ \Gcal_{\mathsf{toroidal}}, \Gcal_{\mathsf{projective}} \big\}\text{ and}\\
\pobs_{\leqslant_{\mathsf{m}}}(\tw^{*}_{\gplanar}) \ &= \ \big\{ \obs_{\leqslant_{\mathsf{m}}}(\Gcal_{\mathsf{toroidal}}), \Ocal_{\mathsf{projective}} \big\}.
\end{align*}
\end{proposition}

By adding the minor-parametric graph $\mathscr{V}$ (see \autoref{svds_fig} for an illustration) to the previous set we obtain the following.

\begin{proposition}
The set $\{ \mathscr{D}^{(1, 0)}, \mathscr{D}^{(0, 1)}, \mathscr{V} \}$ is a minor-universal obstruction for $\apex^{*}_{\gplanar}.$
Moreover,
\begin{align*}
\cobs_{\leqslant_{\mathsf{m}}}(\apex^{*}_{\gplanar}) \ &= \ \big\{ \Gcal_{\mathsf{toroidal}}, \Gcal_{\mathsf{projective}}, \shallowvortexminors \big\}\text{ and}\\
\pobs_{\leqslant_{\mathsf{m}}}(\apex^{*}_{\gplanar}) \ &= \ \big\{ \obs_{\leqslant_{\mathsf{m}}}(\Gcal_{\mathsf{toroidal}}), \Ocal_{\mathsf{projective}}, \obs_{\leqslant_{\mathsf{m}}}(\shallowvortexminors) \big\}.
\end{align*}
\end{proposition}

The former two examples also represent two particular refinements of the GMST in the context of the $\mathsf{vga}$-lattice.
The former, $\tw^{*}_{\gplanar}$ corresponds to the version where the genus component is restricted to zero.
The latter, $\apex^{*}_{\gplanar}$ in addition also eliminates the presence of vortices.

\subsubsection{Parameter hierarchy to the ``left'' of treewidth}

Having introduced the previous graph parameters, we may now complement the hierarchy of parameters represented in \eqref{right_part} that sits to the ``left'' of treewidth as follows.
\begin{align}
\p_\mathsf{0} \ \preceq \ \hw \ \preceq \ \begin{matrix*} \svtw \\ \tw^{*}_{\gplanar} \end{matrix*} \ \preceq \ \apex^{*}_{\gplanar} \ \preceq \ \sctw \ \preceq \ \tw.
\label{left_part}
\end{align}

Note that $\svtw$ and $\tw^{*}_{\gplanar}$ are incomparable parameters.
Moreover, the corresponding minor-parametric obstructions are ordered as follows.
\begin{align}
\nonumber\big\{ \{K_{5}, K_{3,3}\} \big\} \ &\leqslant_{\mathsf{m}}^{**} \ \big\{ \{\Ocal_{\mathsf{projective}} \cup \Ocal_{\mathsf{linkless}}\} \big\}\\
&\leqslant_{\mathsf{m}}^{**} \ \big\{ \obs_{\leqslant_{\mathsf{m}}}(\Gcal_{\mathsf{toroidal}}), \Ocal_{\mathsf{projective}}, \obs_{\leqslant_{\mathsf{m}}}(\shallowvortexminors) \big\}\\[3pt]
\nonumber&\leqslant_{\mathsf{m}}^{**} \begin{matrix*} \big\{ \obs_{\leqslant_{\mathsf{m}}}(\shallowvortexminors) \big\} \\[5pt] \big\{ \obs_{\leqslant_{\mathsf{m}}}(\Gcal_{\mathsf{toroidal}}), \Ocal_{\mathsf{projective}} \big\} \end{matrix*} \ \leqslant_{\mathsf{m}}^{**} \ \big\{ \emptyset \big\} \ \leqslant_{\mathsf{m}}^{**} \ \emptyset.
\end{align}

\subsection{Elimination distance parameters}

In this subsection, we turn our focus to graph parameters defined through various vertex elimination procedures.
These parameters are used to quantify the ``distance'' a given graph has from belonging to a specific graph class, with the measure of this distance determined by the vertex elimination process described in each definition.

To begin, we present a general framework that encompasses many such parameters under a unified definition.

\subsubsection{Variants of elimination}\label{evjvhaiendndodbabsed}

Consider an integer $c \in [0, 2].$
A \emph{$c$-connected component} of a graph is a maximal set $C \subseteq V(G)$ of vertices such that $G[C]$ is \emph{$c$-connected}, that is it contains at least $c$ internally vertex-disjoint paths between every pair of vertices of $C.$
Note that all graphs are $0$-connected.
A subgraph $B$ of $G$ is a \emph{$c$-block} of $G$ if it is a $c$-connected component or an isolated vertex of $G$ or, in case $c = 2$, it  
is a bridge\footnote{A \emph{bridge} of a graph $G$ is a subgraph  of $G$ on two vertices and one edge, whose removal increases the number of $1$-connected components of $G$.} of $G$.
Let $c$-$\mathsf{blocks}(G)$ denote the set of $c$-blocks of $G.$
Utilizing these notions, we can define the \textsl{$c$-elimination distance} to a graph class.

Let $c \in [0, 2]$ and consider a minor-closed graph class $\Hcal$ where each graph in $\obs_{\leqslant_{\mathsf{m}}}(\Hcal)$ is $c$-connected.
We define the \emph{$c$-elimination distance} of a graph $G$ to $\Hcal$, which we denote by $(c, \Hcal)$-$\td(G),$ recursively as follows.
If $G \in \Hcal$, then $(c, \Hcal)$-$\td(G) = 0$, otherwise $(c, \Hcal)$-$\td(G) \leq k$ if there exists a vertex $x \in V(G)$ such that, for every $c$-block $B$ of $G - x$, $(c, \Hcal)$-$\td(B) \leq k-1$.
In other words,
\begin{align}
(c, \Hcal)\text{-}\td(G) \ \coloneqq \ \begin{cases}
                                       0,&\text{if }G \in \Hcal\\
                                       1 + \min\{ (c, \Hcal)\text{-}\td(G - v) \mid v \in V(G) \}, &\text{if }G\text{ is }c\text{-connected}\\
                                       \max\{ (c, \Hcal)\text{-}\td(H) \mid H \in c\text{-}\mathsf{blocks}(G) \}, &\text{otherwise}
                                       \end{cases}\label{c_htd_parameter}
\end{align}

It follows by definition that $(2,\Hcal)\text{-}\td \preceq (1,\Hcal)\text{-}\td \preceq (0,\Hcal)\text{-}\td$.
Depending on the choice of $c$ and $\Hcal$, 
$(c,\Hcal)\text{-}\td$ can express several graph modification parameters.
Note that $(c,\Hcal)\text{-}\td$ expresses a graph modification modulator/target scheme where the eliminated vertices are the modulator and the class $\Hcal$ is the target graph class.

Observe that $(0,\Hcal)\text{-}\td$ corresponds to the parameter known as the \emph{$\Hcal$-apex number}, denoted by $\apex_{\Hcal}$, of a graph $G$ that is the minimum size of a set $S \subseteq V(G)$ of vertices to remove such that $G - S \in \Hcal.$
\begin{align}
\apex_{\Hcal} \ &\coloneqq \ (0,\Hcal)\text{-}\td.\label{apex_parameter}
\end{align}

Additionally, $(1,\Hcal)\text{-}\td$ defines the \emph{elimination distance} of a graph $G$ to $\Hcal$, 
denoted by $\Hcal\text{-}\td(G),$ which was introduced by Bulian and Dawar in \cite{BulianD16graph, BulianD17Fixed}.
\begin{align}
\Hcal\text{-}\td \ &\coloneqq \ (1,\Hcal)\text{-}\td.\label{ed_parameter}
\end{align}

In addition, $(2,\Hcal)\text{-}\td$ defines the \emph{block elimination distance} of a graph $G$ to $\Hcal$, denoted by $\Hcal\text{-}\bed(G)$, which was introduced by Diner, Giannopoulou, Sau, and Stamoulis in \cite{DinerGST22Block}.
\begin{align}
\Hcal\text{-}\bed \ &\coloneqq \ (2,\Hcal)\text{-}\td.\label{bed_parameter}
\end{align}

Notice that both vertex cover and treedepth which we discussed in \autoref{wadfdjgjingngup} and \autoref{treedepth_par} respectively, form special cases of $(c, \Hcal)\text{-}\td$ as $\vc = (0, \excl_{\leqslant_{\mathsf{m}}}(K_2))\text{-}\td$ and $\td = (1, \excl_{\leqslant_{\mathsf{m}}}(K_2))\text{-}\td$.

\subsubsection{Obstructions for apex parameters}
\label{appe4jnfbvjpsjdpers}

Consider a proper minor-closed class $\Hcal$.
We define the \emph{$\Hcal$-barrier number} of a graph $G$ as follows.
\begin{align*}
\mathsf{barrier}_{\Hcal}(G) \ \coloneqq \ \max\{ k \in \Nbbb \mid k \cdot Z \leqslant_{\mathsf{m}} G\text{ and }Z \in \obs_{\leqslant_{\mathsf{m}}}(\Hcal) \}.
\end{align*}

Note that, for every graph $G \in \gall,$ $\mathsf{barrier}_{\Hcal}(G) \leq \mathsf{apex}_{\Hcal}(G),$ i.e., $\mathsf{barrier}_{\Hcal} \preceq \apex_{\Hcal}.$
We say that a proper minor-closed class $\Hcal$ has the \emph{Erd\H{o}s-P{\'o}sa property} if there exists a function $f \colon \mathbb{N} \to \mathbb{N}$ such that $\mathsf{apex}_{\Hcal}(G) \leq f(\mathsf{barrier}_{\Hcal}(G)),$ i.e., if $\mathsf{apex}_{\Hcal} \preceq \mathsf{barrier}_{\Hcal}.$
As a consequence of the celebrated Grid Theorem of Robertson and Seymour \cite{RobertsonS86GMV} the following statement holds.

\begin{proposition}[\!\cite{RobertsonS86GMV, CamesHJR19Atight}]\label{prop_erdos_posa_gap} 
A proper minor-closed class $\Hcal$ has the Erd\H{o}s-P{\'o}sa property if and only if $\obs_{\leqslant_{\mathsf{m}}}(\Hcal)$ contains a planar graph. {Moreover, if $
\Hcal$ has the Erdős-Pósa property, then  the gap function  belongs to $\Ocal(k \log k)$.} 
\end{proposition}

Now, given a graph $Z$ we define the \emph{connectivization} of $Z$, denoted by $\conn(Z),$ as the set of all edge-minimal connected graphs that contain $Z$ as a subgraph.
Clearly, if $Z$ is connected then $\conn(Z) = \{Z\}$.
Given a finite set $\Zcal$ of graphs, we define its \emph{minor-connectivization} as the set $\conn(\Zcal) \coloneqq \mathsf{min}(\cupall\{ \conn(Z) \mid Z \in \Zcal \})$ where $\mathsf{min}(\Zcal)$ denotes the set of minor-minimal graphs in $\Zcal.$

Given a proper minor-closed class $\Hcal$ we define its \emph{connectivity closure} to be the class
\begin{eqnarray}
\mathbf{C}(\Hcal) &  \coloneqq &  \{ G \mid \text{every connected component of }G\text{ belongs to }\Hcal\}.\label{conn_clo}
\end{eqnarray}

The minor-obstruction set of $\mathbf{C}(\Hcal)$ was described by Bulian and Dawar \cite{BulianD17Fixed} as follows.

\begin{proposition}[\!\cite{BulianD17Fixed}]\label{obs_conn}\label{b_d_th_conn}
For every proper minor-closed class $\Hcal$ it holds that
\begin{align*}
\obs_{\leqslant_{\mathsf{m}}}(\mathbf{C}(\Hcal)) \ = \ \conn(\obs_{\leqslant_{\mathsf{m}}}(\Hcal)).
\end{align*}
\end{proposition}

Finally, given a graph $Z$ we define the following minor-parametric graph.
\begin{align*}
\mathscr{H}_{Z} \ \coloneqq \ \langle k \cdot Z \rangle_{k\in\mathbb{N}}.
\end{align*}

Also, given a finite set of graphs $\Zcal$ we define the following minor-parametric family.
\begin{align*}
\mathfrak{H}_{\Zcal} \ \coloneqq \ \big\{ \mathscr{H}_{Z} \mid Z \in \Zcal \big\}.
\end{align*}

The following theorem is now a simple consequence of \autoref{prop_erdos_posa_gap}.

\begin{theorem}\label{prop_erdos_posa_planar}
Let $\Hcal$ be a proper minor-closed class whose minor-obstruction set contains a planar graph.
Then, the set $\mathfrak{H}_{\obs_{\leqslant_{\mathsf{m}}}(\Hcal)}$ is a minor-universal obstruction for $\apex_{\Hcal}$ with a polynomial gap function which belongs to $\Ocal(k \log k)$.
Moreover,
\begin{align*}
\cobs_{\leqslant_{\mathsf{m}}}(\apex_{\Hcal}) \ &= \ \big\{\closure{\leqslant_{\mathsf{m}}}{\mathscr{H}} \mid \mathscr{H} \in \mathfrak{H}_{\obs_{\leqslant_{\mathsf{m}}}(\Hcal)} \big\}\text{ and}\\
\pobs_{\leqslant_{\mathsf{m}}}(\apex_{\Hcal}) \ &= \ \big\{\conn(\obs_{\leqslant_{\mathsf{m}}}(\closure{\leqslant_{\mathsf{m}}}{\{Z\}})) \mid Z \in \obs_{\leqslant}(\Hcal) \big\}.
\end{align*}
\end{theorem}
\begin{proof}
Let $\Zcal = \obs_{\leqslant_{\mathsf{m}}}(\Hcal).$
First and foremost, it is easy to see that $\barrier_{\Hcal} = \p_{\mathfrak{H}_{\Zcal}}.$
Then, the fact that $\mathfrak{H}_{\Zcal}$ is a minor-universal obstruction for $\apex_{\Hcal}$ with a gap function that belongs to $\Ocal(k \log k)$, follows directly from \autoref{prop_erdos_posa_gap}.
By \autoref{cobs_uobs}, $\cobs_{\leqslant_{\mathsf{m}}}(\apex_{\Hcal}) = \big\{ \closure{\leqslant_{\mathsf{m}}}{\mathscr{H}} \mid \mathscr{H} \in \mathfrak{H}_{\Zcal} \big\}$.
Moreover, it is easy to observe that for every $\mathscr{H}_{Z} \in \mathfrak{H}_{\Zcal}$, $\closure{\leqslant_{\mathsf{m}}}{\mathscr{H}_{Z}}$ contains all graphs whose connected components are minors of $Z$, i.e. it is equal to the class $\mathbf{C}(\closure{\leqslant_{\mathsf{m}}}{\{Z\}})$.
Then, by definition and \autoref{obs_conn}, $\pobs_{\leqslant_{\mathsf{m}}}(\apex_{\Hcal}) = \big\{ \conn(\obs_{\leqslant_{\mathsf{m}}}(\closure{\leqslant_{\mathsf{m}}}{\{Z\}})) \mid Z\in\Zcal \big\}$.
\end{proof}

We continue this subsection by presenting several examples of instantiations of the $c$-elimination distance to $\Hcal$.
For each example, we provide their corresponding minor-universal obstructions by applying \autoref{prop_erdos_posa_planar}.

\paragraph{Feedback vertex set.}

The \emph{feedback vertex set} parameter, denoted by $\fvs$, is defined so that for every graph $G$, $\fvs(G)$ is the minimum size of a set $X \subseteq V(G)$ such that $G-X$ is acyclic.
In other words,
\begin{align}
\fvs \ &\coloneqq \ \apex_{\gforest}.
\end{align}

We define the minor-parametric graph $\mathscr{K}^{(K_3)} = \langle k \cdot K_3 \rangle_{k \geq 1}$.
Let $\minorsoftriangles = \closure{\leqslant_{\mathsf{m}}}{\mathscr{K}^{(K_3)}}$ be the class of all graphs whose connected components have at most three vertices.
Note that, \autoref{obs_conn} implies that $\obs_{\leqslant_{\mathsf{m}}}(\minorsoftriangles) = \conn(\{4 \cdot K_1\}) = \{P_4, K_{1,3} \}$.
Applying \autoref{prop_erdos_posa_planar}, we have the following proposition.

\begin{proposition}
The set $\{ \mathscr{K}^{(K_3)} \}$ is a minor-universal obstruction for $\fvs$ with a gap function which belongs to $\Ocal(k \log k)$.
Moreover, $\cobs_{\leqslant_{\mathsf{m}}}(\fvs) = \{ \minorsoftriangles \}$ and $\pobs_{\leqslant_{\mathsf{m}}}(\fvs) = \big\{ \{P_4, K_{1,3}\} \big\}$.
\end{proposition}

\paragraph{Apex to outerplanarity.}

Another example that is illustrative of the applicability of \autoref{prop_erdos_posa_planar} is the parameter defined as the minimum size of a set $X \subseteq V(G)$ such that $G - X$ is outerplanar.
This corresponds to the following instantiation of $\apex_{\Hcal}.$
\begin{align*}
\mathsf{apexouter} \ \coloneqq \ \apex_{\gouterplanar}.
\end{align*}

For a study of the minor-obstruction set of $\Gcal_{\mathsf{apexouter}, 1}$ we refer the reader to \cite{DingD16excluded}.
It is well-known that $\obs_{\leqslant_{\mathsf{m}}}(\gouterplanar) = \{ K_{4}, K_{2,3} \}$.
We define the minor-parametric graphs $\mathscr{K}^{(K_4)} \coloneqq \langle k\cdot K_4 \rangle_{k\geq 1}$ and $\mathscr{K}^{(K_{2,3})} \coloneqq \langle k \cdot K_{2,3} \rangle_{k \geq 1}$.
By definition, $\closure{\leqslant_{\mathsf{m}}}{\mathscr{K}^{(K_4)}}$ is the class of graphs whose connected components are minors of $K_4,$ i.e., have at most four vertices.
Note that, from \autoref{obs_conn}, $\obs_{\leqslant_{\mathsf{m}}}(\closure{\leqslant_{\mathsf{m}}}{\mathscr{K}^{(K_4)}}) = \conn(\{5 \cdot K_1\})$.
Therefore, $\obs_{\leqslant_{\mathsf{m}}}(\closure{\leqslant_{\mathsf{m}}}{\mathscr{K}^{(K_4)}}) = \{ P_{5}, K_{1,4}, K_{1,3}^{s} \}$ which are all non-isomorphic trees on five vertices.
See \autoref{apex_outerplanar_pobs_1} for an illustration.

\begin{figure}[htbp]
\centering
\includegraphics[width=0.25\linewidth]{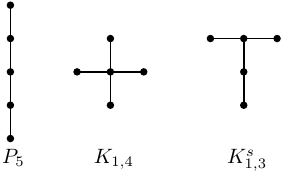}
\caption{\label{apex_outerplanar_pobs_1} The graphs $P_{5}$, $K_{1,4}$, and $K_{1,3}^s$ which are all non-isomorphic trees on five vertices.}
\end{figure}

Moreover, $\closure{\leqslant_{\mathsf{m}}}{\mathscr{K}^{(K_{2,3})}}$ is the class whose connected components are minors of $K_{2,3}$.
First of all, observe that $K_{1,4} \not\leqslant_{\mathsf{m}} K_{2,3}.$
Therefore, $\obs_{\leqslant_{\mathsf{m}}}(\closure{\leqslant_{\mathsf{m}}}{\mathscr{K}^{(K_{2,3})}}) = \{ G \in \conn(\{ 6 \cdot K_{1} \}) \mid K_{1,4} \not\leqslant_{\mathsf{m}} G \} \cup \{ K_{1,4}, K_{4}, C_{5}, K_{3}^{de}, K_{3}^{es} \}$ where $ \conn(\{ 6 \cdot K_{1} \}) = \{ P_{6}, Q_{2}, K_{1,5}, K_{1,4}^{\mathsf{s}}, K_{1,3}^{\mathsf{ds}}, K_{1,3}^{\mathsf{2s}} \}$ which are all non-isomorphic trees on six vertices and $\{ G \in \conn(\{ 6 \cdot K_{1} \}) \mid K_{1,4} \not\leqslant_{\mathsf{m}} G \} = \{ P_{6}, Q_{2}, K_{1,3}^{ds}, K_{1,3}^{2s} \}.$
See \autoref{apex_planar_pobs_1} for an illustration of $\conn(\{ 6 \cdot K_{1} \})$ and \autoref{apex_outerplanar_pobs_2} for an illustration of $\{ K_{4}, C_{5}, K_{3}^{de}, K_{3}^{es} \}$.

\begin{figure}[htbp]
\centering
\includegraphics[width=0.6\linewidth]{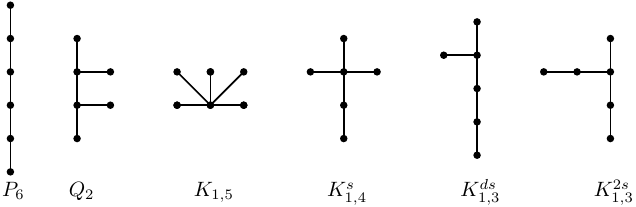}
\caption{\label{apex_planar_pobs_1}The graphs $P_{6}$, $Q_{2}$, $K_{1,5}$, $K_{1,4}^{s}$, $K_{1,3}^{ds}$, and $K_{1,3}^{2s}$ which are all non-isomorphic trees on six vertices.}
\end{figure}

\begin{figure}[htbp]
\centering
\includegraphics[width=0.5\linewidth]{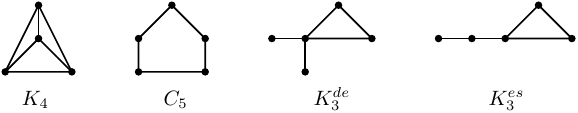}
\caption{\label{apex_outerplanar_pobs_2}The graphs $K_{4}$, $C_{5}$, $K_{3}^{de}$, and $K_{3}^{es}$.}
\end{figure}

Applying \autoref{prop_erdos_posa_planar} we have the following proposition.

\begin{theorem}
The set $\{\mathscr{K}^{(K_4)}, \mathscr{K}^{(K_{2,3})}\}$ is a minor-universal obstruction for $\mathsf{apexouter}$ with a gap function that belongs to $\Ocal(k \log k)$.
Moreover,
\begin{align*}
\cobs_{\leqslant_{\mathsf{m}}}(\mathsf{apexouter}) \ = \ \big\{ &\closure{\leqslant_{\mathsf{m}}}{\mathscr{K}^{(K_4)}}, \closure{\leqslant_{\mathsf{m}}}{\mathscr{K}^{(K_{2,3})}} \big\}\text{, and}\\
\pobs_{\leqslant_{\mathsf{m}}}(\mathsf{apexouter}) \ = \ \big\{ &\{P_{5}, K_{1,4}, K_{1,3}^s \},\\
&\{ P_{6}, Q_{2}, K_{1,3}^{ds}, K_{1,3}^{2s}, K_{1,4}, K_{4}, C_{5}, K_{3}^{de}, K_{3}^{es} \} \big\}.
\end{align*}
\end{theorem}

\paragraph{Apex to planarity.}

A consequence of the study in \cite{PaulPTW24Obstructions} on obstructions to Erd\H{o}s-P{\'o}sa dualities, is a minor-universal obstruction for $\apex_{\Hcal},$ for every proper minor-closed class $\Hcal,$ vastly generalizing \autoref{prop_erdos_posa_gap} and \autoref{prop_erdos_posa_planar}.

\medskip
A particularly interesting (and simple to present) instantiation of this result is to consider the parameter defined as the minimum size of a set $X \subseteq V(G)$ such that $G - X$ is planar.
This corresponds to the following instantiation of $\apex_{\Hcal}.$
\begin{align}
\mathsf{apexplanar} \ \coloneqq \ \apex_{\gplanar}.
\end{align}

The parameter $\mathsf{apexplanar}$ is known in the literature as the \textsl{planarizer number}.
The minor-obstruction set $\obs_{\leqslant_{\mathsf{m}}}(\Gcal_{\mathsf{apexplanar}, k})$ is unknown for every positive value of $k$ and its size is expected to grow rapidly as a function of $k$ (see \cite{Dinneen97} for  an exponential lower bound and \cite{SauST23apices} for a triply exponential upper bound).
The identification of $\obs_{\leqslant_{\mathsf{m}}}(\Gcal_{\mathsf{apexplanar}, k})$ is a non-trivial problem even for small values of $k.$
In particular, it has been studied extensively for the case where $k = 1$ in \cite{LiptonMMPRT16sixv, Mattman16forb, Yu06more}.
In this direction, Mattman and Pierce \cite{MattmanPierce2017} conjectured that $\obs_{\leqslant_{\mathsf{m}}}(\Gcal_{\mathsf{apexplanar}, k})$ contains the $Y \Delta Y$-families of $K_{n+5}$ and $K_{3^2,2^n}$ and provided evidence towards this.
Moreover, Jobson and Kézdy \cite{JobsonK21allm} identified \textsl{all} graphs in $\obs_{\leqslant_{\mathsf{m}}}(\Gcal_{\mathsf{apexplanar}, 1})$ of connectivity two.
In addition, they reported that $|\obs_{\leqslant_{\mathsf{m}}}(\Gcal_{\mathsf{apexplanar}, 1})| \geq 401.$

\medskip
Recall that, by the Kuratowski-Wagner theorem $\obs_{\leqslant_{\mathsf{m}}}(\gplanar) = \{ K_{5}, K_{3, 3} \}.$
We define the minor-parametric graphs $\mathscr{K}^{(K_{5})} \coloneqq \langle k \cdot K_{5} \rangle_{k \in \Nbbb}$ and $\mathscr{K}^{(K_{3, 3})} \coloneqq \langle k \cdot K_{3, 3} \rangle_{k \in \Nbbb}.$
By definition, $\closure{\leqslant_{\mathsf{m}}}{\mathscr{K}^{(K_5)}}$ is the class of graphs whose connected components are minors of $K_5,$ i.e., have at most five vertices.
Note that, from \autoref{obs_conn}, $\obs_{\leqslant_{\mathsf{m}}}(\closure{\leqslant_{\mathsf{m}}}{\mathscr{K}^{(K_5)}}) = \conn(\{6 \cdot K_1\})$,
therefore $\obs_{\leqslant_{\mathsf{m}}}(\closure{\leqslant_{\mathsf{m}}}{\mathscr{K}^{(K_5)}}) = \{ P_{6}, Q_{2}, K_{1,5}, K_{1,4}^{\mathsf{s}}, K_{1,3}^{\mathsf{ds}}, K_{1,3}^{\mathsf{2s}} \}$ which are all non-isomorphic trees on six vertices.
See \autoref{apex_planar_pobs_1} for an illustration.
Moreover, $\closure{\leqslant_{\mathsf{m}}}{\mathscr{K}^{(K_{3,3})}}$ is the class whose connected components are minors of $K_{3,3}$.
In this case identifying the minor-obstruction set of $\closure{\leqslant_{\mathsf{m}}}{\mathscr{K}^{(K_{3,3})}}$ is more tricky.
With similar tricks as in the case of $\closure{\leqslant_{\mathsf{m}}}{\mathscr{K}^{(K_{2,3})}}$ one may attempt to identify this set, however this is beyond the scope of this work.
Now, recall the definition of the minor-parametric graphs $\mathscr{D}^{(1, 0)}$ (see \autoref{fig_torus_grid}) and $\mathscr{D}^{(0, 1)}$ (see \autoref{fig_projective_grid}).
A consequence of the general result in \cite{PaulPTW24Obstructions} is the following.

\begin{proposition}
The set $\{ \mathscr{D}^{(1, 0)}, \mathscr{D}^{(0, 1)}, \mathscr{K}^{(K_{5})}, \mathscr{K}^{(K_{3, 3})} \}$ is a minor-universal obstruction for $\mathsf{apexplanar}$ with single exponential gap.
Moreover,
\begin{align*}
\cobs_{\leqslant_{\mathsf{m}}}(\mathsf{apexplanar}) \ = \ \big\{ &\Gcal_{\mathsf{toroidal}}, \Gcal_{\mathsf{projective}}, \closure{\leqslant_{\mathsf{m}}}{\mathscr{K}^{(K_5)}}, \closure{\leqslant_{\mathsf{m}}}{\mathscr{K}^{(K_{3,3})}} \big\}\text{, and}\\
\pobs_{\leqslant_{\mathsf{m}}}(\mathsf{apexplanar}) \ = \ \big\{ &\obs_{\leqslant_{\mathsf{m}}}(\Gcal_{\mathsf{toroidal}}),\\
&\Ocal_{\mathsf{projective}},\\
&\{ P_{6}, Q_{2}, K_{1,5}, K_{1,4}^{\mathsf{s}}, K_{1,3}^{\mathsf{ds}}, K_{1,3}^{\mathsf{2s}} \},\\
&\obs_{\leqslant_{\mathsf{m}}}(\closure{\leqslant_{\mathsf{m}}}{\mathscr{K}^{(K_{3,3})}}) \big\}.
\end{align*}
\end{proposition}

It is easy to see that $\{( k+1) \cdot K_{5}, (k+1) \cdot K_{3,3}\} \subseteq \obs_{\leqslant_{\mathsf{m}}}(\Gcal_{\mathsf{apexplanar}, k}),$ for every $k \in \Nbbb.$
Our results, together with the fact that $\Sbbb_{0} = \{ \Sigma^{(1,0)}, \Sigma^{(0,1)} \},$ provide the following additional information on $\obs_{\leqslant_{\mathsf{m}}}(\Gcal_{\mathsf{apexplanar}, k})$: for every $k\in \Nbbb,$ it contains some graph, say $G_{k}^{\mathsf{t}}$, embeddable in the torus and some graph, say  $G_{k}^{\mathsf{p}},$ embeddable in the projective plane.
Most importantly, our results indicate, that the four-member subset $\{(k+1) \cdot K_{5}, (k+1) \cdot K_{3,3}, G_{k}^{\mathsf{t}}, G_{k}^{\mathsf{p}} \}$ of $\obs_{\leqslant_{\mathsf{m}}}(\Gcal_{\mathsf{apexplanar}, k})$ is sufficient to determine the approximate behaviour of the planarizer number.

\subsubsection{Obstructions for elimination in blocks}
\label{apeplicx}

Let $\Hcal$ be a minor-closed class.
We denote by $\mathbf{A}(\Hcal)$ the class of all apex graphs of $\Hcal$, i.e. $\mathbf{A}(\Hcal) \coloneqq \Gcal_{\apex_{\Hcal},1}$.
Bulian and Dawar in \cite{BulianD17Fixed} gave the following alternative definition for $\Hcal\text{-}\td$.

\begin{proposition}[\!\cite{BulianD17Fixed}]\label{prop_dawar_ed}
Let $\Hcal$ be a proper minor-closed class and
\begin{align*}
\Ccal_{0} \ &\coloneqq \ \mathbf{C}(\Hcal)\text{ and}\\
\Ccal_{i+1} \ &\coloneqq \ \mathbf{C}(\mathbf{A}(\Ccal_{i}))\text{, for every }i\in\Nbbb.
\end{align*}

Then, $\Gcal_{\Hcal\text{-}\td, k} = \Ccal_{k}$, for every non-negative integer $k.$
\end{proposition}

A \emph{2-rooted-graph} is a triple $\mathbf{Z}=(Z,v,u)$ where $Z$ is a connected graph and $v,u\in V(Z).$
We say that two 2-rooted-graphs $\mathbf{Z}=(Z,v,u)$ and $\mathbf{Z'}=(Z',v',u')$ are \emph{isomorphic}  
if there is an isomorphism $\sigma \colon V(Z)\to V(Z')$ such that $\sigma(v)=v'$ and $\sigma(u)=u'.$
We define $\mathcal{R}(Z)$ as the set of all pairwise non-isomorphic 2-rooted-graphs whose underlying graph is $Z.$

Let $\Hcal$ be a proper minor-closed class.
An \emph{$\Hcal$-chain} of size $t$ is a graph obtained by considering a sequence 
\begin{align}
(Z^1,v^1,u^1), \ldots, (Z^t,v^t,u^t)\label{seq_two_root}
\end{align}
of $2$-rooted graphs, where for each $i \in [t],$ there exists some $Z^i \in \conn(\obs_{\leqslant_{\mathsf{m}}}(\Hcal))$ such that $(Z^{i}, v^i, u^i) \in \Rcal(Z),$ and adding the edges in $\big\{\{u^{i}, v^{i+1}\} \mid i \in [t-1] \big\}$ in the graph obtained as the disjoint union of the graphs in \eqref{seq_two_root}.
An $\Hcal$-chain where $(Z^{i},v^i,u^i)$ is isomorphic to some particular  $\mathbf{Z},$ for every $i \in [t],$ is called \emph{$\mathbf{Z}$-chain}.
For every 2-rooted graph $\mathbf{Z} = (Z,v,u)\in\mathcal{R}(Z),$ we define
the minor-parametric graph $\mathscr{E}^{\mathbf{Z}} = \langle\mathscr{E}^{\mathbf{Z}}_{t}\rangle_{t\in\mathbb{N}}$
where $\mathscr{E}^{\mathbf{Z}}_{t}$ is the $\mathbf{Z}$-chain of size $t.$

The following observation follows directly from the definitions.

\begin{observation}\label{one_obst_Htd}
There exists a function $f \colon \nn{1}{1}$ such that for every proper minor-closed class $\Hcal$ and every $k \in \Nbbb,$ every $\Hcal$-chain of size $t \cdot f(h),$ where $h = \max\{|Z| \mid Z \in \obs_{\leqslant_{\mathsf{m}}}(\Hcal) \},$ contains $\mathscr{E}^{\mathbf{Z}}_{t}$ as a minor for some $\mathbf{Z} \in \mathcal{R}(Z)$ where $Z \in \conn(\obs_{\leqslant_{\mathsf{m}}}(\Hcal)).$
\end{observation}
\begin{proof}
Let $\Zcal = \conn(\obs_{\leqslant_{\mathsf{m}}}(\Hcal))$ and $c = \max_{Z \in \Zcal} |\mathcal{R}(Z)|$.
Then, let $f(t) = c \cdot t$.
By definition, for every $t \in \Nbbb$, every $\mathscr{E}^{\mathbf{Z}}_{t}$ is a $\mathbf{Z}$-chain of length $t$ where $\mathbf{Z} = (Z, u, v)$.
Moreover, it is easy to see that for every $\Hcal$-chain of length at least $f(t)$ there exists a $\mathbf{Z} = (Z, u, v) \in \mathcal{R}(Z),$ where $Z \in \Zcal,$ that repeats at least $t$ times, i.e. it contains $\mathscr{E}^{\mathbf{Z}}_{t}$ as a minor.
\end{proof}

The next observation is also an immediate consequence of the definition of $\Hcal$-$\td$ and the minor-parametric graphs $\mathscr{E}^{\mathbf{Z}}.$

\begin{observation}\label{obst_Htd_lower}
Let $\Hcal$ be a minor-closed class and $Z \in \conn(\obs_{\leqslant_{\mathsf{m}}}(\Hcal)).$
Then for every graph $G,$ every 2-rooted-graph graph $\mathbf{Z} \in \mathcal{R}(Z),$ every $\mathscr{E}^{\mathbf{Z}} = \langle \mathscr{E}^{\mathbf{Z}}_{t} \rangle_{{t\in\mathbb{N}}},$ and every $t\in\Nbbb,$ it holds that 
$\Hcal\text{-}\td(\mathscr{E}^{\mathbf{Z}}_{t}) = \Omega(\log t).$
\end{observation}
\begin{proof}
Let $G$ be a graph and $\mathbf{Z} \in \mathcal{R}(Z)$ such that $\mathscr{E}^{\mathbf{Z}}_{t} \leqslant G,$ for some $t \in \Nbbb.$
We argue that in this case $\Hcal\text{-}\td(G) = \Omega(\log t)$.
To see this, observe that every vertex of $\mathscr{E}^{\mathbf{Z}}_{t}$ intersects precisely one copy of the underlying graph $Z$ of $\mathbf{Z}$ and its removal separates the graph $\mathscr{E}^{\mathbf{Z}}_{t}$ in at most two connected components.
Therefore, inductively, the minimum height elimination tree is obtained by removing vertices so as to balance the number of copies of $Z$ in each connected component.
This gives an elimination tree of height $\Omega(\log t)$.
\end{proof}

We are now in the position to prove the following theorem.

\begin{theorem}\label{more_erdos_posa_planar}
Let $\Hcal$ be a proper-minor closed class whose minor-obstruction set contains a planar graph.
Then, any $\lesssim$-minimization of the set
$$\bigcup_{Z \in \conn(\obs_{\leqslant_{\mathsf{m}}}(\Hcal))} \{ \mathscr{E}^{\mathbf{Z}} \mid \mathbf{Z} \in \Rcal(Z) \}$$
is a minor-universal obstruction for $\Hcal\text{-}\td$.
\end{theorem}
\begin{proof}
The lower bound of the desired equivalence follows from \autoref{obst_Htd_lower}.
Therefore, it suffices to prove the upper bound.
For every $k \in \Nbbb,$ let $\Hcal\text{-}\Pcal_{k}$ be the set that consists of all $\Hcal$-chains of size $k.$
We prove that there exists a function $h \colon \mathbb{N} \to \mathbb{N}$ such that, for every graph $G$, if $\Hcal\text{-}\td(G) > h(k)$ then there exists $H \in \Hcal\text{-}\Pcal_{k}$ such that $H \leqslant G$.
We define the  following graph classes.
\begin{align*}
\Gcal_{0} \ &= \ \mathbf{C}(\Gcal_{\barrier_{\Hcal}, 1})\text{ and}\\
\Gcal_{i+1} \ &= \ \mathbf{C}(\Gcal_{\barrier_{\Gcal_{i}}, 1})\text{, for every }i \geq 0.
\end{align*}

We first argue that for every $i \geq 0$,
\begin{enumerate}
\item[(1)] $\obs_{\leqslant_{\mathsf{m}}}(\Gcal_{i})$ contains a planar graph;
\item[(2)] Every graph $H \in \obs_{\leqslant_{\mathsf{m}}}(\Gcal_{i})$ contains an $\Hcal$-chain of size at least $i+2$ as a minor.
\end{enumerate}

For $i = 0$, condition (1) holds trivially.
As for condition (2), let $\Zcal = \obs_{\leqslant_{\mathsf{m}}}(\Hcal).$
Observe that, by definition of $\barrier_{\Hcal}$, $\obs_{\leqslant_{\mathsf{m}}}(\Gcal_{\barrier_{\Hcal}, 1}) = \big\{ 2 \cdot Z \mid Z \in \Zcal \big\}$.
Then, \autoref{obs_conn} implies that $\obs_{\leqslant_{\mathsf{m}}}(\mathbf{C}(\Gcal_{\barrier_{\Hcal}, 1}))$ contains the minor-minimal graphs in $\Hcal\text{-}\Pcal_{2}$.

Now, assume that (1) and (2) hold for every non-negative integer less than $i \geq 1$.
To prove (1), by the induction hypothesis $\obs_{\leqslant_{\mathsf{m}}}(\Gcal_{i-1})$ contains a planar graph $H'.$
Let $H$ be a graph obtained from $2 \cdot H'$ by adding an edge $xy$ between a vertex $x$ of one of the two copies of $H'$ and a vertex $y$ of the other copy of $H'.$
Now, we can draw each $H'$ such that the outer faces contain the vertices $x$ and $y$ respectively.
Then clearly, $H$ is planar and by \autoref{obs_conn}, $H \in \obs_{\leqslant_{\mathsf{m}}}(\Gcal_{i}).$

To prove (2), let $H' \in \obs_{\leqslant_{\mathsf{m}}}(\Gcal_{i - 1}).$
By the induction hypothesis, $H'$ contains an $\Hcal$-chain of size at least $i + 1$ as a minor.
By \autoref{obs_conn} and the definition of $\Gcal_{i},$ we have that $\obs_{\leqslant_{\mathsf{m}}}(\Gcal_{i})$ contains a graph $H$ that is obtained from $2 \cdot H'$ by adding an edge $xy$ between a vertex $x$ of one of the two copies of $H'$ and a vertex $y$ of the other copy of $H'.$
Let $A$ be the graph induced by one of the two copies of $H'$ in $H$ and $B$ be the graph induced by the other copy of $H'$ in $H$. 
Now, let $\mathcal{X} = \big\{X_{u} \mid u \in V(P_{Z})\big\}$ and $\mathcal{X}' = \big\{X'_{u} \mid u \in V(P_{Z})\big\}$) be two minor models of an $\Hcal$-chain of size at least $i + 1$ in $A$ and in $B$ respectively. 
Also consider a $u-v$ path $P$ in $H$ such that $u \in X_{u'}$ for some $X_{u'} \in \mathcal{X}$ and $v \in X_{v'}$ for some $X_{v'} \in \mathcal{X}'$ that is internally disjoint from any set in $\mathcal{X}$ and $\mathcal{X}'$.
It is easy to verify that the graph induced by $\cupall \mathcal{X} \cup \cupall \mathcal{X}' \cup P$ contains an $\Hcal$-chain of size at least $i+2$ as a minor.

Therefore, by \autoref{prop_erdos_posa_planar}, for every $i \in \Nbbb,$ there exists a function $f_{i} \colon \Nbbb \to \Nbbb$ that certifies that $(\Gcal_{i}, \gall)$ is an EP-pair with gap function $f_{i}.$
For every $k \in \Nbbb$, we define
$$h(k) := \sum_{i \in [k]} f_{i}(1).$$

For a minor-closed class $\Gcal$, let $\Ccal_{0}(\Gcal) = \mathbf{C}(\Gcal)$ and $\Ccal_{i+1}(\Gcal) = \mathbf{C}(\mathbf{A}(\Ccal_{i}(\Gcal)))$.
Let $G$ be a graph such that $\Hcal\text{-}\td(G) > h(k)$.
By \autoref{prop_dawar_ed}, $G \notin \Ccal_{h(k)}(\Hcal)$.
Since for every non-negative integer $k$, $\Gcal_{\apex_{\Hcal}, k} \subseteq \Gcal_{\Hcal\text{-}\td, k}$ and for every minor-closed class $\Gcal$, $\Gcal \subseteq \mathbf{C}(\Gcal)$, it also holds that $G \notin \Ccal_{h(k) - \alpha}(\mathbf{C}(\Gcal_{\apex_{\Hcal}, \alpha}))$ for any constant $\alpha \leq h(k)$.
Moreover, by \autoref{prop_erdos_posa_gap}, for every $k$, $\Gcal_{\barrier_{\Hcal}, k} \subseteq \Gcal_{\apex_{\Hcal}, f_{0}(k)}$ where $f_{0}(k) = \Ocal(k \log k)$.
Hence, $G \notin \Ccal_{h(k) - \alpha}(\mathbf{C}(\Gcal_{\barrier_{\Hcal}, \beta}))$, where $f_{0}(\beta) = \alpha$.
Set $\beta = 1$.
Then $G \notin \Ccal_{h(k) - f_{0}(1)}(\mathbf{C}(\Gcal_{\barrier_{\Hcal}, 1}))$ or equivalently $G \notin \Ccal_{h(k) - f_{0}(1)}(\Gcal_{0})$. Inductively repeating this argument $h(k)$ many times, implies that $G \notin \Ccal_{0}(\mathbf{C}(\Gcal_{\barrier_{\Gcal_{h(k)-1}}, 1}))$.
Hence, $G \notin \Gcal_{h(k)}$.
Then, by (2), our claim follows.

To conclude, we have that if $\Hcal\text{-}\td(G) > h(k) \cdot f(h),$ where $f \colon \Nbbb \to \Nbbb$ is the function implied by \autoref{one_obst_Htd} and $h = \max\{|Z| \mid Z \in \Zcal \}$, then $G$ contains $\mathscr{E}^{\mathbf{Z}}_{k}$ as a minor for some $\mathbf{Z} \in \Rcal(Z)$ where $Z \in \conn(\Zcal).$
\end{proof}

We continue our presentation with a few indicative examples of elimination distance parameters and, for each of them, we give the corresponding universal obstructions, by applying \autoref{more_erdos_posa_planar}.

\paragraph{Elimination distance to a forest.}

The next parameter under consideration is the \emph{elimination distance to a forest}, that is defined as follows.
\begin{align}
\edforest \ &\coloneqq \ (1, \gforest)\text{-}\td.
\end{align}

In other words, for every graph $G$ and non-negative integer $k$, $\edforest(G) \leq k$ if and only if $G$ has a vertex $v$ where, for each connected component $C$ of $G - v$, it holds that $\edforest(C) \leq k-1$ and where we agree that, for every acyclic graph $F$, $\edforest(F) = 0$.
The parameter $\edforest$ has been considered in \cite{DekkerJansen2024FVS} in the context of kernelization algorithms.

Towards an application of \autoref{more_erdos_posa_planar}, let $\{\mathscr{P}^{(K_{3}), 1}, \mathscr{P}^{(K_{3}), 2}\}$, where $\mathscr{P}^{(K_{3}), 1}$ and $\mathscr{P}^{(K_{3}), 2}$ are the two minor-parametric graphs depicted in \autoref{fig_K3_paths}.

\begin{figure}[htbp]
\centering
\includegraphics[width=0.85\linewidth]{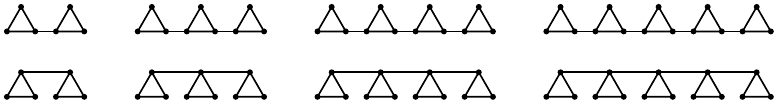}
\caption{\label{fig_K3_paths} Two minor-parametric graphs, corresponding to the two ways triangles can be linearly arranged on a path. On the first row we have $\mathscr{P}^{(K_{3}), 1} = \langle \mathscr{P}^{(K_{3}), 1}_2, \mathscr{P}^{(K_{3}), 1}_3, \mathscr{P}^{(K_{3}), 1}_4, \mathscr{P}^{(K_{3}), 1}_5, \ldots \rangle$ while on the second row we have $\mathscr{P}^{(K_{3}), 2} = \langle \mathscr{P}^{(K_{3}), 2}_2, \mathscr{P}^{(K_{3}), 2}_3, \mathscr{P}^{(K_{3}), 2}_4, \mathscr{P}^{(K_{3}), 2}_5, \ldots \rangle$.}
\end{figure}

These two minor-parametric graphs are created by the two different ways to linearly arrange triangles along a path and are generated by the two non-isomorphic 2-rooted graphs $(K_{3}, v, u)$ and $(K_{3}, v, v)$, where $v \neq u$.
Let $\Ccal^{(1)} = \closure{\leqslant_{\mathsf{m}}}{\mathscr{P}^{(K_{3}), 1}}$ and  $\Ccal^{(2)} = \closure{\leqslant_{\mathsf{m}}}{\mathscr{P}^{(K_{3}),2}}$.
Also, let $\Ocal^{(1)}$ and $\Ocal^{(2)}$ be the set of graphs depicted in \autoref{obs_path_triangles_1} and \autoref{obs_path_triangles_2} respectively.
We prove the following two lemmata.

\begin{figure}[htbp]
\centering
\includegraphics[width=0.5\linewidth]{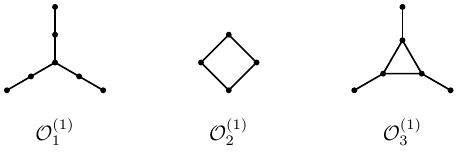}
\caption{\label{obs_path_triangles_1}The minor-obstruction set $\Ocal^{(1)} = \big\{\Ocal^{(1)}_1,\Ocal^{(1)}_2,\Ocal^{(1)}_3\big\}$.}
\end{figure}

\begin{lemma}
\label{pbs3e}
$\obs_{\leqslant_{\mathsf{m}}}(\Ccal^{(1)}) = \Ocal^{(1)}$.
\end{lemma}
\begin{proof}
It is easy to see that for every $Z \in \Ocal^{(1)}$, $Z \notin \Ccal^{(1)}$ and moreover that $\Ocal^{(1)}$ is a minor-antichain.
Let $G$ be a connected graph such that $\Ocal^{(1)}_{1}, \Ocal^{(1)}_{2}, \Ocal^{(1)}_{3} \not\leqslant_{\mathsf{m}} G$.
Let $P$ be a longest path in $G$.
We claim that for every $C \in \textsf{cc}(G - P)$, $|V(C)| = 1$.
Assume that there exists a connected component $C \in \textsf{cc}(G - P)$ such that $|V(C)| \geq 2$.
It is easy to see that depending on which vertices are contained in $N(V(C)) \cap V(P)$, we either contradict the longest path assumption or we conclude that $\Ocal^{(1)}_{1} \leqslant_{\mathsf{m}} G$.

Now, since $\Ocal^{(1)}_{2} \not\leqslant_{\mathsf{m}} G$ then all cycles in $G$ are triangles. Moreover, since for every $C \in \textsf{cc}(G - P)$, $|V(C)| = 1$, and $\Ocal^{(1)}_{3} \not\leqslant_{\mathsf{m}} G$, all triangles in $G$ have at least two vertices on $P$ which have to be consecutive.
Hence $G \in \Ccal^{(1)}$.
If $G$ is not connected we apply the same arguments in each connected component of $G$.
\end{proof}

\begin{figure}[htbp]
\centering
\includegraphics[width=0.85\linewidth]{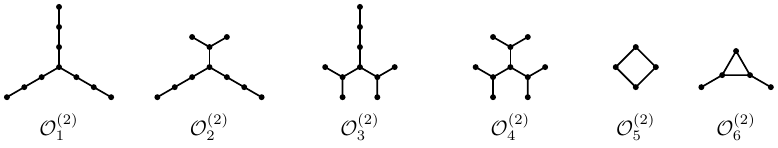}
\caption{\label{obs_path_triangles_2}The minor-obstruction set $\Ocal^{(2)}=\big\{\Ocal^{(2)}_1, \Ocal^{(2)}_2,\Ocal^{(2)}_3, \Ocal^{(2)}_4, \Ocal^{(2)}_5, \Ocal^{(2)}_6 \big\}$.}
\end{figure}

\begin{lemma}
\label{ios3e}
$\obs_{\leqslant_{\mathsf{m}}}(\Ccal^{(2)}) = \Ocal^{(2)}$.
\end{lemma}
\begin{proof}
It is easy to see that for every $Z \in \Ocal^{(2)}$, $Z \notin \Ccal^{(2)}$ and moreover, that $\Ocal^{(2)}$ is a minor-antichain.
Let $G$ be a connected graph such that for every $Z \in \Ocal^{(2)}$, $Z \not\leqslant_{\mathsf{m}} G$. 
First observe that if $|V(G)| \leq 4$ then $G \in \Ccal^{(2)}$.
We call $u \in V(G)$ a \emph{spine vertex} if $|\textsf{cc}(G - v) | \geq 3$ or $\textsf{cc}(G - v) = \{ C, C' \}$ and $|V(C)|, |V(C')| \geq 2$.

Assume that $|V(G)| \geq 5$.
Since $\Ocal^{(2)}_{5} \not\leqslant_{\mathsf{m}} G$, all cycles in $G$ are triangles.
Then, $G$ contains at least one cut-vertex.
We moreover argue that it contains a spine vertex.
Assume it does not.
Let $u \in V(G)$ be a cut-vertex.
Since $u$ is not a spine vertex, we have that $|\textsf{cc}(G - u)| = 2$ and one of the two components is trivial.
Then, since all cycles in $G$ are triangles and $\Ocal^{(2)}_{6} \not\leqslant_{\mathsf{m}} G$, $u$ has a unique neighbour in the non-trivial component, which must be a spine vertex.

Now, let $u, v \in V(G)$ be two spine vertices at distance at least two.
We first claim that every internal vertex $z$ on a $u$-$v$ path $P$ is also a spine vertex.
Suppose that $z$ is not a cut-vertex.
Since every cycle is a triangle, the two neighbours of $z$ on the path $P$ are adjacent, implying that $\Ocal^{(2)}_{6} \leqslant_{\mathsf{m}} G$, a contradiction.
So $z$ is a cut-vertex and since $u$ and $v$ are spine vertices, it implies that $z$ is also a spine vertex.

Next, let $u, v \in V(G)$ be spine vertices at maximum distance in $G$ (it might be that $u = v$) and $P$ be a $u\text{-}v$ path in $G$.
We prove that every connected component $C$ of $G \setminus P$ has at most two vertices and has a unique neighbour on $P$.
This implies that $G \in \Ccal^{(2)}$.
First, suppose that $N(C) \cap V(P) = \{ u \}$ (as $u$ is a spine vertex such a connected component exists).
If $|V(C)| \geq 3$, then by the same argument as in the second paragraph, either $\Ocal^{(1)}_{5} \leqslant_{\mathsf{m}} G$ or $\Ocal^{(1)}_{6} \leqslant_{\mathsf{m}} G$ or $u$ has a unique neighbour $w$ in $C$ that is a spine vertex, contradicting the choice of $u$ and $v$.
So if $u = v$, we can already conclude that $G \in \Ccal^{(2)}$.
Suppose that $u$ and $v$ are distinct vertices.
First observe that $C$ cannot be adjacent to two distinct vertices of $P$.
Indeed, if these neighbours are not adjacent, then $\Ocal^{(1)}_{5} \leqslant_{\mathsf{m}} G$, otherwise $\Ocal^{(1)}_{6} \leqslant_{\mathsf{m}} G$.
So it remains to consider the case where $C$ is adjacent to a unique internal vertex of $P$, say $z$.
If $|V(C)| \geq 3$, then we can observe that we get one of $\Ocal^{(2)}_{1}$, $\Ocal^{(2)}_{2}$, $\Ocal^{(2)}_{3}$, or $\Ocal^{(2)}_{4}$ as a minor, where the center vertex is $z$.
It follows as claimed, that every connected component $C$ of $G \setminus P$ has at most two vertices and has a unique neighbour on $P$.
\end{proof}

Based on \autoref{more_erdos_posa_planar} and lemmata \ref{pbs3e} and \ref{ios3e}, we conclude with the following proposition.

\begin{theorem}
The set $\{\mathscr{P}^{(K_{3}), 1}, \mathscr{P}^{(K_{3}),2}\}$ is a minor-universal obstruction for $\edforest$.
Moreover, $\cobs_{\leqslant_{\mathsf{m}}}(\edforest) = \big\{ \Ccal^{(1)}, \Ccal^{(2)} \big\}$ and $\pobs_{\leqslant_{\mathsf{m}}}(\edforest) = \big\{\Ocal^{(1)},\Ocal^{(2)}\big\}$.
\end{theorem}

\paragraph{Elimination distance to a forest of paths.}

An equivalent way to \eqref{treedepth_parameter} to define the treedepth of a graph is the elimination distance to an edgeless graph, that is $\td = (1, \excl_{\leqslant_{\mathsf{m}}}(K_{2}))\text{-}\td$.
As we have seen in \autoref{treedepth_par} the set $\{\mathscr{P}\}$ containing the sequence $\mathscr{P}$ of paths is a minor-universal obstruction for $\td$.
Moreover, it is easy to see that $\edforest \preceq \td$.

Suppose that we wish to detect a graph parameter that sits between $\edforest$ and $\td$ in the hierarchy.
A natural way to approach this objective would be to consider as the target class
of the elimination distance procedure to be the class of linear forests $\closure{\leqslant_{\mathsf{m}}}{\mathscr{P}}$, i.e. the class of forests of paths.
As $\obs_{\leqslant_{\mathsf{m}}}(\closure{\leqslant_{\mathsf{m}}}{\mathscr{P}}) = \{K_{3}, K_{1,3}\}$, we define the following parameter.
\begin{align}
\edpforest \ &\coloneqq \ (1, \excl_{\leqslant_{\mathsf{m}}}(\{K_{3},K_{1,3})\})\text{-}\td.
\end{align}

We define the minor-parametric graph $\mathscr{Q}$, where 
$\mathscr{Q} = \{\mathscr{Q}_{t} \}_{t\in\mathbb{N}_{\geq 1}}$ are the ternary caterpillars depicted in \autoref{catrexlem}.

\begin{figure}[htbp]
\centering
\includegraphics[width=0.6\linewidth]{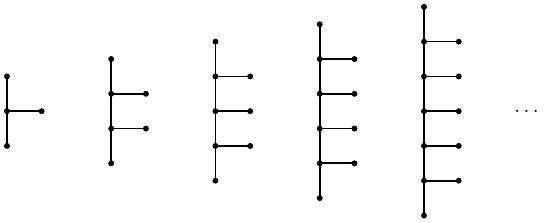}
\caption{\label{catrexlem}The minor-parametric graph $\mathscr{Q} = \langle \mathscr{Q}_1, \mathscr{Q}_2,\mathscr{Q}_3,\mathscr{Q}_4,\mathscr{Q}_5,\dots \rangle$ of ternary caterpillars.}
\end{figure}

Note that $\closure{\leqslant_{\mathsf{m}}}{\mathscr{Q}}$ is the class of caterpillars, that is $\caterpillars \coloneqq \excl_{\leqslant_{\mathsf{m}}}(\{K_{3}, S_{1,1,1}\})$, where $S_{1,1,1} = \Ocal_{1}^{(1)}$ is obtained from $K_{1,3}$ after subdividing once each of the edges of $K_{1,3}$.

To conclude, it is not difficult to observe that the $\mathbf{Z}$-chains that are constructed by considering $K_{3}$ and $K_{1,3}$ are either the two minor-parametric graphs $\mathscr{P}^{(1)}$ and $\mathscr{P}^{(2)}$ of linearly arranged triangles (because of $K_{3}$) or minor-parametric graphs that are equivalent to the minor-parametric graph $\mathscr{Q}$ of ternary caterpillars (because of $K_{1,3}$).
Therefore, by applying \autoref{more_erdos_posa_planar}, only the ternary caterpillars survive the $\lesssim$-minimization operation and we conclude with the following proposition.

\begin{theorem}\label{threqe}
The set $\{ \mathscr{Q} \}$ is a minor-universal obstruction for $\edpforest$.
Moreover, $\cobs_{\leqslant_{\mathsf{m}}}(\edpforest) = \big\{ \caterpillars \big\}$ and $\pobs_{\leqslant_{\mathsf{m}}}(\edpforest) = \big\{ \{K_{3}, S_{1,1,1}\} \big\}$.
\end{theorem}

\paragraph{Elimination distance to planarity.}

Another consequence of the study in \cite{PaulPTW24Obstructions} on obstructions to Erd\H{o}s-P{\'o}sa dualities, is a minor-universal obstruction for $\Hcal\text{-}\td,$ for every proper minor-closed class $\Hcal,$ therefore generalizing \autoref{more_erdos_posa_planar}.

\medskip
To exemplify these results, we once more look into a particular instantiation of $\Hcal$-$\td,$ where $\Hcal = \gplanar.$
We define
\begin{align}
\mathsf{edplanar} \ &\coloneqq \ \gplanar\text{-}\td.
\end{align}

The first step is to understand which minor-parametric graphs are generated by all $\mathbf{Z}$-chains generated by linearly arranging $K_{5}$ and $K_{3,3}$ along a path in non-isomorphic ways.
We define $\mathscr{P}^{(K_{5}), 1}$ and $\mathscr{P}^{(K_{5}), 2}$ to be the minor-parametric graphs corresponding to the two ways $K_{5}$ can be arranged (see the first two graphs in \autoref{fig_td_planar_obs}) and $\mathscr{P}^{(K_{3,3}), 1}$, $\mathscr{P}^{(K_{3,3}), 2},$ and $\mathscr{P}^{(K_{3,3}), 3}$ to be the minor-parametric graphs corresponding to the three ways $K_{3,3}$ can be arranged (see the last three graphs in \autoref{fig_td_planar_obs}).
Moreover, recall the definition of the minor-parametric graphs $\mathscr{D}^{(1, 0)}$ (see \autoref{fig_torus_grid}) and $\mathscr{D}^{(0, 1)}$ (see \autoref{fig_projective_grid}).

\begin{figure}[htbp]
\centering
\includegraphics[width=0.6\linewidth]{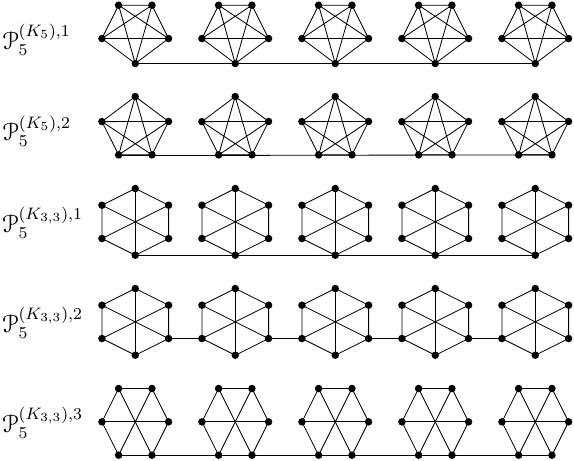}
\caption{\label{fig_td_planar_obs} An instance of each of the five minor-parametric graphs obtained from the two non-isomorphic ways to linearly arrange $K_{5}$ along a path and the three non-isomorphic ways to linearly arrange $K_{3,3}$ along a path.}
\end{figure}

A consequence of the general result in \cite{PaulPTW24Obstructions} is the following.

\begin{proposition}
The set
\begin{align*}
\{ \mathscr{D}^{(1, 0)}, \mathscr{D}^{(0, 1)}, \mathscr{P}^{(K_{5}), 1}, \mathscr{P}^{(K_{5}), 2}, \mathscr{P}^{(K_{3,3}), 1}, \mathscr{P}^{(K_{3,3}), 2}, \mathscr{P}^{(K_{3,3}), 3}\}
\end{align*}
is a minor-universal obstruction for $\mathsf{edplanar}$.
\end{proposition}

\subsubsection{More on the hierarchy of minor-monotone parameters}

In this subsection, we have considered two types of parameters, namely apex parameters and elimination distance parameters.
This leads to two partial hierarchies of parameters.
For the apex parameters, we have
\begin{align}
\mathsf{apexplanar} \ \preceq \ \apexouter \ \preceq \ \fvs \ \preceq \ \vc
\end{align}
and the relationship between their minor-parametric obstructions is
\begin{align}
\nonumber\big\{\{P_{3}\}\big\} \ \leqslant^{**}_{\mathsf{m}} \ \big\{ &\{P_4,K_{1,3}\}\big\}\\
\nonumber\leqslant^{**}_{\mathsf{m}} \ \big\{ &\{P_{5}, K_{1,4}, K_{1,3}^s \},\\
\nonumber&\{ P_{6}, Q_{2}, K_{1,3}^{ds}, K_{1,3}^{2s}, K_{1,4}, K_{4}, C_{5}, K_{3}^{de}, K_{3}^{es} \} \big\}\\
\leqslant^{**}_{\mathsf{m}} \ \big\{ &\obs_{\leqslant_{\mathsf{m}}}(\Gcal_{\mathsf{toroidal}}),\\
\nonumber&\Ocal_{\mathsf{projective}},\\
\nonumber&\{ P_{6}, Q_{2}, K_{1,5}, K_{1,4}^{\mathsf{s}}, K_{1,3}^{\mathsf{ds}}, K_{1,3}^{\mathsf{2s}} \},\\
\nonumber&\obs_{\leqslant_{\mathsf{m}}}(\closure{\leqslant_{\mathsf{m}}}{\mathscr{K}^{(K_{3,3})}}) \big\}.
\end{align}

For the elimination distance parameters, we have that
\begin{align}
\mathsf{edplanar} \ \preceq \ \edforest \ \preceq \ \edpforest \ \preceq \ \td
\end{align}
and for the corresponding minor-parametric obstructions
\begin{align}
\big\{\{K_3,K_{1,3}\}\big\} \ &\leqslant^{**}_{\mathsf{m}} \ \big\{\{K_{3},S_{1,1,1}\}\big\} \ \leqslant^{**}_{\mathsf{m}} \ \big\{\Ocal^{(1)},\Ocal^{(2)}\big\}\\
\nonumber&\leqslant^{**}_{\mathsf{m}} \ \pobs_{\leqslant_{\mathsf{m}}}(\mathsf{edplanar}).
\end{align}

\subsection{Biconnected variants}
\label{blosckldkior}

Our next and last step in the realm of the minor relation is to focus on variants of already known parameters whose decomposition scheme applies to the $2$-blocks of the input graph.
For simplicity, we use the term \emph{blocks} instead of $2$-blocks.

\subsubsection{Biconnected pathwidth}\label{bicpw}

The first parameter we consider is biconnected pathwidth, denoted by $\bipw$, which we have already seen as part of our presentation in \autoref{diidnidialdedfs}.
Recall that the biconnected pathwidth of a graph $G$ is defined so that
\begin{align}
\bipw(G) \ &\coloneqq \ \min\{ \pw(B) \mid B\text{ is a block of }G\}.\label{bi_pathwidth_parameter}
\end{align}

We define the minor-parametric graph $\mathscr{T}^{a} = \langle \mathscr{T}_k^{a} \rangle_{k \in \mathbb{N}_{geq 2}}$, obtained from the complete ternary tree $\mathscr{T}_k$, defined in \autoref{pathwidth_def}, by adding a new vertex and making it adjacent to the leaves of $\mathscr{T}_k$.
See \autoref{fig_bpw} for an illustration.

\begin{figure}[htbp]
\centering
\includegraphics[width=0.8\linewidth]{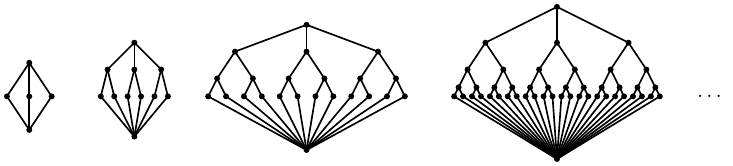}
\caption{\label{fig_bpw} The minor-parametric graph $\mathscr{T}^a = \langle \mathscr{T}^a_2, \mathscr{T}^a_3, \mathscr{T}^a_4, \mathscr{T}^a_5,\ldots\rangle$.}
\end{figure}

We also define the minor-parametric graph $\mathscr{T}^{a*} = \langle \mathscr{T}^{a*}_k \rangle_{k\in\mathbb{N}_{\geq 2}}$ so that $\mathscr{T}^{a*}_k$
is obtained from $\mathscr{T}^{a}_k$ after taking its dual and, for each pair of 
double edges that occur, subdividing one of them once.
See \autoref{fig_dual_Tk} for an illustration.

\begin{figure}[htbp]
\centering
\includegraphics[width=0.8\linewidth]{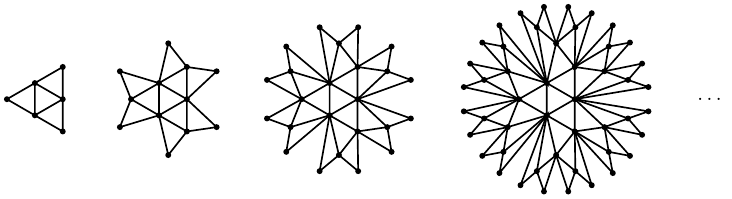}
\caption{\label{fig_dual_Tk} The minor-parametric graph $\mathscr{T}^{a*} = \langle \mathscr{T}^{a*}_2, \mathscr{T}^{a*}_3, \mathscr{T}^{a*}_4, \mathscr{T}^{a*}_5, \ldots \rangle$.}
\end{figure}

Note that $\tw(\mathscr{T}^{a}_k) = \tw(\mathscr{T}^{a*}_k) = 2$, while for the biconnected graphs $\mathscr{T}^a_k$ and $\mathscr{T}^{a*}_k$, it holds that 
$\bipw(\mathscr{T}^a_k) = Ω(k)$ and $\bipw(\mathscr{T}^{a*}_k) = Ω(k)$.

In addition, note that $\mathscr{T}^{a}$ is a minor-omnivore of the class of apex forests $\gfapex$, i.e., the graphs that have \emph{apex number} at most one to the class $\excl_{\leqslant_{\mathsf{m}}}(K_{3})$ of acyclic graphs (see e.g., \cite[Lemma 2.2]{dang2018minors}).
Also $\mathscr{T}^{a*}$ is a minor-omnivore of the class $\gouterplanar = \excl_{\leqslant_{\mathsf{m}}}(\{K_{4}, K_{2,3}\})$ of outerplanar graphs (see e.g., \cite[Lemma 2.4]{dang2018minors}).
One may show (see e.g., \cite[Lemma 2.1]{dang2018minors} or \cite{DinneenCF0fForbidden}) that $\obs_{\leqslant_{\mathsf{m}}}(\closure{\leqslant_{\mathsf{m}}}{\mathscr{T}^{a}}) = \{S_3, 2 \cdot K_{3}, K_{4}\}$, where $S_3 = \mathscr{T}^{a*}_2$ is the octahedron $K_{2,2,2}$ minus one triangle (in fact $\closure{\leqslant_{\mathsf{m}}}{\mathscr{T}^{a}}$ consists of the duals of the outerplanar graphs).
See \autoref{biconnected_pathwidth_pobs} for an illustration of these graphs.

\begin{figure}[htbp]
\centering
\includegraphics[width=0.41\linewidth]{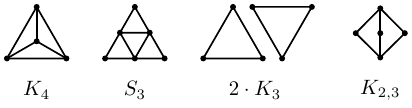}
\caption{\label{biconnected_pathwidth_pobs}The graphs $K_4, S_3, 2 \cdot K_3, K_{2,3}$.}
\end{figure}

Dang and Thomas \cite{dang2018minors} as well as Huynh,  Joret,  Micek, and  Wood \cite{HuynhJMW20Seymour} independently proved that there exists a function $f \colon \nton$ such that every 2-connected graph $G$ where $\pw(G) \geq f(k)$ contains as a minor either $\mathscr{T}_k^{a}$ or $\mathscr{T}_k^{a*}$.
We conclude with the following proposition

\begin{theorem}
The set $\{\mathscr{T}^a, \mathscr{T}^{a*}\}$ is a minor-universal obstruction for $\bipw$.
Moreover, it holds that $\cobs_{\leqslant_{\mathsf{m}}}(\bipw) = \{\gfapex, \gouterplanar\}$ and $\pobs_{\leqslant_{\mathsf{m}}}(\bipw) = \big\{ \{K_4, S_3, 2\cdot K_3\}, \{K_4, K_{2,3}\} \big\}$.
\end{theorem}


Marshall and Wood \cite{MarshallW15circumference} defined the minor-monotone parameter $\mathsf{g} \colon \gall\to\mathbb{N}$ such that $\mathsf{g}(G)$ is the minimum $k$ for which there exists a function $h \colon \nton$ such that every $\mathsf{g}(G)$-connected $G$-minor-free graph has pathwidth at most $h(G)$. 
It was conjectured in \cite{MarshallW15circumference} that 
$\obs_{\leqslant_{\mathsf{m}}}(\{G \in \gall \mid \mathsf{g}(G) \leq 2\}) = \{K_4, S_3, 2\cdot K_3, K_4, K_{2,3}\}$.
This conjecture was proved in \cite[Section 2]{dang2018minors} as a consequence of the results of \cite{dang2018minors} and \cite{HuynhJMW20Seymour}.
The proof can be seen as a consequence of the fact that $\pobs_{\leqslant_{\mathsf{m}}}(\bipw) = \big\{ \{K_4, S_3, 2\cdot K_3\}, \{K_4, K_{2,3}\} \big\}$.

\subsubsection{Block treedepth}
\label{subsec_block_tree}

Consider a minor-closed class $\Hcal$ whose minor-obstruction set consists of $2$-connected graphs.
In \autoref{evjvhaiendndodbabsed} we defined the parameter $\Hcal\text{-}\bed =(2, \Hcal)\text{-}\td$ as the block elimination distance of a graph $G$ to the class $\Hcal.$
The parameter $\Hcal\text{-}\bed$ was first introduced in \cite{DinerGST22Block}.

In this subsection, we focus on a particular simple instantiation of $\Hcal\text{-}\bed$ where $\Hcal$ is the class $\gforest$ of acyclic graphs.
We define the \emph{block treedepth} of a graph $G$, denoted by $\btd$ as follows.
\begin{align}
\btd(G) \ \coloneqq \ \gforest\text{-}\bed.\label{block_treedepth_parameter}
\end{align}

The parameter $\btd$ was considered by Huynh,  Joret,  Micek,  Seweryn, and  Wollan \cite{HuynhJMSW22Excluding}.

We define the minor-parametric graph $\mathscr{L} = \langle \mathscr{L}_k \rangle_{k\in\mathbb{N}_{\geq 2}}$ where for every $k \geq 2,$ $\mathscr{L}_k$ is the \emph{$k$-ladder}: the Cartesian product of the path $P_k$ and $K_2.$
See \autoref{fig_ladder} for an illustration.

\begin{figure}[htbp]
\centering
\includegraphics[width=0.6\linewidth]{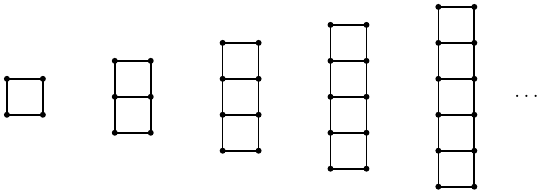}
\caption{\label{fig_ladder} The minor-parametric graph $\mathscr{L} = \langle \mathscr{L}_2, \mathscr{L}_3, \mathscr{L}_4, \mathscr{L}_5, \ldots \rangle$ of ladders.}
\end{figure}

According to \cite{HuynhJMSW22Excluding}, there exists a function $f \colon \nton$ such that every graph $G$ where $\btd(G) \geq  f(k)$ contains $\mathscr{L}_k$ as a minor. Moreover, it is easy to verify that $\btd(\mathscr{L}_k) = \Omega(\log k)$.
The graph class $\mathcal{\ladderminors} \coloneqq \closure{\leqslant_{\mathsf{m}}}{\mathscr{L}}$ of minors of ladders naturally arises in different contexts of graph theory.
Indeed, it contains all graphs with mixed search number at most two (see \cite{TakahashiYK95Mixed, TakahashiUK95b}).
Following the terminology of \cite{HarveyW17Parameters}, these graphs also consist of the graphs with \textsl{Cartesian path product number} at most two.
Finally, following the terminology of \cite{MescoffPT22Themixed}, these graphs also admit a loose path decomposition of width at most two.
Most interestingly, Takahashi,  Ueno,  and Kajitani \cite{TakahashiUK95b} identified the minor-obstruction set $\Ocal^{L} = \obs_{\leqslant_{\mathsf{m}}}(\mathcal{\ladderminors})$ that consists of the 36 graphs depicted in \autoref{lasder_j}.

\begin{figure}[htbp]
\centering
\includegraphics[width=1\linewidth]{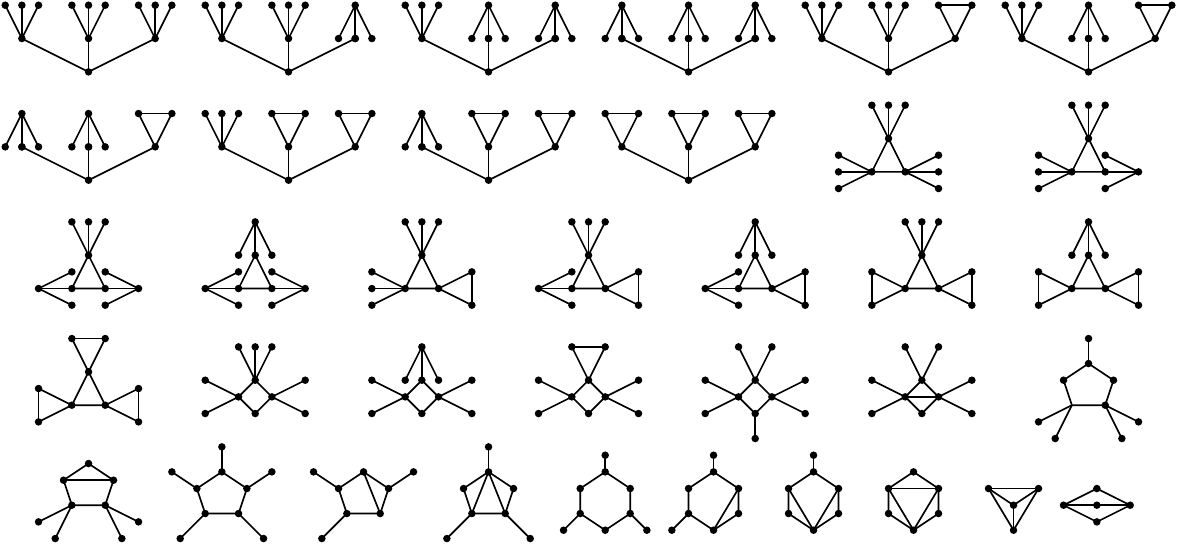}
\caption{\label{lasder_j}The graphs in the minor-obstruction set $\Ocal^{L} = \obs_{\leqslant_{\mathsf{m}}}(\mathcal{\ladderminors})$.}
\end{figure}

We can now conclude with the following proposition.

\begin{proposition}\label{blochdtheorem}
The set $\mathfrak{L} = \{ \mathscr{L} \}$ is a minor-universal obstruction for $\btd$. Moreover, $\cobs_{\leqslant_{\mathsf{m}}}(\btd) = \{ \ladderminors \}$ and $\pobs_{\leqslant_{\mathsf{m}}}(\btd) = \big\{ \Ocal^{L} \big\}$.
\end{proposition}

\subsubsection{From treedepth to pathwidth and treewidth}
 An interesting question concerning graph parameters is whether it is possible  
to construct hierarchies of parameters lying between already known ones.  
In \Cref{treewidth_sec} we have already seen how treewidth can be expressed in terms of clique-sum closures (see {\eqref{treewidth_parameter}}).  
Treedepth and block treedepth, defined in \autoref{treedepth_par} and \autoref{subsec_block_tree} respectively, can also be described using elimination distances.  
An important step toward finding a hierarchy of parameters between treedepth and treewidth was made by  Rambaud in \cite{Rambaud2025Excluding}, who introduced the parameter $k$-tree-depth by combining both the clique-sum and elimination-distance frameworks.

Let $G_1$ and $G_2$ be graphs, and let $C_1$ and $C_2$ be, possibly
empty, cliques of $G_1$ and $G_2$, respectively, of the same size.
Recall that, a \emph{clique-sum} of $G_1$ and $G_2$ is any graph obtained from the
disjoint union of $G_1$ and $G_2$ by identifying $C_1$ and $C_2$ via a
bijection between their vertex sets, and possibly deleting some edges
whose endpoints both belong to the identified clique.
If $|C_1|=|C_2|=r$, we call it an \emph{$r$-clique-sum}.

For $k\in\Nbbb_{\ge 1}$, a \emph{$(<k)$-clique-sum} is an
$r$-clique-sum for some integer $r$ with $0\le r<k$.
Note that $0$-clique-sum means taking the disjoint union.

Given $k\in\Nbbb_{\ge 1}$ and a class $\Ccal$ of graphs, the
\emph{$(<k)$-clique-sum closure} of $\Ccal$, denoted by $\Ccal^{(k)}$,
is the smallest class of graphs that contains $\Ccal$ and is closed
under $(<k)$-clique-sums. Equivalently, $\Ccal^{(k)}$ is the class of
all graphs that can be obtained from graphs in $\Ccal$ by repeatedly
taking $(<k)$-clique-sums.
We also set $\Ccal^{(0)}\coloneqq\Ccal$.

With these definitions, $\Ccal^{(1)}$ is the closure of $\Ccal$ under
disjoint unions, i.e. the connectivity closure of $\Ccal$, defined
in~\eqref{conn_clo}, and $\Ccal^{(2)}$ is the \emph{biconnectivity
closure} of $\Ccal$, that is, the class of all graphs whose $2$-blocks
belong to $\Ccal$.

Recall that in \Cref{apeplicx} we defined, for a class $\Ccal$ of
graphs, its \emph{apex extension} as
$\mathbf{A}(\Ccal)
  \coloneqq
  \Ccal
  \cup
  \{\,G \mid \exists v\in V(G) \text{ such that } G-v\in\Ccal\,\}$.
For $k\in\Nbbb_{\ge 1}$, we recursively define the graph classes
$\Gcal_t^{(k)}$, for $t\in\Nbbb_{\ge 0}$, as follows:
$\Gcal_0^{(k)}\coloneqq \{K_0\}$, where $K_0$ denotes the empty graph,
and, for every $t\ge 1$,
$$\Gcal_t^{(k)}
  \coloneqq
  \bigl(\mathbf{A}(\Gcal_{t-1}^{(k)})\bigr)^{(k)}.$$
The \emph{$k$-treedepth} of a graph $G$ is then defined by
$\td_k(G)
  \coloneqq
  \min\{\,t\in\Nbbb_{\ge 0}\mid G\in\Gcal_t^{(k)}\,\}.$

In words, for $t\ge 1$, a graph $G$ has $k$-treedepth at most $t$ if
and only if $G$ belongs to the $(<k)$-clique-sum closure of the class
of graphs $H$ such that either $\td_k(H)\le t-1$, or there exists a
vertex $v\in V(H)$ with $\td_k(H-v)\le t-1$.

The definition above yields a hierarchy of parameters:
\[
\mathsf{td}_1 = \mathsf{td} \le \mathsf{td}_2 \le \cdots \le \mathsf{td}_\infty,
\]
where $\mathsf{td}_1=\td$, $\mathsf{td}_2=\btd$, and $\mathsf{td}_\infty = \mathsf{tw} + 1$.  
For each tree $T$, we define $\mathscr{H}^T = \{T \square P_k \mid k \in \mathbb{N}\},$
where ``$\square$'' denotes the Cartesian product of graphs.  
Furthermore, for every $k\in\Nbbb_{\geq 1}\cup\{\infty\}$ we define 
$
\mathfrak{H}^{(k)}=\{\mathscr{H}^T \mid T \text{ is a tree on $k$ vertices}\}.
$
 In~\cite{Rambaud2025Excluding}, the following result was established.

\begin{proposition}\label{the_k_td}
For every $k\in\Nbbb_{\geq 1}\cup\{\infty\}$, the set $\mathfrak{H}^{(k)}$ is a minor-universal obstruction for $\td_{k}$. 
\end{proposition}

In~\cite{Rambaud2025Excluding}, Rambaud also introduced, using $k$-clique-sums and path decompositions, a hierarchy of parameters between treedepth and pathwidth:
\[
\mathsf{pd}_1 \le \mathsf{pd}_2 \le \cdots \le \mathsf{pd}_\infty.
\]
Here $\mathsf{pd}_1=\td$ and $\mathsf{pd}_\infty = \mathsf{pw} + 1$.  
Recall that $\mathfrak{T}=\{ \mathscr{T}\}$, where $\mathscr{T}$ is the minor-parametric graph of complete ternary trees of depth $k$ (see \autoref{fig_ternary_trees}).  
A universal  obstruction characterization for the ${\sf pd}_{k}$ hierarchy as proved in \cite{Rambaud2025Excluding} as follows.

\begin{proposition}\label{the_k_pd}
For every $k\in\Nbbb_{\geq 1}\cup\{\infty\}$, the set $\mathfrak{H}^{(k)}\cup\mathfrak{T}$ is a minor-universal obstruction for ${\sf pd}_{k}$. 
\end{proposition}

\subsubsection{More on the hierarchy of minor-monotone parameters}

Observe that biconnected pathwidth is by definition sandwiched between treewidth and pathwidth.
\begin{align}
\tw \ \preceq \ \bipw \ \preceq \ \pw
\end{align}
and the corresponding parametric obstructions are ordered as follows
\begin{align}
\big\{\{K_3\}\big\} \ \leqslant^{**}_{\mathsf{m}} \ \big\{\{K_4,S_3,2\cdot K_3\},\{K_4,K_{2,3}\}\big\} \ \leqslant^{**}_{\mathsf{m}} \ \big\{\{K_{5},K_{3,3}\}\big\}.
\end{align}

Another interesting observation is that the two parameters $\bipw$ and $\btd$ that we considered in this subsection are not $\preceq$-comparable.
Indeed, this follows easily by observing that their minor-universal obstructions \autoref{threqe} and \autoref{blochdtheorem} are $\lesssim^{*}$-incomparable.
Indeed, $\{\mathscr{L}\} \not\lesssim^* \{\mathscr{T}^a, \mathscr{T}^{a*}\}$ since $\mathscr{L} \not\lesssim \mathscr{T}^{a*}$ as we need to delete many vertices from a ladder in order to make it acyclic and $\{\mathscr{T}^a, \mathscr{T}^{a*}\} \not\lesssim^* \{\mathscr{L}\}$ since $\mathscr{T}^{a} \not \lesssim\mathscr{L}$ as not every outerplanar graph is a minor of a ladder and $\mathscr{T}^{a*} \not\lesssim \mathscr{L}$ as not every dual of an outerplanar graph is a minor of a ladder.

\subsection{Hierarchy of minor-monotone parameters}

We conclude this section with \autoref{fig_minor_closed}, which is only a small fraction of the quasi-ordering of minor-monotone parameters ordered by $\preceq$ that corresponds to the minor-monotone parameters we considered in this section.

\begin{figure}[htbp]
\centering
\includegraphics[width=\linewidth]{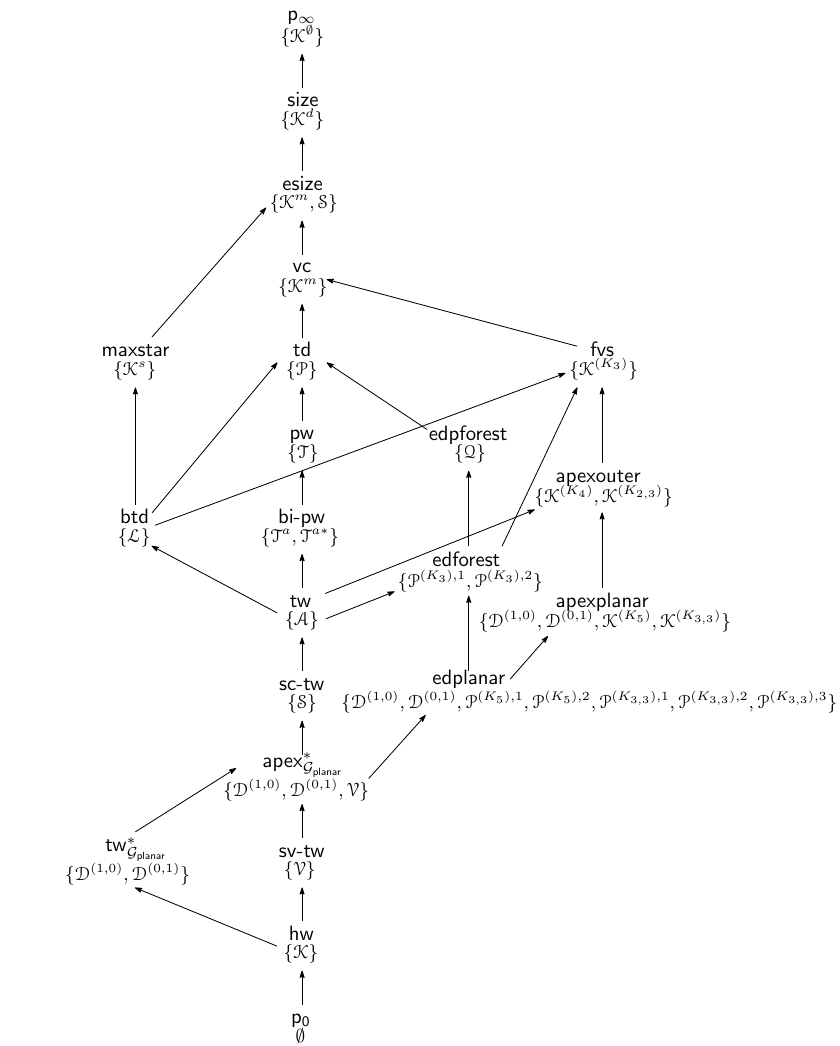}
\caption{\label{fig_minor_closed} The hierarchy of minor-monotone parameters, each accompanied with some minor-universa l obstruction of it.}
\end{figure}

\section{Obstructions of immersion-monotone parameters}
\label{ommsers}

In this section, we consider graphs that may contain {parallel} edges but no loop.
For this, given a graph $G$, we see its edge set as a multiset. 
\rev{page 35, line 55: “graphs with multiple edges”. Very confusing. I would simply state that in this section we allow {parallel} edges in graphs but we do not allow loops. Same comment for footnote 13 in page 40.}
We also use $\gall^{\mathsf{e}}$ in order to denote the class of all graphs with {parallel} edges {but not loop}.
Accordingly, the graph parameters we examine in these sections are functions $\p \colon \gall^{\mathsf{e}} \to \Nbbb,$ i.e., that map graphs with {parallel} edges {and no loop} to non-negative integers.

Let $G$ be a graph and let $e_{1} = xy$ and $e_{2} = yz$ be two edges of $G$ where $y$ is a common endpoint for both.
The result of the operation of \emph{lifting} the edges $e_{1}$ and $e_{2}$ in $G$ is the graph obtained from $G$ if we first remove $e_{1}$ and $e_{2}$ from $G$ and then add the edge $xz.$
As we deal with {graphs with parallel edges}, we agree that if any of $e_{1}$ and $e_{2}$ has multiplicity bigger than two, then its removal reduces its multiplicity by one.
Also if the edge $xz$ already exists, then we simply increase its multiplicity by one.
Moreover, if $e_{1}$ and $e_{2}$ are parallel edges then the lifting operation removes both of them without adding a new edge.
That is because we consider {(loopless)} {graphs} with  {parallel} edges and therefore we suppress any loop that may arise from the lifting operation.

We define the \textsl{immersion} relation on the set of graphs with multiple edges as follows.
We say that a graph $H$ is an \emph{immersion} of a graph $G$, denoted by $H \leqslant_{\mathsf{i}} G,$ if $H$ can be obtained from a subgraph of $G$ after a (possibly empty) sequence of lifting pairs of edges.
Observe that if $H$ is a topological minor of $G$ then $H$ is also an immersion of $G$ (but not vice versa).

Robertson and Seymour \cite{RobertsonS10GMXXIII} demonstrated that $\leqslant_{\mathsf{i}}$ is a well-quasi-ordering on $\gall^{\mathsf{e}}$.
In light of this result, the immersion relation inherits all the beneficial properties that arise from well-quasi-ordering, similarly to its counterpart, the minor relation.

\subsection{A few ``trivial'' immersion-monotone parameters}\label{lesstrcas}

As with the minor-monotone parameters, we can begin by addressing the trivial cases (which, in this context, are somewhat less ``trivial'' than in the minor case).

\paragraph{The parameters $\p_{\infty}$ and $\size$.}

Our first observation is that, by the same arguments as in \autoref{wadfdjgjingngup}, the minor-universal obstructions for $\p_{\infty}$ and $\size$ also serve as immersion-universal obstructions.

In particular, $\{\mathscr{K^{\emptyset}}\}$ is an immersion-universal obstruction for $\p_{\infty}$ and $\{\mathscr{K}^{\mathsf{d}}\}$ is an immersion-universal obstruction for $\size$.

\paragraph{Edge-admissibility.}

The \emph{edge-admissibility} denoted by $\eadm$, is defined so that for every graph $G \in \gall^{\mathsf{e}},$
\begin{align}
\eadm(G) \ \coloneqq \ \min\{ k \in \Nbbb \mid \ &\text{there exists an ordering }v_{1}, \ldots, v_{n}\text{ of }V(G) \text{~such that,~} \nonumber\\
&\text{for every $i\in[2,n]$, there exist at most }k\text{ edge-disjoint }\\
\nonumber&\mbox{paths~} \text{from }v_{i}\text{ to }v_{1}, \ldots, v_{i-1}\text{ in }G \}.
\end{align}

\begin{figure}[htbp]
\centering
\includegraphics[width=0.65\linewidth]{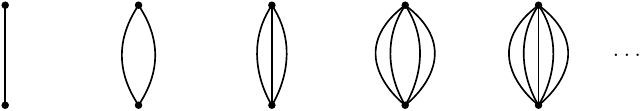}
\caption{\label{pumpkin_sequence}The immersion-parametric graph $ \Theta  = \langle  \Theta_1, \Theta_2, \Theta_3, \Theta_4, \Theta_5 \ldots \rangle$ of $t$-pumpkins.}
\end{figure}

We define the immersion-parametric graph $\Theta \coloneqq \langle \Theta_{t} \rangle_{t \in \mathbb{N}}$, where $ \Theta _{t}$ is the graph on two vertices and $t$ parallel edges, for every $t \in \Nbbb$.
These graphs also known as \emph{$t$-pumpkins} are illustrated in \autoref{pumpkin_sequence}.
Observe that $\Ccal^{\Theta} \coloneqq \closure{\leqslant_{\mathsf{i}}}{\Theta} = \{ \Theta_{t} \mid t\in\mathbb{N}_{\geq 1}\} \cup \{K_{1}, K_0 \}$.
It is proved in \cite{LimniosPPT20Edge} that if $\eadm(G) \geq 2k$ then $\Theta_{k+1} \leqslant_{\mathsf{i}} G$.
It is also straightforward to see that $\eadm(\Theta_{k+1}) \geq k+1$, for every $k \in \Nbbb.$
Therefore, we obtain the following statement.

\begin{proposition}
The set $\{ \Theta \}$ is an immersion-universal obstruction for $\eadm$. Moreover, $\cobs_{\leqslant_{\mathsf{i}}}(\eadm) = \{\Ccal^{\Theta }\}$ and $\pobs_{\leqslant_{\mathsf{i}}}(\eadm) = \big\{ \{3\cdot K_{1}\} \big\}$.
\end{proposition}

\paragraph{Maximum degree.}

A (simple) example of an immersion-monotone parameter is the \emph{degree} $\Delta.$
For every vertex $v \in V(G)$ let $E_{v}$ denote the set of edges of $G$ incident to $v.$
Then, for every graph $G \in \gall^{\mathsf{e}},$
\begin{align}
\Delta(G) \ &\coloneqq \ \max_{v \in V(G)} |E_{v}|.
\end{align}

Recall that $\mathscr{K}^{\mathsf{s}} = \langle K_{1,k} \rangle_{k\in\mathbb{N}}$ is the immersion-parametric graph of stars.
We define the immersion-closed class $\Ccal^{S}_{\mathsf{i}}$ of star forests with at most one component on at least three vertices.

It is straightforward to see that a graph $G$ has bounded degree $\Delta(G)$ if and only if it excludes a star and a pumpkin as an immersion.
Moreover, $\obs_{\leqslant_{\mathsf{i}}}(\Ccal^{S}_{\mathsf{i}}) = \{ \Theta_{2}, 2 \cdot P_{3}, P_{4} \}.$
Therefore,

\begin{proposition} The set $\{ \Theta, \mathscr{K}^{\mathsf{s}} \}$ is an immersion-universal obstruction for $\Delta.$
Moreover, $\cobs_{\leqslant_{\mathsf{i}}}(\Delta) = \big\{ \Ccal^{\Theta},\Ccal^{S}_{\mathsf{i}} \big\}$ and $\pobs_{\leqslant_{\mathsf{i}}}(\Delta) = \big\{ \{3 \cdot K_{1}\}, \{ \Theta_{2}, 2 \cdot P_{3}, P_{4} \} \big\}.$
\end{proposition}

\paragraph{The order of a graph.}



Let us now consider the parameter $\esize$ which we already studied in \autoref{wadfdjgjingngup} in the context of the minor relation.
Recall that $\mathscr{K}^{\mathsf{m}} = \langle k\cdot K_{2} \rangle_{k\in\mathbb{N}}$ and that $\Ccal^{M} = \closure{\leqslant_{\mathsf{m}}}{\mathscr{K}^{\mathsf{m}}} = \closure{\leqslant_{\mathsf{i}}}{\mathscr{K}^{\mathsf{m}}}$ is the set of forests of maximum degree one.
Observe that $\obs_{\leqslant_{\mathsf{i}}}(\Ccal^{M}) = \{P_{3}, \Theta_2 \}$.
It is moreover easy to see that $\esize \sim Δ + \p_{\mathscr{K}^{\mathsf{m}}}$ and therefore we obtain the following proposition.

\begin{proposition} The set $\{ \Theta, \mathscr{K}^{\mathsf{s}}, \mathscr{K}^{\mathsf{m}} \}$ is an immersion-universal obstruction for $\esize.$
Moreover, $\cobs_{\leqslant_{\mathsf{i}}}(\esize) = \big\{ \Ccal^{\Theta},\Ccal^{S}_{\mathsf{i}}, \Ccal^{M} \big\}$ and $\pobs_{\leqslant_{\mathsf{i}}}(\esize) = \big\{ \{3 \cdot K_{1}\}, \{ \Theta_{2}, 2 \cdot P_{3}, P_{4} \}, \{ P_{3}, \Theta_{2} \} \big\}.$
\end{proposition}

We stress that $\esize$ has different parametric obstructions for immersions and minors which is an indication of the different nature of the two quasi-orderings.

\subsection{Tree-cutwidth}\label{treewishhwith}

Tree-cutwidth is an immersion-monotone parameter that was introduced by Wollan \cite{Wollan15Thestructure} as an analogue of treewidth, in the context of the immersion relation.

A \emph{near-partition} of a set $X$ is a family of (possibly empty) subsets $X_{1}, \dots, X_{k}$ of $X$ such that 
$\bigcup_{i=1}^{k} X_{i} = X$ and $X_{i} \cap X_{j} = \emptyset$ for every $i\neq j$.
A \emph{tree-cut decomposition} of a graph $G$ is a pair $\mathcal{T} = (T, {\mathcal{X}})$ such that $T$ is a forest and ${\mathcal{X}} = \{X_{t} \mid t \in V(T)\}$ is a near-partition of the vertices of $V(G)$.
Furthermore, we require that if $T_1, \ldots, T_r$ are the connected components of $T$, then $\bigcup_{t \in V(T_i)} X_t$, for $i \in [r]$, are exactly the vertex sets of connected components of $G$.
We call the elements of $V(T)$ \emph{nodes} of $T$.
For an edge $e = uv \in E(T)$ we write $E_{e}$ for the set of edges of $G$ that have one endpoint in $\bigcup_{t \in V(T_{v})} X_{t}$ and one endpoint in $\bigcup_{t \in V(T_{u})} X_{t}$, where $T_{u}$ and $T_{v}$ are the connected components of $T - e$ that contain $u$ and $v$ respectively.
We define the \emph{adhesion} of $e$ to be $\adh_{\mathcal{T}}(e) = |E_{e}|$.
The adhesion of $e$ is \emph{thin} if it has at most two edges and \emph{bold} otherwise.

Let $G$ be a graph with a tree-cut decomposition $\mathcal{T}=(T, {\mathcal{X}})$.
For every $t \in V(T)$, we define
\begin{align}
w(t) \ &\coloneqq \ |X_{t}| + \big|\{t'\in N_{T}(t)\mid\adh_\mathcal{T}(\{t,t'\})~\text{is bold}\}\big|.\label{kkfkadjg}
\end{align}
We then set
\begin{align}
\width(\mathcal{T}) \ &\coloneqq \ \max\big\{\max_{e\in E(T)}|\adh_{\Tcal}(e)|,\max_{t\in V(T)}w(t)\big\}.\label{aclop}
\end{align}

The \emph{tree-cutwidth} of a graph $G \in \gall^{\mathsf{e}},$ denoted by  $\tcw(G),$ is then defined as follows.
\begin{align}
\tcw(G) \ &\coloneqq \ \min\{ \width(\Tcal) \mid \Tcal\text{ is a tree-cut decomposition of }G.\}
\end{align}


Let $\mathscr{W} \coloneqq \langle \mathscr{W}_{k} \rangle_{k \geq 3}$ be the wall of height $k$, that is the graph obtained from a $(k\times 2k)$-grid after removing a perfect matching that sits between its ``horizontal paths'' and then removing all occurring vertices of degree one.
See \autoref{wsalls} for an illustration.

\begin{figure}[htbp]
\centering
\includegraphics[width=0.95\linewidth]{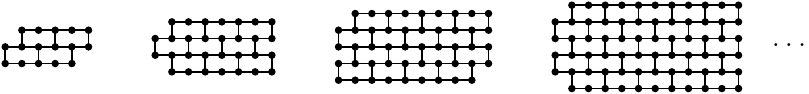}
\caption{\label{wsalls}The immersion-parametric graph $\mathscr{W} = \langle \mathscr{W}_3, \mathscr{W}_4, \mathscr{W}_5, \mathscr{W}_6, \ldots \rangle$ of walls.}
\end{figure}

Wollan \cite{Wollan15Thestructure} proved that there exists a function $f \colon \nton$ such that every graph $G$ with $\tcw(G) \geq f(k)$ contains $\mathscr{W}_{k}$ as an immersion, for every $k \geq 3.$
It is therefore implied that $\tcw \preceq \p_{\mathscr{W}}$.
Moreover, it is known that $\tw \preceq \tcw$ (see e.g., \cite{Wollan15Thestructure}).

As already mentioned, since $\Delta(\mathscr{W}_{k}) \leq 3$, if a graph contains $\mathscr{W}_{k}$ as an immersion, it also contains it as a topological minor and therefore also as a minor.
This implies that since $\p_{\mathscr{W}} \preceq \tw,$ then $\p_{\mathscr{W}} \preceq \tcw$.
We conclude that $\p_{\mathscr{W}} \sim \tcw$.
Note that if a graph is an immersion of the wall $\mathscr{W}_{k}$, then it is planar and has maximum degree three.
In other words, it belongs to the class $\Pcal^{(\leq 3)}$ of  \emph{planar subcubic} graphs.
Also, note that every graph in $\Pcal^{(\leq 3)}$ is a topological minor of some wall in $\mathscr{W}$, therefore $\mathscr{W}$ is an immersion-omnivore of $\Pcal^{(\leq 3)}$.
Let $P_{3}^{2,2}$ be the graph obtained if we duplicate the two edges of $P_{3}$, $K_{1,3}^{2}$ be the graph obtained if we duplicate one of the edges of $K_{1,3},$ and $K^{3}_{1, 2}$ be the graph obtained from $K_{1,2}$ by replacing one edge with three parallel edges.
See \autoref{tree_cutwidth_pobs} for an illustration of these graphs.
 
\begin{figure}[htbp]
\centering
\includegraphics[width=0.6\linewidth]{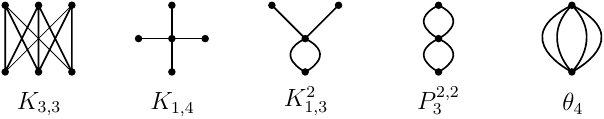}
\caption{\label{tree_cutwidth_pobs}The graphs $K_{3,3}, K_{1,4}, K_{1,3}^{2}, P_{3}^{2,2}, K^{3}_{1,2}, \Theta_{4}$.}\
\end{figure}

Suppose now that a graph $G$ does not belong to $\Pcal^{(\leq 3)}$.
Then, it either has a vertex of degree larger than three or it is not planar.
In the first case, one of $K_{1,4}$, $K_{1,3}^{2}$, $P_{3}^{2,2}$, $K^{3}_{1,2}$ or $\Theta_{4}$ is an immersion of $G$.
In the second case $G$ contains one of $K_{3,3}$ or $K_{5}$ as a topological minor.
As a subcubic graph cannot contain $K_{5}$ as a topological minor, we conclude that, in case $G$ is not planar, it contains $K_{3,3}$ as a topological minor and therefore also as an immersion.
We conclude that one of 
$K_{3,3}$, $K_{1,4}$, $K_{1,3}^{2}$, $P_{3}^{2,2}$, $K_{1,2}^{3}$ or $\Theta _{4}$ is an immersion of $G$.
We summarize our observations in the following.

\begin{proposition}\label{ismdwsaal}
The set $\{ \mathscr{W} \}$ is an immersion-universal obstruction for $\tcw$.
Moreover,
\begin{align*}
\cobs_{\leqslant_{\mathsf{i}}}(\tcw) \ &= \ \{\Pcal^{(\leq 3)}\}\text{ and}\\
\pobs_{\leqslant_{\mathsf{i}}}(\tcw) \ &= \ \big\{ \{K_{3,3}, K_{1,4}, K_{1,3}^{2}, P_{3}^{2,2}, K^{3}_{1,2}, \Theta_{4}\} \big\}.
\end{align*}
\end{proposition}

\subsection{Slim tree-cutwidth}

Note that an important feature in the definition of tree-cutwidth is the concept of a \textsl{bold} adhesion.
Bold adhesions are those that ``essentially count'' as one can observe from \eqref{kkfkadjg} and \eqref{aclop}.
However, what if we were to slightly relax the definition of bold adhesions so that they are those that have size at least two?
This variant of tree-cutwidth  was considered by Ganian and Korchemna \cite{GanianKorchemna2024SlimTreeCutWidth} under the name \textsl{slim tree-cutwidth}.
Following the terminology introduced in \cite{GanianKorchemna2024SlimTreeCutWidth}, we abbreviate this parameter by
\stcw.

It appears that to capture the approximate behaviour of $\stcw$, obtained via the aforementioned relaxation of the boldness definition, we require one more immersion-parametric graph.
According to \cite{GanianKorchemna2024SlimTreeCutWidth}, there exists a function $f \colon \nton$ such that every graph $G$ where $\stcw(G) \geq f(k)$ contains either $\mathscr{W}_{k}$ or $K_{1,k}^{2}$ as an immersion, where $K_{1,k}^{2}$ is the \emph{double-edge star}, obtained from the star $K_{1,k}$ after duplicating all its edges\footnote{The result in \cite{GanianKorchemna2024SlimTreeCutWidth} is stated in terms of simple graphs. Here we give an interpretation of the universal obstruction of \cite{GanianKorchemna2024SlimTreeCutWidth} on graphs with multiple edges}.
Let $\mathscr{K}^{\mathsf{ds}} \coloneqq \langle K_{1,k}^{2} \rangle_{k \in \mathbb{N}}$.

\begin{figure}[htbp]
\centering
\includegraphics[width=0.49\linewidth]{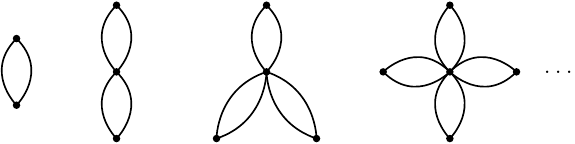}
\caption{\label{double_star_sequence}The immersion-parametric graph $\mathscr{K}^{\mathsf{ds}} = \langle \mathscr{K}^{\mathsf{ds}}_1, \mathscr{K}^{\mathsf{ds}}_2, \mathscr{K}^{\mathsf{ds}}_3, \mathscr{K}^{\mathsf{ds}}_4,\dots\rangle$ of double-edge stars.}
\end{figure}

We also denote by $\Ccal^{DS} \coloneqq \closure{\leqslant_{\mathsf{i}}}{\mathscr{K}^{\mathsf{ds}}}$.
Let us observe that every subdivision of $K_{1,k}^{2}$ belongs to $\Ccal^{DS}$.
To see this, one has to lift pairs of edges that are incident to distinct degree-two vertices.
To complete the picture we need to identify the immersion obstruction set of $\Ccal^{DS}$.
For this, let $\Ocal^{\textsf{ds}}$ be the set of graphs depicted in \autoref{obs_stcw}.

\begin{figure}[htbp]
\centering
\includegraphics[width=0.9\linewidth]{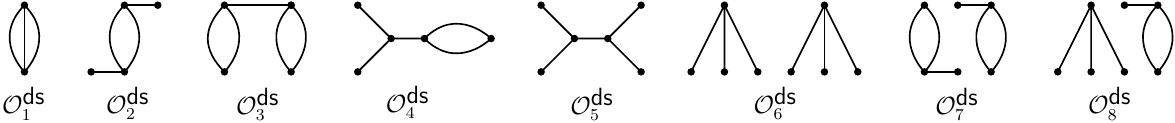}
\caption{\label{obs_stcw}The immersion-obstruction set $\Ocal^{\mathsf{ds}}$.}
\end{figure}

\begin{lemma}\label{smallsobsletution}
$\obs_{\leqslant_{\mathsf{i}}}(\Ccal^{DS}) = \Ocal^{\mathsf{ds}}$.
\end{lemma}
\begin{proof}
It is easy to see that for every $\Ocal^{\textsf{ds}} \cap \Ccal^{DS} = \emptyset$ and moreover that $\Ocal^{\textsf{ds}}$ is an immersion-antichain.
Let $G$ be a graph such that $\Ocal^{\textsf{ds}}_{1}$, $\Ocal^{\textsf{ds}}_{2}$, $\Ocal^{\textsf{ds}}_{3}$, $\Ocal^{\textsf{ds}}_{4}$, $\Ocal^{\textsf{ds}}_{5}$, $\Ocal^{\textsf{ds}}_{6}$, $\Ocal^{\textsf{ds}}_{7}$, $\Ocal^{\textsf{ds}}_{8} \nleqslant_{\mathsf{i}} G$.
Assume that there exist at least two vertices of degree at least three. Let $u, v$ be vertices of degree at least three that are within minimum distance in $G$.
First, observe that since $\Ocal^{\textsf{ds}}_{1} \nleqslant_{\mathsf{i}} G$, no edge incident to either $u$ or $v$ can be of multiplicity at least three.

Assume that $u$ and $v$ are at distance one, i.e. $e = uv \in E(G)$.
If $e$ has multiplicity two and $u$ and $v$ have a common neighbour then $\Ocal^{\textsf{ds}}_{1} \leqslant_{i} G$.
Otherwise they each have a private neighbour and then $\Ocal^{\textsf{ds}}_{2} \leqslant_{i} G$.
Suppose that $e$ has multiplicity and $u$ and $v$ have a common neighbour, say $w$.
If at least one of the edges $uw$ and $vw$ has multiplicity two, then $\Ocal^{\textsf{ds}}_{1} \leqslant_{\mathsf{i}} G$.
Otherwise, both $uw$ and $vw$ are of multiplicity one which implies that both $u$ and $v$ have at least one more neighbour.
If it is a common neighbour, then $\Ocal^{\textsf{ds}}_{1} \leqslant_{\mathsf{i}} G$.
Otherwise, $\Ocal^{\textsf{ds}}_{2} \leqslant_{\mathsf{i}} G$.
It remains to examine the case when $u$ and $v$ have no common neighbour.
Then both $u$ and $v$ have private neighbours.
If both have a private neighbour where the incident edges have multiplicity two, then $\Ocal^{\textsf{ds}}_{3} \leqslant_{\mathsf{i}} G$.
If one of the two has multiplicity one, then $\Ocal^{\textsf{ds}}_{4} \leqslant_{\mathsf{i}} G$.
Otherwise, if both have multiplicity one then both $u$ and $v$ have a second private neighbour and then $\Ocal^{\textsf{ds}}_{5} \leqslant_{\mathsf{i}} G$.

Next, assume that $u$ and $v$ are at distance two, i.e. $u$ and $v$ are not adjacent but have a common neighbour $w \in V(G)$. Let $uw$ and $vw$ be the incident edges.
Observe that neither $uw$ nor $vw$ can be of multiplicity two, since this would contradict the minimum distance assumption.
Then both $u$ and $v$ have at least two more neighbours.
If both are common, then $\Ocal^{\textsf{ds}}_{1} \leqslant_{\mathsf{i}} G$.
If one is common, then $\Ocal^{\textsf{ds}}_{2} \leqslant_{\mathsf{i}} G$.
If none is common, then $\Ocal^{\textsf{ds}}_{5} \leqslant_{\mathsf{i}} G$.

Finally, assume that $u$ and $v$ are at distance at least three.
In this case the graphs induced by vertices of distance up to two from $u$ and $v$ are disjoint.
Then it is easy to see that one of $\Ocal^{\textsf{ds}}_{6} \leqslant_{\mathsf{i}} G$, $\Ocal^{\textsf{ds}}_{7} \leqslant_{\mathsf{i}} G$ or $\Ocal^{\textsf{ds}}_{8} \leqslant_{\mathsf{i}} G$ holds.

Hence, there exists at most one vertex of degree at least three in $G$.
It is not difficult to observe that removing that vertex (if it  exists) yields a graph whose connected components are paths or cycles.
In other words, $G$ is the disjoint union of a set of paths, cycles and a subdivision of $K_{1, k}^{2}$.
It follows that $G \in \Ccal^{DS}$.
\end{proof}

Based on \autoref{smallsobsletution} and the results of \cite{GanianKorchemna2024SlimTreeCutWidth} we conclude the discussion on slim tree-cutwidth with the following proposition.

\begin{theorem}\label{ismdwsaal}
The set $\{\mathscr{W}, \mathscr{K}^{\mathsf{ds}}\}$ is an immersion-universal obstruction for $\stcw$.
Moreover,
\begin{align*}
\cobs_{\leqslant_{\mathsf{i}}}(\stcw) \ &= \ \{\Pcal^{(\leq 3)}, \Ccal^{DS}\}\text{ and}\\
\pobs_{\leqslant_{\mathsf{i}}}(\stcw) \ &= \ \big\{ \{K_{3,3}, K_{1,4}, K_{1,3}^{2}, P_{3}^{2,2}, K_{1,3}^{3} \Theta _{4}\}, \Ocal^{\sf{ds}} \big\}.
\end{align*}
\end{theorem}

\subsection{Carving width}
\label{def_carv_all}

A \emph{carving decomposition} of a graph is a pair $(T, \sigma)$ where $T$ is a ternary tree whose set of leaves is $L$ and $ \sigma$ is a bijection $\sigma \colon L \to V(G)$.
For every edge $e = xy \in E(T)$, we denote by $L_{1}^{e}, L_{2}^{e}$ the set of leaves of the two connected components of $T - e$ and we define $E_{e}$ as the set of all edges that have one endpoint in $\sigma(L_{1}^{e})$ and the other in $\sigma(L_{2}^{e})$.
The \emph{width} of the carving decomposition $(T, \sigma)$ is the maximum $|E_{e}|$ over all edges of $T$.
The \emph{carving width} of a graph $G$, denoted by $\cvw(G),$ is defined as follows.
\begin{align}
\cvw(G) \ \coloneqq \ \min\{ k \in \Nbbb \mid \text{there exists a carving decomposition of }G\text{ of width }k\}.
\end{align}

Carving width was defined by Seymour and Thomas \cite{SeymourT94CallRouting} and can be seen as the ``edge'' analogue of the parameter \textsl{branchwidth}\footnote{Branchwidth was defined in \cite{RobertsonS91GMX} and is equivalent to treewidth.}.
It is easy to prove (see e.g., \cite{NestoridisT14Squareroots}) that
\begin{eqnarray}
\tw \ \preceq \ \cvw \ \preceq \ \Delta + \tw.\label{equkkkdwe}
\end{eqnarray}

Observe that, as $\Delta(\mathscr{W}_{k}) \leq 3$ and $\mathscr{W}_{k}$ is planar, if a graph excludes the wall $\mathscr{W}_{k}$
as an immersion, then it also excludes it as a topological minor and therefore also as a minor.
This implies that the treewidth of graphs with no $W_{k}$ minor is at most $\Ocal(k^9(\log k)^{\Ocal(1)})$ by the grid exclusion theorem \cite{ChuzhoyT21Towards}.
We deduce that $\tw \preceq \p_{\mathscr{W}}$ and this, combined with  \eqref{equkkkdwe}, implies that $\cvw \preceq \Delta + \p_{\mathscr{W}}$. 
Moreover, as  $\tw(\mathscr{W}_{k}) = k$ and $\tw \preceq \cvw$ (because of 
the left part of \eqref{equkkkdwe}), it follows that $\p_{\mathscr{W}} \preceq \cvw$.
It is also easy to see that $\Delta \preceq_{\mathsf{P}} \cvw.$
Collectively, these observations imply that $\Delta + \p_{\mathscr{W}} \preceq_{\mathsf{P}} \cvw$.
We conclude that $\Delta + \p_{\mathscr{W}} \sim_{\mathsf{P}} \cvw$, which allows us to state the following.

\begin{theorem}\label{ismsdwsaal}
The set $\{\Theta, \mathscr{K}^{\mathsf{s}}, \mathscr{W}\}$ is an immersion-universal obstruction for $\cvw$ with polynomial gap.
Moreover,
\begin{align*}
\cobs_{\leqslant_{\mathsf{i}}}(\cvw) \ &= \ \{\Ccal^{\Theta }, \Ccal^{S}_{\mathsf{i}}, \Pcal^{(\leq 3)}\}\text{ and}\\
\pobs_{\leqslant_{\mathsf{i}}}(\cvw) \ &= \ \big\{ \{3\cdot K_{1}\}, \{\Theta_{2}, 2\cdot P_{3}, P_{4}\}, \{K_{3,3}, K_{1,4}, K_{1,3}^{2}, P_{3}^{2,2}, K_{1,2}^{3}, \Theta_{4}\} \big\}.
\end{align*}
\end{theorem}

\subsection{Cutwidth}
\label{cutwidth_sec}

The \emph{cutwidth} of a graph $G$, denoted by $\cw$, is defined so that
\begin{align}
\nonumber\cw(G) \ \coloneqq \ \min\{ k \in \Nbbb \mid \ &\text{there exists an ordering }v_{1}, \ldots, v_{n}\text{ of }V(G) \text{~such that}\\
&\text{for every $i\in[2,n],$ there are at most~} k\text{ edges in }G\text{ with }\\
\nonumber&\text{one endpoint in }\{v_{1} \ldots, v_{i-1}\}\text{ and the other in }\{v_{i} \ldots v_{n}\}\}.
\end{align}

Recalling the vertex ordering definition of pathwidth we have in \autoref{pathwidth_def}, it follows that 
\begin{eqnarray}
\pw \ \preceq \ \cw \ \preceq \ \Delta + \pw.\label{equskkskdwe}
\end{eqnarray}

Recall that $\mathscr{T}=\langle \mathscr{T}_k\rangle_{k\in\mathbb{N}_{\geq 1}}$ is the minor-parametric graph of ternary trees.
As $\Delta(\mathscr{T}_{k}) \leq 3$, we may follow a similar reasoning to the case of $\cvw$ before.
If a graph excludes $\mathscr{T}_{k}$ as an immersion, then it also excludes it as a topological minor and therefore also as a minor.
This implies that $\pw \preceq \p_{\mathscr{T}}$ and this, combined with  \eqref{equskkskdwe}, implies that $\cw \preceq \Delta + \p_{\mathscr{T}}$.
Moreover, it easily follows that $\p_{\mathscr{T}} \preceq \pw$ (see e.g.,~\cite[Proposition 3.2]{KirousisP86Searching} or \cite{ChungS89Graphswith}) and $\pw \preceq \cw$ (because of the left part of \eqref{equskkskdwe}).Therefore $\p_{\mathscr{T}} \preceq \cw$.
It is also easy to see that $\Delta \preceq_{\mathsf{P}} \cw$.
We conclude that $\Delta + \p_{\mathscr{T}} \sim \cw$.
We define $\gtforest$ to be the class of all forests of maximum degree at most three and notice that $\closure{\leqslant_{\mathsf{i}}}{\mathscr{T}} = \gtforest$ and that $\obs_{\leqslant_{\mathsf{i}}}(\gtforest) = \{\Theta_2, K_{1,4} \}$.
We sum up the above observations with the following proposition.

\begin{theorem}\label{ishmssdwsaal}
The set $\{\Theta, \mathscr{K}^{\mathsf{s}}, \mathscr{T}\}$ is an immersion-universal obstruction for $\cw$.
Moreover, $\cobs_{\leqslant_{\mathsf{i}}}(\cw) = \big\{ \Ccal^{\Theta }, \Ccal^{S}_{\mathsf{i}}, \gtforest \big\}$ and $\pobs_{\leqslant_{\mathsf{i}}}(\cw) = \big\{ \{3\cdot K_{1}\}, \{ \Theta_{2}, 2\cdot P_{3}, P_{4}\}, \{\Theta_2, K_{1,4}\} \big\}$.
\end{theorem}

\subsection{Excluding a graph as an immersion}

In the case of graph minors, cliques serve as the graphs that contain every other graph as a minor, given that we choose a clique of sufficiently large order.
It is natural to wonder what is the analogous concept for immersions.
Of course, we may consider cliques where both the number of vertices and multiplicities of edges are increasing.
Instead, we define a somehow more normalized immersion-omnivore of the set of all graphs with multiple edges as follows.
Let $\mathscr{I} \coloneqq \langle \mathscr{I}_{k} \rangle_{k \in \mathbb{N}_{\geq 1}}$ be the immersion-parametric graph where  $\mathscr{I}_{k}$ is the graph obtained from $K_{1, k}$ if we set the multiplicity of its edges to be $k,$ for every $k \geq 1.$
See \autoref{Theta_immersions_seq} for an illustration.

\begin{figure}[htbp]
\centering
\includegraphics[width=0.53\linewidth]{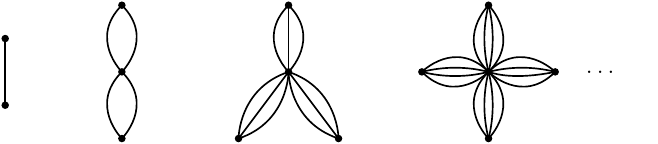}
\caption{\label{Theta_immersions_seq}The immersion-parametric graph $\mathscr{I} = \langle \mathscr{I}_1, \mathscr{I}_2, \mathscr{I}_3, \mathscr{I}_4, \ldots \rangle$.}
\end{figure}

Note that the definition of $\mathscr{I}_{k}$ has two ``degrees of freedom'', contrary to the definitions of $\mathscr{K}^{\mathsf{s}}$ and $\mathscr{K}^{\mathsf{ds}}$.
Moreover, note that every clique on $k$ vertices is an immersion of $\mathscr{I}_{k}$.
A graph decomposition analogue of $\p_{\mathscr{I}}$ has been proposed by Wollan \cite{Wollan15Thestructure}.
In order to describe it, we use tree-cut decompositions and we need to additionally define the notion of a torso of such a decomposition.

Let $\mathcal{T} = (T, {\mathcal{X}})$ be a tree-cut decomposition of a graph $G$. 
The torso at a node $t$ of $T$ is the graph $H_{t}$ defined as follows.
We first assume that $G$ is connected, otherwise we apply the definition to each of the connected components of $G$.
Let $T_{1}, T_{2}, \ldots, T_{p}$ be the connected components of $T - t$.
We set $X_{T_i} = \bigcup_{t' \in V(T_{i})} X_{t'}$ and observe that $\{X_{T_1}, \dots, X_{T_p} \}$ is a near-partition of $V(G) \setminus X_{t}$.
The \emph{torso} $H_t$ is the graph obtained from $G$ by identifying the vertices of $X_{T_{i}}$ into a single vertex $z_{i}$, for each $i \in [p]$, removing all the resulting loops, and keeping multiple edges that may appear.
Let $\wn$ be the graph parameter so that
\begin{align}
\nonumber\wn(G) \ \coloneqq \ \min\{ k \in \Nbbb \mid \ &G\text{ has a tree-cut decomposition where}\\
&\text{every edge has adhesion at most }k\text{ and}\\
\nonumber&\text{every torso has}\leq k\text{ vertices of degree}> k.\}
\end{align}

According to \cite{Wollan15Thestructure}, $\wn \sim \p_{\mathscr{I}}$.
This implies the following.

\begin{theorem}\label{ismsdwsaal}
The set $\{ \mathscr{I} \}$ is an immersion-universal obstruction for $\wn$. Moreover, $\cobs_{\leqslant_{\mathsf{i}}}(\wn) = \big\{ \gall^{\mathsf{e}} \big\}$ and $\pobs_{\leqslant_{\mathsf{i}}}(\wn) = \big\{ \emptyset \big\}$.
\end{theorem}

\subsection{Hierarchy of immersion-monotone parameters}

We conclude this section  with the small fraction of the quasi-ordering of immersion-monotone parameters consisting of the immersion-monotone parameters e have considered in this section, depicted in \autoref{fig_hierarchy_immersions}.

\begin{figure}[htbp]
\centering
\includegraphics[width=0.9\linewidth]{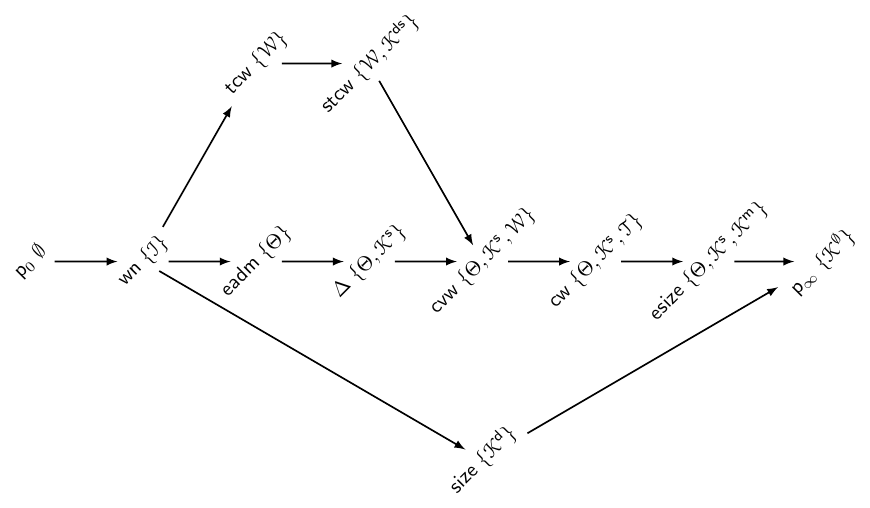}
\caption{\label{fig_hierarchy_immersions} The hierarchy of immersion-monotone parameters.}
\end{figure}

\section{Obstructions of vertex-minor-monotone parameters}
\label{vesjrs_misos}

In this section we turn our focus to the \textsl{vertex-minor} relation and discuss a few key vertex-minor-monotone parameters and their universal obstructions.

\medskip
Let $G$ be a graph and let $v \in V(G)$.
The result of a \emph{local complementation} at a vertex $v \in V(G)$ in $G$ is the graph obtained if we remove from $G$ all edges in $G[N_{G}(v)]$ and then add all edges that are not present in $G[N_{G}(v)]$, i.e., we replace the induced subgraph $G[N_{G}(v)]$ in $G$ by its complement.
Two graphs $H$ and $G$ are \emph{locally equivalent} if one is obtained from the other by a sequence of local complementations.
A graph $H$ is a \emph{vertex-minor} of $G$, denoted by $H \leqslant_{\mathsf{v}} G,$ if $H$ is an induced subgraph of a graph that is locally equivalent to $G$.

With the vertex-minor relation we have to be careful when defining the universe of graphs we work in.
Indeed, locally equivalent graphs may not be isomorphic.
Therefore, we define $\gall^{\mathsf{v}}$ to contain one graph from {each}  equivalence class of locally equivalent graphs.
This forces $(\gall^{\mathsf{v}}, \leqslant_{\mathsf{v}})$ to be a partial-ordering relation.

Another issue we should discuss involves the well-quasi-ordering status of $\leqslant_{\mathsf{v}}$ on  $\gall^{\mathsf{v}}$ which is an open problem.
Therefore, there is no apriori guarantee that vertex-minor-monotone parameters admit class obstructions, parametric obstructions, and/or universal obstructions.

\subsection{The two extreme cases}

We start by studying the two extremal points in the hierarchy of vertex-minor-monotone parameters.
The first is the parameter $\size$.
Notice that, by Ramsey's theorem \cite{Ramsey30Ona}, every large enough graph contains either a independent set or a complete graph on $k+1$ vertices.
By applying a local complementation on a vertex of a complete graph on $k+1$ vertices we obtain $K_{1,k}$ that contains an independent set of $k$ vertices.
Therefore, as in the case of minors and immersions we obtain the following proposition (see \cite[Theorem 3.1]{KwonO14Unavoidable}).

\begin{proposition}
The set $\{ \mathscr{K}^{\mathsf{d}} \}$ is a vertex-minor-universal obstruction for $\size.$
Moreover, $\cobs_{\leqslant_{\mathsf{v}}}(\size) = \big\{ \excl_{\leqslant_{\mathsf{v}}}(K_{2}) \}$ and $\pobs_{\leqslant_{\mathsf{v}}}(\size) = \big\{ \{K_{2}\} \big\}$.
\end{proposition}

On the other end  of the hierarchy, we need to conjure a parameter bounded for every vertex-minor-closed graph class.
For this consider the graph sequence $\mathscr{K}^{s} = \langle K_{k}^{s} \rangle_{k\in\mathbb{N}_{\geq 4}}$ where $K_{k}^{s}$ is obtained from the complete graph $K_{k}$ after subdividing all of its edges once.
Notice that we can obtain every graph $H$ as a vertex-minor of $K_{|H|}^{s}$.
Therefore $\mathscr{K}^{s}$ is a vertex-minor-omnivore of $\gall^{\mathsf{v}}$.
We conclude that $\cobs_{\leqslant_{\mathsf{v}}}(\p_{\mathscr{K}^{s}}) = \{ \gall^{\mathsf{v}} \}$ and $\pobs_{\leqslant_{\mathsf{v}}}(\p_{\mathscr{K}^{s}}) = \{ \emptyset \}$.
Clearly, $p_{\mathscr{K}^{s}}$ is a max-min parameter, defined by the exclusion of subdivided cliques.
To find a decomposition based min-max analogue of $\p_{\mathscr{K}^{s}}$ is a central open problem in the study of vertex-minors.

\subsection{Rankwidth}

The study of parameters that are monotone under vertex-minors is more restricted than the immersion-monotone or the minor-monotone ones.
However it has attracted considerable attention, mostly due to the parameter of rankwidth.

\medskip
The rankwidth parameter was introduced by Oum and Seymour~\cite{OumS06appro} and it is defined using carving decompositions which we already introduced in \autoref{def_carv_all}).
In the context of rankwidth, we shall refer to carving decompositions as \emph{rank decompositions}.
Let $(T, \sigma)$ be a rank decomposition of a graph $G$ and let $e \in E(T)$.
Recall that $L_{1}^{e}, L_{2}^{e}$ are the leaves of the two connected components of $T - e$ which defines the partition $\{ \sigma(L_{1}^{e}), \sigma(L_{2}^{e})\}$ of $V(G)$.
Consider a numbering $\{x_{1}^1, \ldots, x^1_{r_{1}}\}$ of the vertices of $\sigma(L_{1}^{e})$ and numbering $\{x_{1}^2, \ldots, x^2_{r_{2}}\}$ of the vertices of $\sigma(L_{2}^{e})$. We define $M_G(e)$ as the $0$-$1$ matrix over the binary field such that the entry at the $i$-th row and the $j$-th column has value one if $\{x_{i}^1, x_{j}^{2}\} \in E(G)$, otherwise the value is zero.
The \emph{width} of $(T, \sigma)$ is the maximum rank of the matrix $M_G(e)$ over all edges $e$ of $T.$
The \emph{rankwidth} of a graph $G$, denoted by $\rw(G),$ is now defined as follows (see \autoref{fig_rw}) for an example).
\begin{align}
\rw(G) \ &\coloneqq \ \min\{ k \in \Nbbb \mid \text{there is a rank decomposition of }G\text{ of width }k \}.
\end{align}

\begin{figure}[htbp]
\centering
\includegraphics[width=0.7\linewidth]{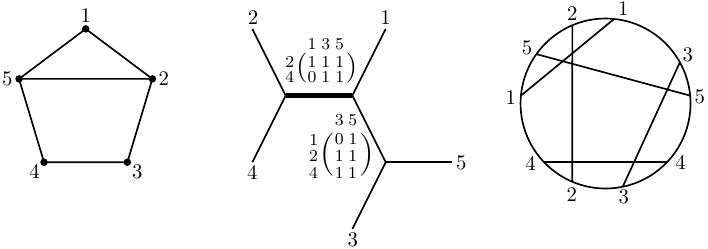}
\caption{\label{fig_rw} A graph $G$ and a rank decomposition $(T, \sigma)$ of $G$.
The partition $\{ \sigma (L_{1}^{e}), \sigma (L_{2}^{e})\}$ corresponding to the thick horizontal tree-edge $e$ is $\{\{2,4\},\{1,5,3\}\}$ and the corresponding matrix $M_G(e)$ has rank two.
The graph $G$ is a circle graph, as witnessed by its chord diagram on the right.}
\end{figure}

The importance of rankwidth emerged from the fact that it is equivalent to the parameter cliquewidth.
Here we avoid giving the definition of cliquewidth.
It has been introduced by Courcelle,  Engelfriet, and  Rozenberg \cite{CourcelleER93}. 
Courcelle,  Makowsky, and Rotics \cite{CourcelleMR00Linear} proved that problems on graphs that are expressible in {MSOL$_{1}$}\footnote{MSOL$_{1}$ is the restriction of MSOL where we are not allowed to quantify on sets of edges.} can be decided in linear time for graphs with small cliquewidth, {given that a rank decomposition of small width is given} {(the construction of a rank decomposition of width $k$ can be done by the algorithm of Hliněný and Oum in \cite{HlinenyO08FindingBranch} 
in  $2^{O(k^2)} n^3$ time or by the algorithm of   Korhonen and  Sokołowski in \cite{KorhonenSokolowski2024Rankwidth}  running in $f(k)\cdot (n^{1+o(1)}) + O(m)$ time for some function $f:\Nbbb\to\Nbbb$).}
Clearly, this algorithmic meta-theorem applies also to rankwidth since its equivalent to cliquewidth. 
Interestingly, cliquewidth is not monotone under any well-studied quasi-ordering relation on graphs except from the induced subgraph relation, a fact which makes its study more complicated, from a ``structural'' point of view.
Rankwidth offers a way out of this, as it is monotone under vertex-minors which, in turn, made it possible to find obstruction characterizations, whose presentation follows.

A graph $G$ is a \emph{circle graph} if and only if it is the intersection graph of a set of chords in a circle: the chords represent the vertices and two vertices are adjacent if and only if their corresponding chords intersect (see \autoref{fig_rw}). 
We use $\gcircle$ for the class of circle graphs.
For a positive integer $k$, the \emph{$(k \times k)$-comparability grid} is the graph $\mathscr{G}^{\mathsf{c}}_k$ where $V(\mathscr{G}^{\mathsf{c}}_k) = [k]^2$ and $E(\mathscr{G}^{\mathsf{c}}_k) \ = \ \{ \{ (i, j), (i', j') \} \mid i \leq i' \wedge j \leq j' \}$.
Let $\mathscr{G}^{\mathsf{c}} = \langle \mathscr{G}^{\mathsf{c}}_k \rangle_{k \geq 2}$ denote the sequence of graphs such that for every $k \geq 2$, $\mathscr{G}^{\mathsf{c}}_k$ is the $(k \times k)$-comparability grid (see \autoref{fig_comp_grid}).

\begin{figure}[htbp]
\centering
\includegraphics[width=0.7\linewidth]{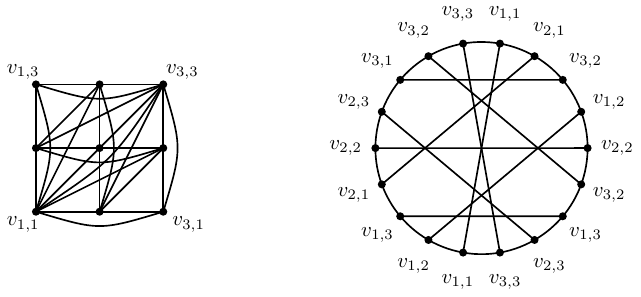}
\caption{\label{fig_comp_grid}The $(3,3)$-comparability grid and its circle graph representation.}
\end{figure}

Geelen, Kwon, McCarty, and Wollan showed \cite{GeelenKMW23Thegrid} that there exists a function $f \colon \mathbb{N} \rightarrow \mathbb{N}$ so that for every $k \in \mathbb{N}$, every graph of rankwidth at least $f(k)$ contains as a vertex-minor the $(k\times k)$-comparability grid $\mathscr{G}^{\mathsf{c}}_k$.
In the same paper, it was observed that every circle graph is a vertex-minor of a comparability grid.
As every comparability grid is also a circle graph, we have $\gcircle = \closure{\leqslant_{\mathsf{v}}}{\mathscr{G}^{\mathsf{c}}}$, i.e., the vertex-minor-parametric graph of comparability grids is a vertex-minor-omnivore of circle graphs.
Moreover, it was proved by Bouchet in~\cite{Bouchet94circl} that $\obs_{\leqslant_{\mathsf{v}}}(\gcircle) = \{W_{5}, W_{7}, W_3^{3s} \}$ where $W_{r}$ is the wheel on $r+1$ vertices and $W_3^{3s}$ is obtained after subdividing the edges of a triangle of $K_{4}$ once (see \autoref{fig_circle_obs}).

\begin{figure}[htbp]
\centering
\includegraphics[width=0.45\linewidth]{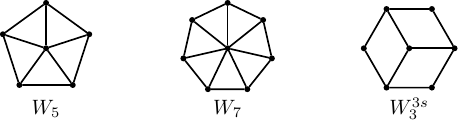}
\caption{\label{fig_circle_obs} The vertex-minor obstruction set of circle graphs \cite{Bouchet94circl}.}
\end{figure}

We conclude with the following proposition.

\begin{proposition}\label{ismsdwsaal}
The set $\{ \mathscr{G}^{\mathsf{c}} \}$ is a vertex-minor-universal obstruction for $\rw$.
Moreover, $\cobs_{\leqslant_{\mathsf{v}}}(\rw) = \big\{ \gcircle \big\}$ and $\pobs_{\leqslant_{\mathsf{v}}}(\rw) = \big\{ \{W_{5}, W_{7}, W_3^{3s}\} \big\}$.
\end{proposition}

Note that, according to \autoref{ismsdwsaal}, the parameter $\p_{\mathscr{G}^{\mathsf{c}}}$ can be seen as the max-min analogue of cliquewidth (in terms of equivalence).
This means that even though cliquewidth is not monotone under vertex-minors, graphs excluding some comparability grid as a vertex-minor have bounded cliquewidth.

\subsection{Linear rankwidth}

A natural step to make is to define the ``path-analogue'' of rankwidth.
The \emph{linear rankwidth} of a graph $G$, denoted 
$\mathsf{lrw}(G),$
is obtained by restricting the tree $T$ of a rank decomposition $(T, \rho)$ of $G$ to be a ternary caterpillar (defined in \autoref{apeplicx}, see also \autoref{catrexlem}).
Kanté and Kwon \cite{KanteK18Linear} conjectured that there exists a function $f \colon \mathbb{N} \rightarrow \mathbb{N}$ so that for every $k\in\mathbb{N}$, every graph of linear rankwidth at least $f(k)$ contains a tree $T$ on $k$ vertices as a vertex-minor.
Resolving this conjecture would be an important step towards identifying a vertex-minor-universal obstruction for linear rankwidth.
{It is worth to mention that this conjecture is further supported by the recent result of \cite{BojanczykO2025Graphs}.}

\subsection{Rankdepth}

Our last step is to consider the analogue of treedepth in the world of vertex-minors. 
This was done by DeVos, Kwon, and Oum \cite{DevosKO20Branchdepth} (in a more abstract setting).
The corresponding parameter is the \emph{rankdepth} of a graph.

The \emph{radius} of a tree $T$ is $r$ if it contains a vertex $v$ such that every vertex of $T$ is at distance at most $r$ from $v$.
Let $(T, \sigma)$ be a pair where $T$ is a tree and $\sigma$ is a bijection from $V(G)$ to the set of leaves of $T$.
Let us call such a pair $(T, \sigma)$ a \emph{decomposition} of $G$.
Let $G$ be a graph and let $(T, \sigma)$ be a decomposition of $G$.
Suppose that $t$ is an internal vertex of $V(T)$.
Let $e_{1}, \ldots, e_{r}$ be the edges of $T$ that are incident to $t$.
Let $A^{t}_{i}$ be the set of leaves in the component of $T - e_{i}$ not containing $t$.
For a subset $I \subseteq [r],$ let $M^{t}_{G}(I)$ be the $0$-$1$ matrix over the binary field whose rows correspond to $\bigcup_{i \in I} A_{i}^{t}$ and columns correspond to $\bigcup_{j \in [r] \setminus I} A^{t}_{j}$ such that an entry is one if and only if the vertices corresponding to the row and the column are adjacent.
Then the \emph{width} of $t$ is the maximum rank of $M^{t}_{G}(I)$ over all $I \subseteq [r].$
 We say that  $(T, \sigma)$  is a \emph{$(k,r)$-decomposition} if all its internal vertices have width at most $k$ and the radius of $T$ is at most $r$.
The \emph{rankdepth} of $G$, denoted by $\rkd(G)$, is now defined as follows.
\begin{align}
\rkd(G) \ \coloneqq \  \min\{ k \in \Nbbb \mid \text{there is a }(k, k)\text{-decomposition of }G \}.
\end{align}

Note that in the special case that $|V(G)| \leq 2,$ $\rkd(G) = 0$.
As observed in \cite{DevosKO20Branchdepth}, rankdepth is vertex-minor-monotone.
Similarly to the case of treedepth, presented in \autoref{treedepth_par}, we consider the vertex-minor-parametric graph $\mathscr{P} = \langle \mathscr{P}_k \rangle_{k\in\mathbb{N}_{\geq 1}}$.
We denote by $\vpaths$ the set of all vertex-minors of paths.
It can be seen that $\rkd$ is unbounded in $\vpaths$ and that $\mathscr{P}$ is a vertex-minor-omnivore of $\vpaths$.
Moreover, Kwon,  McCarty,  Oum, and Wollan \cite{KwonMOW21Obstructions}  proved that there exists a function $f \colon \mathbb{N} \to \mathbb{N}$ so that for every $k \in \mathbb{N}$, every graph of rankdepth at least $f(k)$ contains $\mathscr{P}_k$ as a vertex-minor. 
 Also, according to \cite{KwonO14Graphs} and \cite{AdlerFP14Obstructions}, $\obs_{\leqslant_{\mathsf{v}}}(\vpaths) = \{C_{5}, K_{3}^{3p}, C_{4}^{2p} \}$ where $K_{3}^{3p}$ is obtained if we make adjacent the three vertices of $K_{3}$ with three new vertices and $C_{4}^{2p}$ is obtained if we make adjacent two non-adjacent vertices of $C_{4}$ with two new vertices (see \autoref{fig_vm_path_obstructions}).
 
 \begin{figure}[htbp]
\centering
\includegraphics[width=0.45\linewidth]{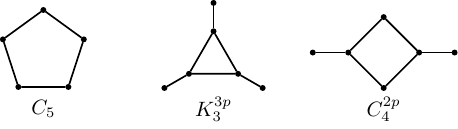}
\caption{\label{fig_vm_path_obstructions} The vertex-minor obstruction set of paths.}
\end{figure}

 We summarize our discussion with the following proposition.

\begin{proposition}\label{llalthereneite}
The set $\{ \mathscr{P} \}$ is a vertex-minor-universal obstruction for $\rkd$.
Moreover, $\cobs_{\leqslant_{\mathsf{v}}}(\rkd) = \big\{ \vpaths \big\}$ and $\pobs_{\leqslant_{\mathsf{v}}}(\rkd) = \big\{ \{ C_{5}, K_{3}^{3p}, C_{4}^{2p} \}\big\}$.
\end{proposition}

Shrubdepth is also defined as analogue of treedepth in the world of dense graphs introduced in \cite{GanianHNOM2019ShrubDepth}. 
  Inspired by the clique-width definition, it is based on so-called tree-models : a labeled rooted tree of bounded depth whose leaves are mapped to the vertices of the graph such that the adjacency between two vertices $u$ and $v$ only depends on their labels and the label of their least common ancestor. It was proved parametrically equivalent to rank-depth by \cite{DevosKO20Branchdepth}. Very recently, \cite{Mahlmann25Forbidden} proved that every class of bounded shrubdepth is also characterized by a finite list of forbidden induced subgraphs.


\subsection{Hierarchy of vertex-minor-monotone parameters}

We conclude this section with the hierarchy of all vertex-minor-monotone parameters that we have considered, depicted in \autoref{fig_hierarchy_vminors}.

\begin{figure}[htbp]
\centering
\includegraphics[width=0.83\linewidth]{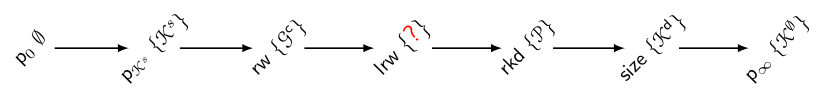}
\caption{\label{fig_hierarchy_vminors}The hierarchy of vertex-minor monotone parameters.}
\end{figure}

We wish to stress that our presentation of vertex-minor-monotone parameters is not exhaustive.
There are more vertex-minor-monotone parameters in the literature, for which universal obstructions have been studied.
For this we refer the interested reader to \cite{KwonO14Unavoidable, KwonO20Scattered, KwonO21Graphs, NguyenO20Theaverage}.
Also, see the survey of Oum \cite{Oum2017RankWidth} for a deeper overview of results regarding rankwidth and related parameters.

\section{Obstructions of other parameters}\label{obsjjdjdother_a}

So far, we reviewed parameters that are monotone for the minor relation (see \autoref{rfisondlclsidop}), for the immersion relation (see \autoref{ommsers}) and for the vertex-minor relation (see \autoref{vesjrs_misos}).
The former two are known to be well-quasi-orderings while the latter one is conjectured to be.
As we discussed in \autoref{sliubeabaueike} under the assumption of a well-quasi-ordering (see \autoref{prop_wqo_cobs_pobs}),  class obstructions and parametric obstructions are guaranteed to exist.
It is however natural to examine parameters that are $\leqslant$-monotone under some quasi-ordering relation $\leqslant$ on $\gall$ that is known not to be a well-quasi-ordering, such as for example the induced subgraph relation.
We first review some results on induced subgraph-monotone parameters and then present two more graph parameters that are not monotone under any of the quasi-ordering relations that we have considered so far.

\subsection{Induced subgraph-monotone parameters}\label{induced_parameters_subsec}

Recall that an induced subgraph of a graph is obtained by deleting vertices.
Hereafter, we denote the induced subgraph relation by $\leqslant_{\mathsf{is}}$.
A graph class is \emph{hereditary} if it is closed under taking induced subgraphs, i.e. if it is $\leqslant_{\mathsf{is}}$-closed.

\paragraph{The $\size$ parameter.}

Clearly $\size$ is induced subgraph-monotone.
By Ramsey's theorem~\cite{Ramsey30Ona}, for every $p \in \mathbb{N}$, there exists a positive integer $R(p)$ such that every graph with at least $R(p)$ vertices either contains a clique of size $p$ or an independent set of size $p$.
This implies that if a hereditary class $\Gcal$ does not contain the complete graph $K_p$ nor its complement $\overline{K}_p$, then $\size$ is bounded in $\Gcal$.
In other words the class of complete graphs and the class of edgeless graphs form two minimal hereditary graph classes for which the parameter $\size$ is unbounded.
One can see that these are the only two such minimal hereditary classes.

Recall that $\mathscr{K}^{\mathsf{d}} = \langle k\cdot K_{1} \rangle_{k\in\mathbb{N}}$ denotes the sequence of edgeless graphs.
Observe that $\mathscr{K}^{\mathsf{d}}$ is an induced subgraph-parametric graph that is trivially an induced subgraph-omnivore for the class of edgeless graphs.
Likewise, we define $\mathscr{K}^{\mathsf{c}} = \langle K_{k} \rangle_{k\in\mathbb{N}}$ to be the sequence of complete graphs.
We may quickly observe that $\mathscr{K}^{\mathsf{c}}$ is an induced subgraph-parametric graph and that it is an induced subgraph-omnivore of the class of complete graphs.
In our terminology, the above discussion yields the following proposition.

\begin{proposition}
The set $\{\mathscr{K}^{\mathsf{d}},\mathscr{K}\}$ is the induced subgraph-universal obstruction for $\size$.
Moreover, $\cobs_{\leqslant_{\mathsf{is}}}(\size) = \big\{ \closure{\leqslant_{\mathsf{is}}}{\mathscr{K}^{\mathsf{d}}}, \closure{\leqslant_{\mathsf{is}}}{\mathscr{K}^{\mathsf{c}}} \big\}$ and $\pobs_{\leqslant_{\mathsf{is}}}(\size) = \big\{ \{2K_1\},\{K_2\} \big\}$.
\end{proposition}

A significant body of work exists on identifying hereditary classes of graphs that are inclusion-minimal with respect to induced subgraph-monotone graph parameters being unbounded.
In our terminology, this translates to (partially) identifying the induced subgraph-class obstruction for various induced subgraph-monotone parameters.

For a more in-depth review of such results, we refer the interested reader to the work of Lozin \cite{Lozin23Hereditary}, which identifies the induced subgraph-class obstruction for a range of parameters, including $\size$, and others, such as \textsl{complex number}, \textsl{matching number}, \textsl{neighbourhood diversity}, \textsl{VC dimension}, and more.

\paragraph{Minimal hereditary classes of unbounded cliquewidth.}

Another noteworthy line of research concerns minimal hereditary classes of unbounded cliquewidth.
Several studies in this direction identify hereditary classes that form part of the induced subgraph-class obstruction set for cliquewidth.
For readers interested in this topic, we refer to key works such as \cite{AtminasBLS21Minimal, AlecuKLZ20Between, CollinsFKLZ18Infinitely, Lozin11Minimal, BrignallC23Aframework, BrignallC22Uncountably}.

 This research has led to notable discoveries, including the identification of classes in the induced subgraph-class obstruction set for cliquewidth which, however, possess a finite induced subgraph-obstruction set.
This result, which is the subject of \cite{AtminasBLS21Minimal}, disproves a conjecture by Daligault, Rao, and Thomassé.
Moreover, it has been shown that there are infinitely many minimal hereditary classes of unbounded cliquewidth \cite{CollinsFKLZ18Infinitely}.
In fact, it appears that there are \textit{uncountably} many such classes, as shown in \cite{BrignallC22Uncountably}.
These findings imply that it is impossible to characterize cliquewidth in terms of a finite class-obstruction set with respect to the induced subgraph relation.

Another interesting direction regards well-quasi-ordering hereditary classes of labeled graphs.
In this direction, there are two prominent conjectures by Pouzet \cite{Pouzet1972BelOrdre} on necessary and sufficient conditions relating to the number of labels needed  and whether they are chosen from a well-quasi-ordered set.
In this context, Lopez \cite{Lopez2024Labelled} shows that for any hereditary class of graphs of bounded linear clique-width, one can decide whether the class is well-quasi-ordered by the induced subgraph relation when vertices are {labeled} from a finite set, and furthermore establishes that for such classes, being well-quasi-ordered with labels from a well-quasi-ordered set is equivalent to being well-quasi-ordered with a bounded finite number of labels.
This result addresses a restricted form of the two conjectures of Pouzet.
Also, very recently, Dumas and Lopez  extend in \cite{DumasLopez2026WellQuasiOrdered}  the result of \cite{Lopez2024Labelled} to clique-width.

\newcommand{\degen}{\mathsf{degeneracy}}

\paragraph{Parameters not characterized by minimal hereditary classes.}

Unfortunately not every induced subgraph-monotone parameter $\p$ can be characterized by means of minimal hereditary classes where $\p$ is unbounded.
The treewidth parameter is one of them~\cite{LozinR22Treewidth}.
But since treewidth is minor-monotone, as we have seen, we have a finite characterization in terms of minor-universal obstructions and minor-parametric obstructions (see \autoref{prop_tw}).

Let us rather consider the \emph{degeneracy} of a graph, a parameter sandwiched between the treewidth and the chromatic number with respect to $\preceq$.
The \emph{degeneracy} of a graph $G$ is defined as follows.
\begin{eqnarray*}
\degen(G) &  \coloneqq &  \min\{ k \in \Nbbb \mid \mbox{every induced subgraph $H$ of $G$ with at least one}\\
&   & \mbox{~~~~~~~~~~~~~~~~~\!vertex of degree at most $k$}\}.
\end{eqnarray*}

Clearly, by its definition the degeneracy is an induced subgraph-monotone parameter.
In \cite{AtminasL22Adichotomy}, it is proved that for any hereditary class $\Gcal$ for which $\obs_{\leqslant_{\mathsf{is}}}$ is finite, $\degen$ is bounded in $\Gcal$ if and only if $\obs_{\leqslant_{\mathsf{is}}}(\Gcal)$ contains a complete graph, a complete bipartite graph, and a forest.

Let us first observe that the class of complete graphs and the class of complete bipartite graphs form minimal hereditary classes of unbounded degeneracy.
However this is not the case for forests and it is not possible to characterize the degeneracy parameter only by means of unbounded minimal hereditary graph classes.
The reason is that the induced subgraph relation is not a well-quasi-ordering on $\gall$, which implies the existence of \emph{infinite} strictly decreasing chains of hereditary graph classes with respect to inclusion.
Moreover, there exists such a chain composed of classes of unbounded degeneracy.
For $k\leq 3$, let $C_k$ denote the cycle of length $k$ and define $\Ccal_k = \excl_{\leqslant_{\mathsf{is}}}(\{C_i \mid i \leq k \})$.
It is easy to see that $\langle \Ccal_k \rangle_{k \geq 3}$ is an infinite strictly decreasing chain of hereditary graph classes with respect to inclusion.
It is known that for every $k \geq 3$, the chromatic number is unbounded in the class $\Ccal_k$ \cite{Erdos59Graph}.
In turn, this implies that the degeneracy parameter is unbounded in the class $\Ccal_k$.
But observe that $\bigcap_{k \geq 3} \Ccal_k$ is precisely the class $\mathcal{F}$ of forests, in which the degeneracy is bounded.
The class of forests is called a \emph{boundary class} for the degeneracy parameter.
In fact, any hereditary class defined as the intersection of an infinite strictly decreasing chain of hereditary graph classes with respect to inclusion, where all classes in the chain are unbounded for some parameter, gives rise to the concept of a \emph{limit class}.
Then, a boundary class is defined as any inclusion-minimal limit class.
This concept was introduced in \cite{Alekseev03Oneasy} and revealed to be very useful to provide obstruction characterizations of hereditary graph classes whose induced-subgraph obstruction set is finite (see \cite{Lozin23Hereditary}).

To conclude the discussion, one may quickly observe that the contrapositive of \autoref{cobs_uobs}, implies that, if $\p$ is an induced subgraph-monotone parameter for which there exists an infinite strictly decreasing chain of hereditary graph classes with respect to inclusion in which $\p$ is unbounded, $\p$ does not admit an induced subgraph-universal obstruction (in fact this is true for any quasi-ordering relation $\leqslant$ which is not a well-quasi-ordering on $\gall$).
However, we believe that this concept deserves a more systematic attention to provide characterizations of graph parameters which are monotone under some quasi-ordering relation which is not a well-quasi-ordering.
In our knowledge, nothing is known in this direction beyond the scope of the (induced) subgraph relation.

\subsection{Tree-partition width}\label{tpw_sec_survpar}

We say that a graph $H$ is a \emph{topological minor} of a graph $G$ if $G$ contains as a subgraph some subdivision of $H$ and we denote this by $\leqslant _{\mathsf{tp}}$.
Notice that if $H$ is a topological minor of $G$ then $H$ is also a minor of $G$.
However, the inverse is false.
We also stress that, the quasi-ordering $\leqslant_{\mathsf{tp}}$ is not a well-quasi-ordering in $\gall$ (see \cite{Ding96Excluding}).

A \emph{tree-partition decomposition} of a graph $G$ is a tree-cut decomposition $\mathcal{T} = (T, {\mathcal{X}})$ of $G$  where $\Xcal = \{X_{t} \subseteq V(G) \mid t \in V(T)\}$ and for every edge $e = xy$ of $T$, if $x \in X_{t}$ and $y \in X_{t'}$, then either $tt' \in E(T)$ or $t=t'$.
The \emph{width} of a tree-partition decomposition $(T, \Xcal)$ is $\max \{ |X_{t}| \mid t \in V(T) \}$.
The  of a graph $G$, denoted by $\tpw(G)$ is now defined as follows.
\begin{align}
\tpw(G) \ \coloneqq \ \min\{ k \in \Nbbb \mid \text{there is a tree-partition decomposition of }G\text{ of width }k\}.
\end{align}

Tree-partition width was introduced independently by Seese \cite{Seese85Tree} and Halin \cite{Halin91Tree} and was later studied in \cite{DingO95somer, DingO96ontre, Wood09ontre, DujmovicMW05Layout} (see also \cite{GiannopoulouKRT17packi} for an extension of $\tpw$ on multigraphs).
It is easy to see that $\tpw$ is subgraph-monotone.
On the negative side, we do not know of any other quasi-ordering relation for which $\tpw$ has some closedness property.
For instance, Ding and Oporowski \cite{DingO96ontre}, gave examples of graphs of maximum degree three where subdividing edges may result in a graph with larger tree-partition width and this implies that $\tpw$ is not monotone under minors, topological minors, or immersions.
According to \cite{DingO96ontre}, $\tpw \preceq \Delta + \tw,$ witnessed by a polynomial gap function, and moreover, $\tw \preceq \tpw,$ witnessed by a polynomial gap function as well (see also \cite{Wood09ontre}).

Interestingly, it appears that $\tpw$ has a universal obstruction with respect to the topological minor relation.
Let $\mathscr{F} = \langle\mathscr{F}_{k} \rangle_{k\in\mathbb{N}_{\geq 4}}$ where $\mathscr{F}_{k}$ is the \emph{$k$-fan}, that is the graph obtained by taking a path on $k$ vertices and making all of its vertices adjacent with a new vertex.
Let also $\mathscr{J}^{\mathsf{s}} = \langle\mathscr{J}^{\mathsf{s}}_{k} \rangle_{k\in\mathbb{N}_{\geq 1}}$ where $\mathscr{J}^{\mathsf{s}}_{k}$ is obtained from $\mathscr{J}_{k}$ (depicted in \autoref{Theta_immersions_seq}) after subdividing all its edges once.
Finally, let $\mathscr{P}^{\mathsf{ms}} = \langle \mathscr{P}^{\mathsf{ms}}_{k} \rangle_{k\in\mathbb{N}_{\geq 1}}$ be the graph obtained from a path on $k+1$ vertices if we increase the multiplicity of each edge to $k$ and then subdivide once each resulting edge (see \autoref{fig_dingraph} for examples of the defined graphs).
Also, recall that $\mathscr{W} = \langle \mathscr{W}_{k} \rangle_{k\in\mathbb{N}}$ where $\mathscr{W}_{k}$ is the wall of order $k$ depicted in \autoref{wsalls}.

\begin{figure}[ht]
\begin{center}
\bigskip
\begin{tikzpicture}[thick,scale=0.47]
\tikzstyle{sommet}=[circle, draw, fill=black, inner sep=0pt, minimum width=3.6pt]
\tikzstyle{newsommet}=[circle, draw, fill=green!50!black, inner sep=0pt, minimum width=4pt]

\begin{scope}[xshift=-5cm,yshift=0cm]

\node[] (c) at (0:0) {} ;
\draw[] (c) node[sommet]{};
\foreach \i in {1,...,7}{
	\node[] (\i) at (4-\i,-2) {} ;
	\draw[] (\i) node[sommet]{};
	\draw (c.center) -- (\i.center) ;
}

\draw (1.center) -- (2.center) ;
\draw (2.center) -- (3.center) ;
\draw (3.center) -- (4.center) ;
\draw (4.center) -- (5.center) ;
\draw (5.center) -- (6.center) ;
\draw (6.center) -- (7.center) ;

\node (F7) at (0,-3){$\mathscr{F}_{7}$};
\end{scope}

\begin{scope}[xshift=2cm,yshift=-1cm]
\node[] (c) at (0:0) {} ;
\draw[] (c) node[sommet]{};

\node[] (N) at (0,2.5) {} ;
\draw[] (N) node[sommet]{};

\node[] (W) at (-2.5,0) {} ;
\draw[] (W) node[sommet]{};

\node[] (S) at (0,-2.5) {} ;
\draw[] (S) node[sommet]{};

\node[] (E) at (2.5,0) {} ;
\draw[] (E) node[sommet]{};

\foreach \i in {0,...,3}{
 	\node[] (\i) at (-1.25,-0.75+\i*0.5) {} ;
	\draw[] (\i) node[sommet]{};
    \draw (c.center) -- (\i.center) ;
    \draw (W.center) -- (\i.center) ;
    }

\foreach \i in {0,...,3}{
 	\node[] (\i) at (-0.75+\i*0.5,1.25) {} ;
	\draw[] (\i) node[sommet]{};
    \draw (c.center) -- (\i.center) ;
    \draw (N.center) -- (\i.center) ;
    }

\foreach \i in {0,...,3}{
 	\node[] (\i) at (1.25,-0.75+\i*0.5) {} ;
	\draw[] (\i) node[sommet]{};
    \draw (c.center) -- (\i.center) ;
    \draw (E.center) -- (\i.center) ;
    }

\foreach \i in {0,...,3}{
 	\node[] (\i) at (-0.75+\i*0.5,-1.25) {} ;
	\draw[] (\i) node[sommet]{};
    \draw (c.center) -- (\i.center) ;
    \draw (S.center) -- (\i.center) ;
    }

    \node (Js4) at (0,-3.5){$\mathscr{J}_{4}^{\sf s}$};

\end{scope}

\begin{scope}[xshift=6cm,yshift=-1cm]

\foreach \i in {0,...,4}{
 	\node[] (n\i) at (2.5*\i,0) {} ;
	\draw[] (n\i) node[sommet]{};
    }

\foreach \j in {1,3,5,7}{
    \foreach \i in {0,...,3}{
        \node[] (\j\i) at (\j*1.25,-0.75+\i*0.5) {} ;
	    \draw[] (\j\i) node[sommet]{};
        }
    }

\foreach \j in {0,1,2,3}{
    \draw (1\j.center) -- (n0.center) ;
    \draw (1\j.center) -- (n1.center) ;
    }

\foreach \j in {0,1,2,3}{
    \draw (3\j.center) -- (n1.center) ;
    \draw (3\j.center) -- (n2.center) ;
    }

\foreach \j in {0,1,2,3}{
    \draw (5\j.center) -- (n2.center) ;
    \draw (5\j.center) -- (n3.center) ;
    }

\foreach \j in {0,1,2,3}{
    \draw (7\j.center) -- (n3.center) ;
    \draw (7\j.center) -- (n4.center) ;
    }
    \node (Pms4) at (5,-1.5){$\mathscr{P}_{4}^{\sf ms}$};

\end{scope}
\end{tikzpicture}
\end{center}
\vspace{-0.5cm}
\caption{\label{fig_dingraph}Examples of the topological minor-universal obstructions for tree-partition width.
}
\end{figure}


Ding and  Oporowski \cite{DingO96ontre} proved that there exists a function $f \colon \nton$ such that every graph $G$ where $\tpw(G) \geq f(k)$ contains as a topological minor one of the graphs $\mathscr{F}_{k}$, $\mathscr{J}^{\mathsf{s}}_{k}$, $\mathscr{P}^{\mathsf{ms}}_{k}$, or $\mathscr{W}_{k}$.
This means that the result of \cite{DingO96ontre} can be restated as follows.

\begin{theorem}\label{ismsdwssaal} 
The set   $\{ \mathscr{F}, \mathscr{J}^{\mathsf{s}}, \mathscr{P}^{\mathsf{ms}}, \mathscr{W} \}$ is a topological-minor-universal obstruction for $\tpw$.
\end{theorem}

Note that \autoref{ismsdwssaal} gives a max-min analogue of tree-partition width for the topological minor relation under which tree-partition width is not monotone.
Motivated by \autoref{ismsdwssaal}, we think it is an interesting problem to find some quasi-ordering for which tree-partition is monotone and moreover admits a universal obstruction characterization.

\subsection{Edge-treewidth}\label{multedge}

As we already mentioned, it is a running challenge, given a graph parameter $\p$, to properly identify a quasi-ordering relation on graphs for which $\p$ is monotone.
Edge-treewidth is a typical example of this approach that we believe is worth to present here.
Let $G$ be a graph that may contain parallel edges but no loops.
The \emph{edge-treewidth} of $G$, denoted by $\etw(G)$, is defined as follows.
\begin{eqnarray}
\begin{minipage}{12cm}
$\etw(G)$ is the minimum $k \in \Nbbb$ for which there exists an ordering $v_{1}, \ldots, v_{n}$ of $V(G)$ such that for every $i \in [n]$ there are at most $k$ edges with one endpoint in $\{ v_{1}, \ldots, v_{i - 1}\}$ and the other in the connected component of $G[\{ v_{i}, \ldots, v_{n} \}]$ that contains $v_{i}$.
\end{minipage}
\end{eqnarray}

Edge-treewidth was defined in \cite{MagnePaulSharmaThilikos2023EdgeTreewidth} as a tree-like analogue of cutwidth.
We define the \emph{block-degree} of a graph as the maximum degree of its blocks and we denote it by $\Delta^{\mathsf{b}}$.
According to \cite{MagnePaulSharmaThilikos2023EdgeTreewidth}, $\etw \sim_{\mathsf{P}} \tw + \Delta^{\mathsf{b}}$.
Also there are counterexamples, described in \cite{MagnePaulSharmaThilikos2023EdgeTreewidth}, showing that $\etw$ is not monotone under topological minors, immersions, or minors.
However it is subgraph-monotone.
For the study of $\etw$, a new quasi-ordering relation was defined in \cite{MagnePaulSharmaThilikos2023EdgeTreewidth}, as follows.
 A graph $H$ is a \emph{weak-topological-minor} of a graph $G$, denoted by $H \leqslant_{\mathsf{wtp}} G$, if $H$ is obtained from a subgraph of $G$ by contracting edges whose endpoints are both incident with two vertices and moreover have degree two (see \autoref{weaktopminor}).
 We observe that the cycle on two vertices is a weak-topological-minor of the cycle on three vertices (and henceforth of every chordless cycle).
 
\begin{figure}[htbp]
\centering
\includegraphics[width=0.53\linewidth]{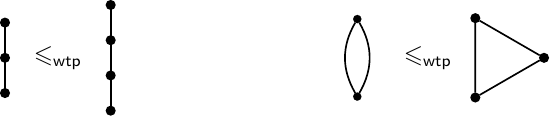}
\caption{\label{weaktopminor}The path $P_3$ is a weak-topological-minor of the path $P_4$.
The double edge graph $ \Theta_{2}$ is a weak-topological-minor of the cycle $C_3$.}
\end{figure}

It was proved in \cite{MagnePaulSharmaThilikos2023EdgeTreewidth} that $\etw$ is monotone under weak-topological-minors.
Moreover, \cite{MagnePaulSharmaThilikos2023EdgeTreewidth} gave a weak-topological-minor-universal obstruction for $\etw$ with respect to weak topological minors as follows.
In \autoref{lesstrcas} we defined $ \Theta = \langle \Theta _{k} \rangle_{k\in\mathbb{N}}$, so that $\Theta _{k}$ is the graph on two vertices and $k$ parallel edges.
We also define $\Theta^{\mathsf{s}} = \langle \Theta _{k}^{\mathsf{s}} \rangle_{k\in\mathbb{N}}$, so that $\Theta _{k}^{\mathsf{s}} = K_{2,k}$.
Next, we define $\mathscr{F}^{\mathsf{s}} = \langle\mathscr{F}_{k}^{\mathsf{s}}\rangle_{k \in \mathbb{N}_{\geq 4}}$ where $\mathscr{F}_{k}^{\mathsf{s}}$ is the graph obtained from a $k$-fan $\mathscr{F}_{k}$ if we subdivide all edges with both endpoints of degree three once.
Also consider $\mathscr{F}^{\mathsf{ss}} = \langle\mathscr{F}_{k}^{\mathsf{ss}} \rangle_{k\in\mathbb{N}_{\geq 4}}$ where $\mathscr{F}_{k}^{\mathsf{ss}}$ is the graph obtained from a $\mathscr{F}_{k}^{\mathsf{s}}$ if we subdivide all edges with both endpoints of degree at least three once.
Finally, let $\mathscr{W}^{\mathsf{s}} = \langle \mathscr{W}_{k}^{\mathsf{s}} \rangle_{k\in\mathbb{N}}$ where $\mathscr{W}_{k}^{
\mathsf{s}}$ is obtained from the wall of order $k$, after subdividing all edges whose both endpoints have degree three once.
See \autoref{fig_universal_obstructions} for an illustration of the previous definitions.

\begin{figure}[htbp]
\centering
\includegraphics[width=0.9\linewidth]{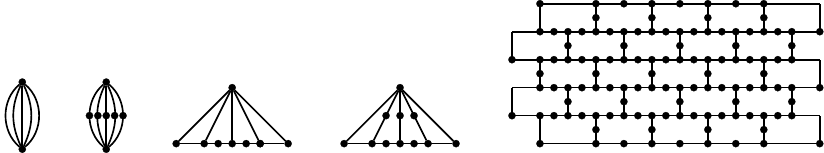}
\caption{\label{fig_universal_obstructions}From left to right, the graphs $ \Theta _{5}$, $ \Theta _{5}^{\mathsf{s}}$, $\mathscr{F}_{5}^{\mathsf{s}}$, $\mathscr{F}_{5}^{\mathsf{ss}}$, and $\mathscr{W}_{6}^{
\mathsf{s}}$.}
\end{figure}

We are now in position to restate the results of \cite{MagnePaulSharmaThilikos2023EdgeTreewidth} as follows.

\begin{theorem}\label{ismssdwshsaal} 
The set $\{\Theta , \Theta^{\mathsf{s}}, \mathscr{F}^{\mathsf{s}}, \mathscr{F}^{\mathsf{ss}}, \mathscr{W}^{
\mathsf{s}}\}$ is a weak-topological-minor-universal obstruction for $\etw$.
\end{theorem}

\section{Conclusion and open directions}

Clearly the list of parameters that we examined in this work is far from being complete.
We expect that more results are known or will appear that can be stated in terms of the unified framework of universal obstructions.
We conclude with some open problems and directions of research on universal obstructions.

\medskip
According to \cite{PaulPT2023GraphParameters}, given that $\leqslant$ is a well-quasi-ordering on graphs, we know that $\pobs_{\leqslant}(\p)$ exists for every $\leqslant $-monotone parameter $\p$.
However, well-quasi-ordering does not imply that $\pobs_{\leqslant }(\p)$ is finite.
As observed in \cite{PaulPT2023GraphParameters}, this is linked to a more powerful order-theoretic condition, namely that $\leqslant $ is an $\omega^2$-well-quasi-ordering on graphs.
While this has been conjectured for the minor and the immersion relation, a proof appears to be still far from reach (see \cite{Thomas1989wellquasi} for the best result in this direction).
It is an interesting question whether the finiteness of $\pobs_{\leqslant}(\p)$ can be proved for particular families of graph parameters using constructive arguments.
This was already possible using the results on Erdős-Pósa dualities \cite{CamesHJR19Atight, RobertsonS86GMV, PaulPTW24Obstructions, PaulPTW24Delineating} for vertex removal parameters in \autoref{appe4jnfbvjpsjdpers} (\autoref{prop_erdos_posa_planar}) and for elimination distance parameters in \autoref{apeplicx} (\autoref{more_erdos_posa_planar}).
Is it possible to prove similar results about the existence of finite universal obstructions for other families of parameters?
This kind of general condition have been investigated for bounding the size of obstructions of minor-closed graph classes (see e.g., \cite{SauST23apices,LagergrenA91mini,Lagergren98,AdlerGK08comp} and it is a challenge up to which point they may be used to derive  bounds on parametric obstructions.

Note that all definitions of this work are based on the general equivalence relation between parameters.
What about more refined notions of equivalence, such as polynomial equivalence or linear equivalence (defined in \autoref{sliubeabaueike})?
The natural question is the following: suppose that $\p$ is a $\leqslant$-monotone parameter.
Is there a universal obstruction $\mathfrak{H}$ such that $\p_{\mathfrak{H}} \sim_{\mathsf{P}} \p$?
As we already observed in \autoref{treewidth_sec}, grids or annuli give a positive answer to this question for the case of treewidth, however if we ask for a linear dependence, i.e., $\p_{\mathfrak{H}} \sim_{\mathsf{L}} \p$, we know that the answer is negative.
Additionally, as we have seen in \autoref{treedepth_par}, the answer to this question is negative for $\td,$ even when we ask for a polynomial dependence.
Nonetheless, \autoref{mystery_theorem} shows that such a polynomial dependence for $\td$ is possible when we do not demand universal obstructions to be a $\lesssim$-antichain.
A running challenge is to build a theory revealing when universal obstructions of polynomial / linear dependence may exist (and when not).
What would be the algorithmic implications of such a theory?

All universal obstructions that we presented (grids, complete ternary trees, paths, etc.) enjoy some sort of regularity in their definition.
Is there a way to further formalize this so as  to obtain characterizations of parameters in terms of \textsl{finite descriptions} of their universal obstructions, instead of parametric obstructions?

Universal / parametric / class obstructions provide a convenient way to compare (in the approximate sense) graph parameters, via the relations $\subseteq^*, \leqslant ^{**}, \lesssim^*$, displayed in \autoref{conbncpsjdsusion}.
This also permits us to define hierarchies of families of graph parameters by fixing first their obstructions.
For instance, we may ask what is the parameter whose class obstruction consists of the set of all caterpillars.
As we have seen in \autoref{threqe}, this corresponds to the elimination distance to a forest of paths.
Such a guess was possible by observing that caterpillars is a path-like structure, which indicates that the corresponding graph parameter should have an elimination ordering definition.
Similar guesses can be done for other (non depicted) parameters in the partial-ordering of \autoref{fig_minor_closed}.
An interesting question would be to find general schemes of elimination distance or tree/path decompositions (or combinations of them) that automatically can be derived by the corresponding universal / parametric / class obstructions.
This research direction becomes much more challenging when dealing with parameters that, in \autoref{fig_minor_closed} are ``below $\tw$''. In these cases the derived decomposition theorems are essentially refinements of the Graph Minors Structure Theorem, by Robertson and Seymour in \cite{RobertsonS03GMXVI}.

We conclude with the remark that all the parameters we presented are parameters on graphs.
Certainly the same concepts of obstructions can be defined also on parameters on other combinatorial structures such as matroids, directed graphs, graphs with matchings (see \cite{HatzelRabinovichWiederrecht2019Cyclewidth, KawarabayashiK15TheDirected, GeelenGW09Tangles} for examples of such results concerning extensions of treewidth).
However their exposition escapes the objectives of the present work.

\paragraph{Acknowledgements.}

We are thankful to Sebastian Wiederrecht for fruitful discussions on the topics of this work.
We also thank Hugo Jacob for helpful comments on an earlier version of this work.
We would also wish to thank the anonymous reviewers for their valuable comments and suggestions on an earlier version of this paper, which have significantly improved it.

\newpage

\bibliographystyle{plainurl}


\begin{thebibliography}{100}

\bibitem{AdlerFP14Obstructions}
Isolde Adler, Arthur~M. Farley, and Andrzej Proskurowski.
\newblock Obstructions for linear rank-width at most 1.
\newblock {\em Discret. Appl. Math.}, 168:3--13, 2014.
\newblock \href {https://doi.org/10.1016/j.dam.2013.05.001}
  {\path{doi:10.1016/j.dam.2013.05.001}}.

\bibitem{AdlerGK08comp}
Isolde Adler, Martin Grohe, and Stephan Kreutzer.
\newblock Computing excluded minors.
\newblock In {\em Proceedings of the {N}ineteenth {A}nnual {ACM}-{SIAM}
  {S}ymposium on {D}iscrete {A}lgorithms}, pages 641--650. ACM, New York, 2008.
\newblock URL: \url{https://dl.acm.org/doi/10.5555/1347082.1347153}.

\bibitem{AlecuKLZ20Between}
Bogdan Alecu, Mamadou~Moustapha Kant\'e, Vadim Lozin, and Viktor Zamaraev.
\newblock Between clique-width and linear clique-width of bipartite graphs.
\newblock {\em Discrete Math.}, 343(8):111926, 14, 2020.
\newblock \href {https://doi.org/10.1016/j.disc.2020.111926}
  {\path{doi:10.1016/j.disc.2020.111926}}.

\bibitem{Alekseev03Oneasy}
Vladimir~E. Alekseev.
\newblock On easy and hard hereditary classes of graphs with respect to the
  independent set problem.
\newblock {\em Discrete Applied Mathematics}, 132(1):17--26, 2003.
\newblock Stability in Graphs and Related Topics.
\newblock \href {https://doi.org/10.1016/S0166-218X(03)00387-1}
  {\path{doi:10.1016/S0166-218X(03)00387-1}}.

\bibitem{Archdeacon81AKuratowski}
Dan Archdeacon.
\newblock A {K}uratowski theorem for the projective plane.
\newblock {\em J. Graph Theory}, 5(3):243--246, 1981.
\newblock \href {https://doi.org/10.1002/jgt.3190050305}
  {\path{doi:10.1002/jgt.3190050305}}.

\bibitem{ArchdeaconH89akuratowski}
Dan Archdeacon and Phil Huneke.
\newblock A kuratowski theorem for nonorientable surfaces.
\newblock {\em J. Comb. Theory, Ser. {B}}, 46(2):173--231, 1989.
\newblock \href {https://doi.org/10.1016/0095-8956(89)90043-9}
  {\path{doi:10.1016/0095-8956(89)90043-9}}.

\bibitem{AtminasBLS21Minimal}
A.~Atminas, R.~Brignall, V.~Lozin, and J.~Stacho.
\newblock Minimal classes of graphs of unbounded clique-width defined by
  finitely many forbidden induced subgraphs.
\newblock {\em Discrete Appl. Math.}, 295:57--69, 2021.
\newblock \href {https://doi.org/10.1016/j.dam.2021.02.007}
  {\path{doi:10.1016/j.dam.2021.02.007}}.

\bibitem{AtminasL22Adichotomy}
A.~Atminas and V.~Lozin.
\newblock A dichotomy for graphs of bounded degeneracy, 2022.
\newblock URL: \url{https://arxiv.org/abs/2206.09252}, \href
  {https://arxiv.org/abs/2206.09252} {\path{arXiv:2206.09252}}.

\bibitem{BellenbaumD02Two}
Patrick Bellenbaum and Reinhard Diestel.
\newblock Two short proofs concerning tree-decompositions.
\newblock {\em Combinatorics, Probability \& Computing}, 11(6):541--547, 2002.
\newblock \href {https://doi.org/10.1017/S0963548302005369}
  {\path{doi:10.1017/S0963548302005369}}.

\bibitem{BerteleB72nonse}
Umberto Bertel\'{e} and Francesco Brioschi.
\newblock {\em Nonserial Dynamic Programming}.
\newblock Mathematics in Science and Engineering, Vol. 91. Academic Press, New
  York and London, 1972.

\bibitem{BienstockRST91Quickly}
Daniel Bienstock, Neil Robertson, Paul~D. Seymour, and Robin Thomas.
\newblock Quickly excluding a forest.
\newblock {\em J. Comb. Theory, Ser. {B}}, 52(2):274--283, 1991.
\newblock \href {https://doi.org/10.1016/0095-8956(91)90068-U}
  {\path{doi:10.1016/0095-8956(91)90068-U}}.

\bibitem{Bodlaender98}
Hans~L. Bodlaender.
\newblock A partial $k$-arboretum of graphs with bounded treewidth.
\newblock {\em Theor. Comput. Sci.}, 209(1-2):1--45, 1998.
\newblock \href {https://doi.org/10.1016/S0304-3975(97)00228-4}
  {\path{doi:10.1016/S0304-3975(97)00228-4}}.

\bibitem{BojanczykO2025Graphs}
Mikolaj Bojanczyk and Pierre Ohlmann.
\newblock Graphs of unbounded linear cliquewidth must transduce all trees.
\newblock {\em CoRR}, abs/2501.17556, 2025.
\newblock URL: \url{https://doi.org/10.48550/arXiv.2501.17556}, \href
  {https://arxiv.org/abs/2501.17556} {\path{arXiv:2501.17556}}, \href
  {https://doi.org/10.48550/ARXIV.2501.17556}
  {\path{doi:10.48550/ARXIV.2501.17556}}.

\bibitem{Bouchet94circl}
Andr\'e Bouchet.
\newblock Circle graph obstructions.
\newblock {\em Journal of Combinatorial Theory Series B}, 60:107--144, 1994.
\newblock \href {https://doi.org/10.1006/jctb.1994.1008}
  {\path{doi:10.1006/jctb.1994.1008}}.

\bibitem{BrignallC22Uncountably}
R.~Brignall and D.~Cocks.
\newblock Uncountably many minimal hereditary classes of graphs of unbounded
  clique-width.
\newblock {\em Electron. J. Combin.}, 29(1):Paper No. 1.63, 27, 2022.
\newblock \href {https://doi.org/10.37236/10483} {\path{doi:10.37236/10483}}.

\bibitem{BrignallC23Aframework}
R.~Brignall and D.~Cocks.
\newblock A framework for minimal hereditary classes of graphs of unbounded
  clique-width.
\newblock {\em SIAM J. Discrete Math.}, 37(4):2558--2584, 2023.
\newblock \href {https://doi.org/10.1137/22M1487448}
  {\path{doi:10.1137/22M1487448}}.

\bibitem{BulianD16graph}
Jannis Bulian and Anuj Dawar.
\newblock Graph isomorphism parameterized by elimination distance to bounded
  degree.
\newblock {\em Algorithmica}, 75(2):363--382, 2016.
\newblock \href {https://doi.org/10.1007/s00453-015-0045-3}
  {\path{doi:10.1007/s00453-015-0045-3}}.

\bibitem{BulianD17Fixed}
Jannis Bulian and Anuj Dawar.
\newblock Fixed-parameter tractable distances to sparse graph classes.
\newblock {\em Algorithmica}, 79(1):139--158, 2017.
\newblock \href {https://doi.org/10.1007/s00453-016-0235-7}
  {\path{doi:10.1007/s00453-016-0235-7}}.

\bibitem{CamesHJR19Atight}
Wouter Cames~van Batenburg, Tony Huynh, Gwena{\"e}l Joret, and Jean-Florent
  Raymond.
\newblock A tight {E}rd{\H o}s-{P}{\'o}sa function for planar minors.
\newblock {\em Adv. Comb.}, pages Paper No. 2, 33, 2019.
\newblock \href {https://doi.org/10.19086/aic.10807}
  {\path{doi:10.19086/aic.10807}}.

\bibitem{ChekuriC16Polynomial}
Chandra Chekuri and Julia Chuzhoy.
\newblock Polynomial bounds for the grid-minor theorem.
\newblock {\em Journal of ACM}, 63(5):40:1--40:65, 2016.
\newblock \href {https://doi.org/10.1145/2820609} {\path{doi:10.1145/2820609}}.

\bibitem{ChungS89Graphswith}
Fan R.~K. Chung and Paul~D. Seymour.
\newblock Graphs with small bandwidth and cutwidth.
\newblock {\em Discret. Math.}, 75(1-3):113--119, 1989.
\newblock \href {https://doi.org/10.1016/0012-365X(89)90083-6}
  {\path{doi:10.1016/0012-365X(89)90083-6}}.

\bibitem{ChuzhoyT21Towards}
Julia Chuzhoy and Zihan Tan.
\newblock Towards tight(er) bounds for the excluded grid theorem.
\newblock {\em J. Comb. Theory, Ser. {B}}, 146:219--265, 2021.
\newblock \href {https://doi.org/10.1016/j.jctb.2020.09.010}
  {\path{doi:10.1016/j.jctb.2020.09.010}}.

\bibitem{CollinsFKLZ18Infinitely}
A.~Collins, J.~Foniok, N.~Korpelainen, V.~Lozin, and V.~Zamaraev.
\newblock Infinitely many minimal classes of graphs of unbounded clique-width.
\newblock {\em Discrete Appl. Math.}, 248:145--152, 2018.
\newblock \href {https://doi.org/10.1016/j.dam.2017.02.012}
  {\path{doi:10.1016/j.dam.2017.02.012}}.

\bibitem{Courcelle90them}
Bruno Courcelle.
\newblock The monadic second-order logic of graphs. {I}. recognizable sets of
  finite graphs.
\newblock {\em Information and Computation}, 85(1):12--75, 1990.
\newblock \href {https://doi.org/10.1016/0890-5401(90)90043-H}
  {\path{doi:10.1016/0890-5401(90)90043-H}}.

\bibitem{Courcelle92}
Bruno Courcelle.
\newblock The monadic second-order logic of graphs {III:} tree-decompositions,
  minor and complexity issues.
\newblock {\em {RAIRO} - Theoretical Informatics and Applications},
  26:257--286, 1992.
\newblock \href {https://doi.org/10.1051/ita/1992260302571}
  {\path{doi:10.1051/ita/1992260302571}}.

\bibitem{Courcelle97}
Bruno Courcelle.
\newblock The expression of graph properties and graph transformations in
  monadic second-order logic.
\newblock In {\em Handbook of Graph Grammars and Computing by Graph
  Transformations, Volume 1: Foundations}, pages 313--400. World Scientific,
  1997.

\bibitem{CourcelleER93}
Bruno Courcelle, Joost Engelfriet, and Grzegorz Rozenberg.
\newblock Handle-rewriting hypergraph grammars.
\newblock {\em J. Comput. Syst. Sci.}, 46(2):218--270, 1993.
\newblock \href {https://doi.org/10.1016/0022-0000(93)90004-G}
  {\path{doi:10.1016/0022-0000(93)90004-G}}.

\bibitem{CourcelleMR00Linear}
Bruno Courcelle, Johann~A. Makowsky, and Udi Rotics.
\newblock Linear time solvable optimization problems on graphs of bounded
  clique-width.
\newblock {\em Theory of Computing Systems}, 33(2):125--150, 2000.
\newblock \href {https://doi.org/10.1007/s002249910009}
  {\path{doi:10.1007/s002249910009}}.

\bibitem{dang2018minors}
Thanh~N. Dang and Robin Thomas.
\newblock Minors of two-connected graphs of large path-width, 2018.
\newblock \href {https://arxiv.org/abs/1712.04549} {\path{arXiv:1712.04549}}.

\bibitem{DekkerJansen2024FVS}
David J.~C. Dekker and Bart M.~P. Jansen.
\newblock Kernelization for feedback vertex set via elimination distance to a
  forest.
\newblock {\em Discrete Applied Mathematics}, 346:192--214, 2024.
\newblock \href {https://doi.org/10.1016/j.dam.2023.12.016}
  {\path{doi:10.1016/j.dam.2023.12.016}}.

\bibitem{DemaineHNRT04approximation}
Erik~D. Demaine, Mohammad~Taghi Hajiaghayi, Naomi Nishimura, Prabhakar Ragde,
  and Dimitrios~M. Thilikos.
\newblock Approximation algorithms for classes of graphs excluding
  single-crossing graphs as minors.
\newblock {\em J. Comput. Syst. Sci.}, 69(2):166--195, 2004.
\newblock \href {https://doi.org/10.1016/j.jcss.2003.12.001}
  {\path{doi:10.1016/j.jcss.2003.12.001}}.

\bibitem{DemaineHT05exponential}
Erik~D. Demaine, Mohammad~Taghi Hajiaghayi, and Dimitrios~M. Thilikos.
\newblock Exponential speedup of fixed-parameter algorithms for classes of
  graphs excluding single-crossing graphs as minors.
\newblock {\em Algorithmica}, 41(4):245--267, 2005.
\newblock \href {https://doi.org/10.1007/s00453-004-1125-y}
  {\path{doi:10.1007/s00453-004-1125-y}}.

\bibitem{DemaineHK09AlgorithmicGraph}
Erik~D. Demaine, MohammadTaghi Hajiaghayi, and Ken{-}ichi Kawarabayashi.
\newblock Algorithmic graph minor theory: Improved grid minor bounds and
  {W}agner's contraction.
\newblock {\em Algorithmica}, 54(2):142--180, 2009.
\newblock \href {https://doi.org/10.1007/s00453-007-9138-y}
  {\path{doi:10.1007/s00453-007-9138-y}}.

\bibitem{DevosKO20Branchdepth}
Matt DeVos, O{-}joung Kwon, and {Sang-il} Oum.
\newblock Branch-depth: generalizing tree-depth of graphs.
\newblock {\em European Journal of Combinatorics}, 90:103186, 2020.
\newblock \href {https://doi.org/10.1016/j.ejc.2020.103186}
  {\path{doi:10.1016/j.ejc.2020.103186}}.

\bibitem{Diestel10grap}
Reinhard Diestel.
\newblock {\em {Graph Theory}}, volume 173.
\newblock Springer-Verlag, 5th edition, 2017.
\newblock \href {https://doi.org/10.1007/978-3-662-53622-3}
  {\path{doi:10.1007/978-3-662-53622-3}}.

\bibitem{DinerGST22Block}
{\"{O}}znur~Yasar Diner, Archontia~C. Giannopoulou, Giannos Stamoulis, and
  Dimitrios~M. Thilikos.
\newblock Block elimination distance.
\newblock {\em Graphs Comb.}, 38(5):133, 2022.
\newblock \href {https://doi.org/10.1007/s00373-022-02513-y}
  {\path{doi:10.1007/s00373-022-02513-y}}.

\bibitem{Ding96Excluding}
Guoli Ding.
\newblock Excluding a long double path minor.
\newblock {\em J. Combin. Theory Ser. B}, 66(1):11--23, 1996.
\newblock \href {https://doi.org/10.1006/jctb.1996.0002}
  {\path{doi:10.1006/jctb.1996.0002}}.

\bibitem{DingD16excluded}
Guoli Ding and Stan Dziobiak.
\newblock Excluded-minor characterization of apex-outerplanar graphs.
\newblock {\em Graphs Comb.}, 32(2):583--627, 2016.
\newblock \href {https://doi.org/10.1007/s00373-015-1611-9}
  {\path{doi:10.1007/s00373-015-1611-9}}.

\bibitem{DingO95somer}
Guoli Ding and Bogdan Oporowski.
\newblock Some results on tree decomposition of graphs.
\newblock {\em Journal of Graph Theory}, 20(4):481--499, 1995.
\newblock \href {https://doi.org/10.1002/jgt.3190200412}
  {\path{doi:10.1002/jgt.3190200412}}.

\bibitem{DingO96ontre}
Guoli Ding and Bogdan Oporowski.
\newblock On tree-partitions of graphs.
\newblock {\em Discrete Mathematics}, 149(1-3):45--58, 1996.
\newblock \href {https://doi.org/10.1016/0012-365X(94)00337-I}
  {\path{doi:10.1016/0012-365X(94)00337-I}}.

\bibitem{Dinneen97}
Michael~J. Dinneen.
\newblock Too many minor order obstructions (for parameterized lower ideals).
\newblock {\em Journal of Universal Computer Science}, 3(11):1199--1206, 1997.
\newblock \href {https://doi.org/10.3217/jucs-003-11}
  {\path{doi:10.3217/jucs-003-11}}.

\bibitem{DinneenCF0fForbidden}
Michael~J. Dinneen, Kevin Cattell, and Michael~R. Fellows.
\newblock Forbidden minors to graphs with small feedback sets.
\newblock {\em Discret. Math.}, 230(1-3):215--252, 2001.
\newblock \href {https://doi.org/10.1016/S0012-365X(00)00083-2}
  {\path{doi:10.1016/S0012-365X(00)00083-2}}.

\bibitem{DornT09Semi}
Frederic Dorn and Jan~Arne Telle.
\newblock Semi-nice tree-decompositions: The best of branchwidth, treewidth and
  pathwidth with one algorithm.
\newblock {\em Discrete Applied Mathematics}, 157(12):2737--2746, 2009.
\newblock \href {https://doi.org/10.1016/j.dam.2008.08.023}
  {\path{doi:10.1016/j.dam.2008.08.023}}.

\bibitem{DowneyF99Parameterized}
Rod~G. Downey and Michael~R. Fellows.
\newblock {\em Parameterized complexity}.
\newblock Springer, 1999.

\bibitem{DujmovicMW05Layout}
Vida Dujmovi\'c, Pat Morin, and David~R. Wood.
\newblock Layout of graphs with bounded tree-width.
\newblock {\em SIAM J. Comput.}, 34(3):553--579, 2005.
\newblock \href {https://doi.org/10.1137/S0097539702416141}
  {\path{doi:10.1137/S0097539702416141}}.

\bibitem{DumasLopez2026WellQuasiOrdered}
Ma{\"e}l Dumas and Aliaume Lopez.
\newblock Well-quasi-ordered classes of bounded clique-width, 2026.
\newblock \href {https://arxiv.org/abs/2601.18571} {\path{arXiv:2601.18571}},
  \href {https://doi.org/10.48550/arXiv.2601.18571}
  {\path{doi:10.48550/arXiv.2601.18571}}.

\bibitem{Dyck1888Beitrage}
Walther Dyck.
\newblock Beitr\"{a}ge zur {A}nalysis situs.
\newblock {\em Math. Ann.}, 32(4):457--512, 1888.
\newblock \href {https://doi.org/10.1007/BF01443580}
  {\path{doi:10.1007/BF01443580}}.

\bibitem{EibenGHJK22AUnifying}
Eduard Eiben, Robert Ganian, Thekla Hamm, Lars Jaffke, and O{-}joung Kwon.
\newblock {A Unifying Framework for Characterizing and Computing Width
  Measures}.
\newblock In {\em Innovations in Theoretical Computer Science Conference
  (ITCS)}, volume 215 of {\em Leibniz International Proceedings in Informatics
  (LIPIcs)}, pages 63:1--63:23, 2022.
\newblock \href {https://doi.org/10.4230/LIPIcs.ITCS.2022.63}
  {\path{doi:10.4230/LIPIcs.ITCS.2022.63}}.

\bibitem{Erdos59Graph}
P.~Erd{\H o}s.
\newblock Graph theory and probability.
\newblock {\em Canadian J. Math.}, 11:34--38, 1959.
\newblock \href {https://doi.org/10.4153/CJM-1959-003-9}
  {\path{doi:10.4153/CJM-1959-003-9}}.

\bibitem{FellowsL92Onwell}
Michael~R. Fellows and Micheal~A. Langston.
\newblock On well-partial-order theory and its application to combinatoria
  problems of {VLSI} design.
\newblock {\em {SIAM} Journal on Discrete Mathematics}, 5(1):117--126, 1992.
\newblock \href {https://doi.org/10.1137/0405010.}
  {\path{doi:10.1137/0405010.}}

\bibitem{Francis99ConwayZIP}
George~K. Francis and Jeffrey~R. Weeks.
\newblock Conway's {ZIP} proof.
\newblock {\em Amer. Math. Monthly}, 106(5):393--399, 1999.
\newblock \href {https://doi.org/10.2307/2589143} {\path{doi:10.2307/2589143}}.

\bibitem{GanianHKMORS16Are}
Robert Ganian, Petr Hlin{\v e}n{\'y}, Joachim Kneis, Daniel Meister, Jan
  Obdr{\v z}{\'a}lek, Peter Rossmanith, and Somnath Sikdar.
\newblock Are there any good digraph width measures?
\newblock {\em J. Combin. Theory Ser. B}, 116:250--286, 2016.
\newblock \href {https://doi.org/10.1016/j.jctb.2015.09.001}
  {\path{doi:10.1016/j.jctb.2015.09.001}}.

\bibitem{GanianHNOM2019ShrubDepth}
Robert Ganian, Petr Hlin{\v{e}}n{\'y}, Jaroslav Ne{\v{s}}et{\v{r}}il, Jan
  Obdr{\v{z}}{\'a}lek, and Patrice~Ossona de~Mendez.
\newblock Shrub-depth: Capturing height of dense graphs.
\newblock {\em Logical Methods in Computer Science}, 15(1):7:1--7:25, 2019.
\newblock \href {https://arxiv.org/abs/1707.00359} {\path{arXiv:1707.00359}},
  \href {https://doi.org/10.23638/LMCS-15(1:7)2019}
  {\path{doi:10.23638/LMCS-15(1:7)2019}}.

\bibitem{GanianKorchemna2024SlimTreeCutWidth}
Robert Ganian and Viktoriia Korchemna.
\newblock Slim tree-cut width.
\newblock {\em Algorithmica}, 86(8):2714--2738, 2024.
\newblock \href {https://doi.org/10.1007/s00453-024-01241-4}
  {\path{doi:10.1007/s00453-024-01241-4}}.

\bibitem{GareyJ79Computers}
Micheal~R. Garey and David~S. Johnson.
\newblock {\em Computers and Intractability: {A} Guide to the Theory of
  NP-Completeness}.
\newblock W.H. Freeman, 1979.

\bibitem{GeelenGW09Tangles}
Jim Geelen, Bert Gerards, and Geoff Whittle.
\newblock Tangles, tree-decompositions and grids in matroids.
\newblock {\em Journal of Combinatorial Theory, Series B}, 99(4):657--667,
  2009.
\newblock \href {https://doi.org/10.1016/j.jctb.2007.10.008}
  {\path{doi:10.1016/j.jctb.2007.10.008}}.

\bibitem{geelen2016generalization}
Jim Geelen and Benson Joeris.
\newblock A generalization of the grid theorem, 2016.
\newblock \href {https://arxiv.org/abs/1609.09098} {\path{arXiv:1609.09098}}.

\bibitem{GeelenKMW23Thegrid}
Jim Geelen, O-joung Kwon, Rose McCarty, and Paul Wollan.
\newblock The grid theorem for vertex-minors.
\newblock {\em J. Combin. Theory Ser. B}, 158:93--116, 2023.
\newblock \href {https://doi.org/10.1016/j.jctb.2020.08.004}
  {\path{doi:10.1016/j.jctb.2020.08.004}}.

\bibitem{GiannopoulouKRT17packi}
Archontia~C. Giannopoulou, O{-}joung Kwon, Jean{-}Florent Raymond, and
  Dimitrios~M. Thilikos.
\newblock Packing and covering immersion-expansions of planar sub-cubic graphs.
\newblock {\em European Journal of Combinatorics}, 65:154--167, 2017.
\newblock \href {https://doi.org/10.1016/j.ejc.2017.05.009}
  {\path{doi:10.1016/j.ejc.2017.05.009}}.

\bibitem{GloverHW79graphsthat}
Henry~H Glover, John~P Huneke, and Chin~San Wang.
\newblock 103 graphs that are irreducible for the projective plane.
\newblock {\em Journal of Combinatorial Theory, Series B}, 27(3):332--370,
  1979.
\newblock URL:
  \url{https://www.sciencedirect.com/science/article/pii/0095895679900224},
  \href {https://doi.org/10.1016/0095-8956(79)90022-4}
  {\path{doi:10.1016/0095-8956(79)90022-4}}.

\bibitem{GroheKMW11Finding}
Martin Grohe, Ken{-}ichi Kawarabayashi, D{\'{a}}niel Marx, and Paul Wollan.
\newblock Finding topological subgraphs is fixed-parameter tractable.
\newblock In {\em Annual {ACM} on Symposium on Theory of Computing (STOC)},
  pages 479--488, 2011.
\newblock \href {https://doi.org/10.1145/1993636.1993700}
  {\path{doi:10.1145/1993636.1993700}}.

\bibitem{GuT12ImprovedBounds}
Qian{-}Ping Gu and Hisao Tamaki.
\newblock Improved bounds on the planar branchwidth with respect to the largest
  grid minor size.
\newblock {\em Algorithmica}, 64(3):416--453, 2012.
\newblock \href {https://doi.org/10.1007/s00453-012-9627-5}
  {\path{doi:10.1007/s00453-012-9627-5}}.

\bibitem{Hadwiger43Uber}
Hugo Hadwiger.
\newblock {\"U}ber eine {K}lassifikation der {S}treckenkomplex.
\newblock {\em Vierteljschr. Naturforsch. Ges. Z{\"u}rich}, 88:133--143, 1943.

\bibitem{Halin91Tree}
Rudolf Halin.
\newblock Tree-partitions of infinite graphs.
\newblock {\em Discrete Mathematics}, 97(1-3):203--217, 1991.

\bibitem{HarveyW17Parameters}
Daniel~J. Harvey and David~R. Wood.
\newblock Parameters tied to treewidth.
\newblock {\em J. Graph Theory}, 84(4):364--385, 2017.
\newblock \href {https://doi.org/10.1002/jgt.22030}
  {\path{doi:10.1002/jgt.22030}}.

\bibitem{HatzelRabinovichWiederrecht2019Cyclewidth}
Meike Hatzel, Roman Rabinovich, and Sebastian Wiederrecht.
\newblock Cyclewidth and the grid theorem for perfect matching width of
  bipartite graphs.
\newblock In Ignasi Sau and Dimitrios~M. Thilikos, editors, {\em
  Graph-Theoretic Concepts in Computer Science - 45th International Workshop,
  {WG} 2019, Vall de N{\'u}ria, Spain, June 19--21, 2019, Revised Papers},
  volume 11789 of {\em Lecture Notes in Computer Science}, pages 53--65.
  Springer, 2019.
\newblock \href {https://arxiv.org/abs/1902.01322} {\path{arXiv:1902.01322}},
  \href {https://doi.org/10.1007/978-3-030-30786-8_5}
  {\path{doi:10.1007/978-3-030-30786-8_5}}.

\bibitem{HlinenyO08FindingBranch}
Petr Hlin{\v e}n{\'{y}} and Sang{-}il Oum.
\newblock Finding branch-decompositions and rank-decompositions.
\newblock {\em {SIAM} Journal on Computing}, 38(3):1012--1032, 2008.
\newblock \href {https://doi.org/10.1137/070685920}
  {\path{doi:10.1137/070685920}}.

\bibitem{HlinenyOSG07Width}
Petr Hliněný, Sang-il Oum, Detlef Seese, and Georg Gottlob.
\newblock Width parameters beyond tree-width and their applications.
\newblock {\em The Computer Journal}, 51(3):326--362, 09 2007.
\newblock \href
  {https://arxiv.org/abs/https://academic.oup.com/comjnl/article-pdf/51/3/326/1139343/bxm052.pdf}
  {\path{arXiv:https://academic.oup.com/comjnl/article-pdf/51/3/326/1139343/bxm052.pdf}},
  \href {https://doi.org/10.1093/comjnl/bxm052}
  {\path{doi:10.1093/comjnl/bxm052}}.

\bibitem{HuynhJMSW22Excluding}
Tony Huynh, Gwena\"el Joret, Piotr Micek, Micha\l~T. Seweryn, and Paul Wollan.
\newblock Excluding a ladder.
\newblock {\em Combinatorica}, 42(3):405--432, 2022.
\newblock \href {https://doi.org/10.1007/s00493-021-4592-8}
  {\path{doi:10.1007/s00493-021-4592-8}}.

\bibitem{HuynhJMW20Seymour}
Tony Huynh, Gwena{\"e}l Joret, Piotr Micek, and David~R. Wood.
\newblock Seymour's conjecture on $2$-connected graphs of large pathwidth.
\newblock {\em Combinatorica}, 40:839--868, 2020.
\newblock \href {https://doi.org/10.1007/s00493-020-3941-3}
  {\path{doi:10.1007/s00493-020-3941-3}}.

\bibitem{Oum2017RankWidth}
Sang il~Oum.
\newblock Rank-width: Algorithmic and structural results.
\newblock {\em Discrete Applied Mathematics}, 231:15--24, November 2017.
\newblock \href {https://doi.org/10.1016/j.dam.2016.08.006}
  {\path{doi:10.1016/j.dam.2016.08.006}}.

\bibitem{Jancar99ANote}
Petr Jancar.
\newblock A note on well quasi-orderings for powersets.
\newblock {\em Inf. Process. Lett.}, 72(5-6):155--160, 1999.
\newblock \href {https://doi.org/10.1016/S0020-0190(99)00149-0}
  {\path{doi:10.1016/S0020-0190(99)00149-0}}.

\bibitem{JobsonK21allm}
Adam~S. Jobson and Andr\'{e}~E. K\'{e}zdy.
\newblock All minor-minimal apex obstructions with connectivity two.
\newblock {\em Electronic Journal of Combinatorics}, 28(1):1.23, 58, 2021.
\newblock \href {https://doi.org/10.37236/8382} {\path{doi:10.37236/8382}}.

\bibitem{Kaminski12Maxcut}
Marcin Kami\'nski.
\newblock M{AX}-{CUT} and containment relations in graphs.
\newblock {\em Theoret. Comput. Sci.}, 438:89--95, 2012.
\newblock \href {https://doi.org/10.1016/j.tcs.2012.02.036}
  {\path{doi:10.1016/j.tcs.2012.02.036}}.

\bibitem{KanteK18Linear}
Mamadou~Moustapha Kant{\'e} and O{-}joung Kwon.
\newblock Linear rank-width of distance-hereditary graphs {II}. vertex-minor
  obstructions.
\newblock {\em European Journal of Combinatorics}, 74:110--139, 2018.
\newblock \href {https://doi.org/10.1016/j.ejc.2018.07.009}
  {\path{doi:10.1016/j.ejc.2018.07.009}}.

\bibitem{KawarabayashiKR12Thedisjoint}
Ken-ichi Kawarabayashi, Yusuke Kobayashi, and Bruce Reed.
\newblock The disjoint paths problem in quadratic time.
\newblock {\em J. Combin. Theory Ser. B}, 102(2):424--435, 2012.
\newblock \href {https://doi.org/10.1016/j.jctb.2011.07.004}
  {\path{doi:10.1016/j.jctb.2011.07.004}}.

\bibitem{KawarabayashiK15TheDirected}
Ken{-}ichi Kawarabayashi and Stephan Kreutzer.
\newblock The directed grid theorem.
\newblock In {\em Annual {ACM} on Symposium on Theory of Computing (STOC)},
  pages 655--664. {ACM}, 2015.
\newblock \href {https://doi.org/10.1145/2746539.2746586}
  {\path{doi:10.1145/2746539.2746586}}.

\bibitem{KawarabayashiRossman2022Polynomial}
Ken{-}ichi Kawarabayashi and Benjamin Rossman.
\newblock A polynomial excluded-minor approximation of treedepth.
\newblock {\em Journal of the European Mathematical Society}, 24(4):1449--1470,
  2022.
\newblock \href {https://doi.org/10.4171/JEMS/1133}
  {\path{doi:10.4171/JEMS/1133}}.

\bibitem{kawarabayashi2020quickly}
Ken-ichi Kawarabayashi, Robin Thomas, and Paul Wollan.
\newblock Quickly excluding a non-planar graph.
\newblock {\em arXiv preprint arXiv:2010.12397}, 2020.
\newblock URL: \url{https://arxiv.org/abs/2010.12397}.

\bibitem{Kinnersley92TheVertex}
Nancy~G. Kinnersley.
\newblock The vertex separation number of a graph equals its path-width.
\newblock {\em Inf. Process. Lett.}, 42(6):345--350, 1992.
\newblock \href {https://doi.org/10.1016/0020-0190(92)90234-M}
  {\path{doi:10.1016/0020-0190(92)90234-M}}.

\bibitem{KirousisP85inte}
Lefteris~M. Kirousis and Christos~H. Papadimitriou.
\newblock Interval graphs and searching.
\newblock {\em Discrete Mathematics}, 55(2):181--184, 1985.
\newblock \href {https://doi.org/10.1016/0012-365X(85)90046-9}
  {\path{doi:10.1016/0012-365X(85)90046-9}}.

\bibitem{KirousisP86Searching}
Lefteris~M. Kirousis and Christos~H. Papadimitriou.
\newblock Searching and pebbling.
\newblock {\em Theoretical Computer Science}, 47(2):205--212, 1986.
\newblock \href {https://doi.org/10.1016/0304-3975(86)90146-5}
  {\path{doi:10.1016/0304-3975(86)90146-5}}.

\bibitem{KorhonenSokolowski2024Rankwidth}
Tuukka Korhonen and Marek Soko{\l}owski.
\newblock Almost-linear time parameterized algorithm for rankwidth via dynamic
  rankwidth.
\newblock In {\em Proceedings of the 56th Annual {ACM} Symposium on Theory of
  Computing, {STOC} 2024}, pages 1538--1549. ACM, 2024.
\newblock \href {https://arxiv.org/abs/2402.12364} {\path{arXiv:2402.12364}},
  \href {https://doi.org/10.1145/3618260.3649732}
  {\path{doi:10.1145/3618260.3649732}}.

\bibitem{KwonMOW21Obstructions}
O{-}joung Kwon, Rose McCarty, {Sang-il} Oum, and Paul Wollan.
\newblock Obstructions for bounded shrub-depth and rank-depth.
\newblock {\em Journal of Combinatorial Theory, Series {B}}, 149:76--91, 2021.
\newblock \href {https://doi.org/10.1016/j.jctb.2021.01.005}
  {\path{doi:10.1016/j.jctb.2021.01.005}}.

\bibitem{KwonO14Graphs}
O{-}joung Kwon and {Sang-il} Oum.
\newblock Graphs of small rank-width are pivot-minors of graphs of small
  tree-width.
\newblock {\em Discrete Applied Mathematics}, 168:108--118, 2014.
\newblock \href {https://doi.org/10.1016/j.dam.2013.01.007}
  {\path{doi:10.1016/j.dam.2013.01.007}}.

\bibitem{KwonO14Unavoidable}
O-joung Kwon and Sang-il Oum.
\newblock Unavoidable vertex-minors in large prime graphs.
\newblock {\em European J. Combin.}, 41:100--127, 2014.
\newblock \href {https://doi.org/10.1016/j.ejc.2014.03.013}
  {\path{doi:10.1016/j.ejc.2014.03.013}}.

\bibitem{KwonO20Scattered}
O-joung Kwon and Sang-il Oum.
\newblock Scattered classes of graphs.
\newblock {\em SIAM J. Discrete Math.}, 34(1):972--999, 2020.
\newblock \href {https://doi.org/10.1137/19M1293776}
  {\path{doi:10.1137/19M1293776}}.

\bibitem{KwonO21Graphs}
O-joung Kwon and Sang-il Oum.
\newblock Graphs of bounded depth-2 rank-brittleness.
\newblock {\em J. Graph Theory}, 96(3):361--378, 2021.
\newblock \href {https://doi.org/10.1002/jgt.22619}
  {\path{doi:10.1002/jgt.22619}}.

\bibitem{Lagergren98}
Jens Lagergren.
\newblock Upper bounds on the size of obstructions and intertwines.
\newblock {\em J. Combin. Theory Ser. B}, 73(1):7--40, 1998.
\newblock \href {https://doi.org/10.1006/jctb.1997.1788}
  {\path{doi:10.1006/jctb.1997.1788}}.

\bibitem{LagergrenA91mini}
Jens Lagergren and Stefan Arnborg.
\newblock Finding minimal forbidden minors using a finite congruence.
\newblock In {\em Proc. of the 18th International Colloquium on Automata,
  Languages and Programming (ICALP)}, volume 510 of {\em LNCS}, pages 532--543,
  1991.
\newblock \href {https://doi.org/10.1007/3-540-54233-7\_161}
  {\path{doi:10.1007/3-540-54233-7\_161}}.

\bibitem{LardasPTZ23OnStrict}
Emmanouil Lardas, Evangelos Protopapas, Dimitrios~M. Thilikos, and Dimitris
  Zoros.
\newblock On strict brambles.
\newblock {\em Graphs and Combinatorics}, 39(2):24, 2023.
\newblock \href {https://doi.org/10.1007/s00373-023-02618-y}
  {\path{doi:10.1007/s00373-023-02618-y}}.

\bibitem{LimniosPPT20Edge}
Stratis Limnios, Christophe Paul, Joanny Perret, and Dimitrios~M. Thilikos.
\newblock Edge degeneracy: Algorithmic and structural results.
\newblock {\em Theor. Comput. Sci.}, 839:164--175, 2020.
\newblock \href {https://doi.org/10.1016/j.tcs.2020.06.006}
  {\path{doi:10.1016/j.tcs.2020.06.006}}.

\bibitem{LiptonMMPRT16sixv}
Max Lipton, Eoin Mackall, Thomas~W. Mattman, Mike Pierce, Samantha Robinson,
  Jeremy Thomas, and Ilan Weinschelbaum.
\newblock Six variations on a theme: almost planar graphs.
\newblock {\em Involve. A Journal of Mathematics}, 11(3):413--448, 2018.
\newblock \href {https://doi.org/10.2140/involve.2018.11.413}
  {\path{doi:10.2140/involve.2018.11.413}}.

\bibitem{Lopez2024Labelled}
Aliaume Lopez.
\newblock Labelled well quasi ordered classes of bounded linear clique-width.
\newblock {\em CoRR}, abs/2405.10894, 2024.
\newblock Preprint.
\newblock URL: \url{https://arxiv.org/abs/2405.10894}.

\bibitem{Lozin23Hereditary}
Vadim Lozin.
\newblock Hereditary classes of graphs: a parametric approach.
\newblock {\em Discrete Appl. Math.}, 325:134--151, 2023.
\newblock \href {https://doi.org/10.1016/j.dam.2022.10.016}
  {\path{doi:10.1016/j.dam.2022.10.016}}.

\bibitem{LozinR22Treewidth}
Vadim Lozin and Igor Razgon.
\newblock Tree-width dichotomy.
\newblock {\em European Journal of Combinatorics}, 103:103517, 2022.
\newblock URL:
  \url{https://www.sciencedirect.com/science/article/pii/S0195669822000130},
  \href {https://doi.org/10.1016/j.ejc.2022.103517}
  {\path{doi:10.1016/j.ejc.2022.103517}}.

\bibitem{Lozin11Minimal}
Vadim~V. Lozin.
\newblock Minimal classes of graphs of unbounded clique-width.
\newblock {\em Ann. Comb.}, 15(4):707--722, 2011.
\newblock \href {https://doi.org/10.1007/s00026-011-0117-2}
  {\path{doi:10.1007/s00026-011-0117-2}}.

\bibitem{MagnePaulSharmaThilikos2023EdgeTreewidth}
Lo{\"i}c Magne, Christophe Paul, Abhijat Sharma, and Dimitrios~M. Thilikos.
\newblock Edge-treewidth: Algorithmic and combinatorial properties.
\newblock {\em Discrete Applied Mathematics}, 341:40--54, 2023.
\newblock \href {https://doi.org/10.1016/j.dam.2023.07.023}
  {\path{doi:10.1016/j.dam.2023.07.023}}.

\bibitem{Mahlmann25Forbidden}
Nikolas M{\"{a}}hlmann.
\newblock Forbidden induced subgraphs for bounded shrub-depth and the
  expressive power of {MSO}.
\newblock In {\em International Colloquium on Automata, Languages, and
  Programming, ({ICALP})}, volume 334 of {\em LIPIcs}, pages 167:1--167:18.
  Schloss Dagstuhl - Leibniz-Zentrum f{\"{u}}r Informatik, 2025.
\newblock \href {https://doi.org/10.4230/LIPICS.ICALP.2025.167}
  {\path{doi:10.4230/LIPICS.ICALP.2025.167}}.

\bibitem{MarshallW15circumference}
Emily~Abernethy Marshall and David~R. Wood.
\newblock Circumference and pathwidth of highly connected graphs.
\newblock {\em J. Graph Theory}, 79(3):222--232, 2015.
\newblock \href {https://doi.org/10.1002/jgt.21825}
  {\path{doi:10.1002/jgt.21825}}.

\bibitem{Mattman16forb}
Thomas~W. Mattman.
\newblock Forbidden minors: finding the finite few.
\newblock In {\em A primer for undergraduate research}, Found. Undergrad. Res.
  Math., pages 85--97. Birkh\"{a}user/Springer, Cham, 2017.

\bibitem{MattmanPierce2017}
Thomas~W. Mattman and Mike Pierce.
\newblock The {$K_{n+5}$} and {$K_{3^2,1^n}$} families and obstructions to
  {$n$-apex}.
\newblock In {\em Contemporary Mathematics}, volume 689. American Mathematical
  Society, 2017.
\newblock \href {https://doi.org/10.1090/conm/689}
  {\path{doi:10.1090/conm/689}}.

\bibitem{MescoffPT22Themixed}
Guillaume Mescoff, Christophe Paul, and Dimitrios~M. Thilikos.
\newblock The mixed search game against an agile and visible fugitive is
  monotone.
\newblock {\em CoRR}, abs/2204.10691, 2022.
\newblock \href {https://arxiv.org/abs/2204.10691} {\path{arXiv:2204.10691}},
  \href {https://doi.org/10.48550/arXiv.2204.10691}
  {\path{doi:10.48550/arXiv.2204.10691}}.

\bibitem{MoharT01Graphs}
Bojan Mohar and Carsten Thomassen.
\newblock {\em Graphs on Surfaces}.
\newblock Johns Hopkins series in the mathematical sciences. Johns Hopkins
  University Press, 2001.
\newblock URL:
  \url{https://www.press.jhu.edu/books/title/1675/graphs-surfaces}.

\bibitem{Moehring90grap}
Rolf~H. M{\"o}hring.
\newblock Graph problems related to gate matrix layout and {PLA} folding.
\newblock In {\em Computational graph theory}, volume~7 of {\em Computing
  Supplementum}, pages 17--51. Springer, 1990.
\newblock \href {https://doi.org/10.1007/978-3-7091-9076-0_2}
  {\path{doi:10.1007/978-3-7091-9076-0_2}}.

\bibitem{NestoridisT14Squareroots}
Nestor~V. Nestoridis and Dimitrios~M. Thilikos.
\newblock Square roots of minor closed graph classes.
\newblock {\em Discret. Appl. Math.}, 168:34--39, 2014.
\newblock \href {https://doi.org/10.1016/j.dam.2013.05.026}
  {\path{doi:10.1016/j.dam.2013.05.026}}.

\bibitem{NesetrilO06Treedepth}
Jaroslav Ne\v{s}et\v{r}il and Patrice~Ossona de~Mendez.
\newblock Tree-depth, subgraph coloring and homomorphism bounds.
\newblock {\em European Journal of Combinatorics}, 27:1022--1041, 2006.
\newblock \href {https://doi.org/10.1016/j.ejc.2005.01.010}
  {\path{doi:10.1016/j.ejc.2005.01.010}}.

\bibitem{NguyenO20Theaverage}
Huy-Tung Nguyen and Sang-il Oum.
\newblock The average cut-rank of graphs.
\newblock {\em European J. Combin.}, 90:103183, 22, 2020.
\newblock \href {https://doi.org/10.1016/j.ejc.2020.103183}
  {\path{doi:10.1016/j.ejc.2020.103183}}.

\bibitem{OumS06appro}
{Sang-il} Oum and Paul~D. Seymour.
\newblock Approximating rank-width and branchwidth.
\newblock {\em Journal of Combinatorial Theory, Series B}, 96(4):514--528,
  2006.
\newblock \href {https://doi.org/10.1016/j.jctb.2005.10.006}
  {\path{doi:10.1016/j.jctb.2005.10.006}}.

\bibitem{PaulPT2023GraphParameters}
Christophe Paul, Evangelos Protopapas, and Dimitrios~M. Thilikos.
\newblock Graph parameters, universal obstructions, and wqo.
\newblock {\em Order}, 42(4):849--894, 2025.
\newblock \href {https://doi.org/10.1007/s11083-025-09713-0}
  {\path{doi:10.1007/s11083-025-09713-0}}.

\bibitem{PaulPTW24Delineating}
Christophe Paul, Evangelos Protopapas, Dimitrios~M. Thilikos, and Sebastian
  Wiederrecht.
\newblock Delineating half-integrality of the {E}rd{\H o}s-{P}{\'o}sa property
  for minors: the case of surfaces.
\newblock In {\em 51st {I}nternational {C}olloquium on {A}utomata, {L}anguages,
  and {P}rogramming}, volume 297 of {\em LIPIcs. Leibniz Int. Proc. Inform.},
  pages Art. No. 114, 19. Schloss Dagstuhl. Leibniz-Zent. Inform., Wadern,
  2024.
\newblock \href {https://doi.org/10.4230/lipics.icalp.2024.114}
  {\path{doi:10.4230/lipics.icalp.2024.114}}.

\bibitem{PaulPTW24Obstructions}
Christophe Paul, Evangelos Protopapas, Dimitrios~M. Thilikos, and Sebastian
  Wiederrecht.
\newblock {O}bstructions to {E}rd{\H o}s-{P}{\'o}sa dualities for minors, 2024.
\newblock URL: \url{https://arxiv.org/abs/2407.09671}, \href
  {https://arxiv.org/abs/2407.09671} {\path{arXiv:2407.09671}}.

\bibitem{Pouzet1972BelOrdre}
Maurice Pouzet.
\newblock Un bel ordre d’abritement et ses rapports avec les bornes d’une
  multirelation.
\newblock {\em C. R. Acad. Sci. Paris Sér. A-B}, 274:A1677--A1680, 1972.

\bibitem{Rambaud2025Excluding}
Clément Rambaud.
\newblock Excluding a rectangular grid, 2025.
\newblock URL: \url{https://arxiv.org/abs/2501.11617}, \href
  {https://arxiv.org/abs/2501.11617} {\path{arXiv:2501.11617}}.

\bibitem{Ramsey30Ona}
F.~P. Ramsey.
\newblock On a problem of formal logic.
\newblock {\em Proceedings of the London Mathematical Society},
  s2-30(1):264--286, 1930.
\newblock \href {https://doi.org/10.1112/plms/s2-30.1.264}
  {\path{doi:10.1112/plms/s2-30.1.264}}.

\bibitem{RobertsonS86GMV}
Neil Robertson and P.~D. Seymour.
\newblock Graph minors. {V}. {E}xcluding a planar graph.
\newblock {\em J. Combin. Theory Ser. B}, 41(1):92--114, 1986.
\newblock \href {https://doi.org/10.1016/0095-8956(86)90030-4}
  {\path{doi:10.1016/0095-8956(86)90030-4}}.

\bibitem{RobertsonS04GMXX}
Neil Robertson and P.~D. Seymour.
\newblock {Graph Minors. XX. Wagner's conjecture}.
\newblock {\em Journal of Combinatorial Theory, Series B}, 92(2):325--357,
  2004.
\newblock \href {https://doi.org/10.1016/j.jctb.2004.08.001}
  {\path{doi:10.1016/j.jctb.2004.08.001}}.

\bibitem{RobertsonS84GMIII}
Neil Robertson and Paul Seymour.
\newblock Graph minors. {III}. planar tree-width.
\newblock {\em Journal of Combinatorial Theory Series B}, 36:49--64, 1984.
\newblock \href {https://doi.org/10.1016/0095-8956(84)90013-3}
  {\path{doi:10.1016/0095-8956(84)90013-3}}.

\bibitem{RobertsonST94Quickly}
Neil Robertson, Paul Seymour, and Robin Thomas.
\newblock Quickly excluding a planar graph.
\newblock {\em J. Combin. Theory Ser. B}, 62(2):323--348, 1994.
\newblock \href {https://doi.org/10.1006/jctb.1994.1073}
  {\path{doi:10.1006/jctb.1994.1073}}.

\bibitem{RobertsonS83GMI}
Neil Robertson and Paul~D. Seymour.
\newblock Graph minors. {I}. excluding a forest.
\newblock {\em Journal of Combinatorial Theory Series B}, 35:36--61, 1983.
\newblock \href {https://doi.org/10.1016/0095-8956(83)90079-5}
  {\path{doi:10.1016/0095-8956(83)90079-5}}.

\bibitem{RobertsonS91GMX}
Neil Robertson and Paul~D. Seymour.
\newblock Graph minors. {X}. obstructions to tree-decomposition.
\newblock {\em Journal of Combinatorial Theory, Series B}, 52(2):153--190,
  1991.
\newblock \href {https://doi.org/10.1016/0095-8956(91)90061-N}
  {\path{doi:10.1016/0095-8956(91)90061-N}}.

\bibitem{RobertsonS93Excluding}
Neil Robertson and Paul~D. Seymour.
\newblock Excluding a graph with one crossing.
\newblock In {\em Graph structure theory}, volume 147, pages 669--675. American
  Mathematical Society, 1993.
\newblock \href {https://doi.org/10.1090/conm/147}
  {\path{doi:10.1090/conm/147}}.

\bibitem{RobertsonS95GMXIII}
Neil Robertson and Paul~D. Seymour.
\newblock Graph minors . {XIII.} {The} disjoint paths problem.
\newblock {\em Journal of Combinatorial Theory, Series B}, 63(1):65--110, 1995.
\newblock \href {https://doi.org/10.1006/jctb.1995.1006}
  {\path{doi:10.1006/jctb.1995.1006}}.

\bibitem{RobertsonS03GMXVI}
Neil Robertson and Paul~D. Seymour.
\newblock Graph minors. {XVI.} excluding a non-planar graph.
\newblock {\em Journal of Combinatorial Theory, Series B}, 89(1):43--76, 2003.
\newblock \href {https://doi.org/10.1016/S0095-8956(03)00042-X}
  {\path{doi:10.1016/S0095-8956(03)00042-X}}.

\bibitem{RobertsonS10GMXXIII}
Neil Robertson and Paul~D. Seymour.
\newblock Graph minors. {XXIII.} {N}ash-{W}illiams' immersion conjecture.
\newblock {\em Journal of Combinatorial Theory, Series B}, 100(2):181--205,
  2010.
\newblock \href {https://doi.org/10.1016/j.jctb.2009.07.003}
  {\path{doi:10.1016/j.jctb.2009.07.003}}.

\bibitem{RobertsonST95sachs}
Neil Robertson, Paul~D. Seymour, and Robin Thomas.
\newblock Sachs' linkless embedding conjecture.
\newblock {\em J. Comb. Theory, Ser. {B}}, 64(2):185--227, 1995.
\newblock \href {https://doi.org/10.1006/jctb.1995.1032}
  {\path{doi:10.1006/jctb.1995.1032}}.

\bibitem{SauST23apices}
Ignasi Sau, Giannos Stamoulis, and Dimitrios~M. Thilikos.
\newblock $k$-apices of minor-closed graph classes. {I}. {B}ounding the
  obstructions.
\newblock {\em Journal of Combinatorial Theory, Series B}, 161:180--227, 2023.
\newblock \href {https://doi.org/10.1016/j.jctb.2023.02.012}
  {\path{doi:10.1016/j.jctb.2023.02.012}}.

\bibitem{Seese85Tree}
Detlef Seese.
\newblock {\em Tree-partite graphs and the complexity of algorithms}.
\newblock Springer, 1985.

\bibitem{SeymourT93graph}
P.~D. Seymour and Robin Thomas.
\newblock Graph searching and a min-max theorem for tree-width.
\newblock {\em J. Combin. Theory Ser. B}, 58(1):22--33, 1993.
\newblock \href {https://doi.org/10.1006/jctb.1993.1027}
  {\path{doi:10.1006/jctb.1993.1027}}.

\bibitem{SeymourT94CallRouting}
Paul~D. Seymour and Robin Thomas.
\newblock Call routing and the ratcatcher.
\newblock {\em Comb.}, 14(2):217--241, 1994.
\newblock \href {https://doi.org/10.1007/BF01215352}
  {\path{doi:10.1007/BF01215352}}.

\bibitem{TakahashiUK95b}
Atsushi Takahashi, Shuichi Ueno, and Yoji Kajitani.
\newblock Minimal forbidden minors for the family of graphs with
  proper-path-width at most two.
\newblock {\em IEICE Trans. Fundamentals}, E78-A:1828--1839, 1995.

\bibitem{TakahashiYK95Mixed}
Atsushi Takahashi, Shuichi Ueno, and Yoji Kajitani.
\newblock Mixed searching and proper-path-width.
\newblock {\em Theoretical Computer Science}, 137:253--268, 1995.
\newblock \href {https://doi.org/10.1016/0304-3975(94)00160-K}
  {\path{doi:10.1016/0304-3975(94)00160-K}}.

\bibitem{Thilikos12GraphMinors}
Dimitrios~M. Thilikos.
\newblock Graph minors and parameterized algorithm design.
\newblock In Hans~L. Bodlaender, Rod Downey, Fedor~V. Fomin, and D{\'{a}}niel
  Marx, editors, {\em The Multivariate Algorithmic Revolution and Beyond -
  Essays Dedicated to Michael R. Fellows on the Occasion of His 60th Birthday},
  volume 7370 of {\em Lecture Notes in Computer Science}, pages 228--256.
  Springer, 2012.
\newblock \href {https://doi.org/10.1007/978-3-642-30891-8\_13}
  {\path{doi:10.1007/978-3-642-30891-8\_13}}.

\bibitem{ThilikosW22Killing}
Dimitrios~M. Thilikos and Sebastian Wiederrecht.
\newblock Killing a vortex.
\newblock In {\em IEEE Annual Symposium on Foundations of Computer Science
  (FOCS)}, pages 1069--1080, 2022.
\newblock \href {https://doi.org/10.1109/FOCS54457.2022.00104}
  {\path{doi:10.1109/FOCS54457.2022.00104}}.

\bibitem{thilikos2023excluding}
Dimitrios~M. Thilikos and Sebastian Wiederrecht.
\newblock Excluding surfaces as minors in graphs, 2023.
\newblock \href {https://arxiv.org/abs/2306.01724} {\path{arXiv:2306.01724}}.

\bibitem{Thomas1989wellquasi}
Robin Thomas.
\newblock Well-quasi-ordering infinite graphs with forbidden finite planar
  minor.
\newblock {\em Trans. Amer. Math. Soc.}, 312(1):279--313, 1989.
\newblock \href {https://doi.org/10.2307/2001217} {\path{doi:10.2307/2001217}}.

\bibitem{Wagner37Uber}
Klaus Wagner.
\newblock {\"U}ber eine eigenschaft der ebenen komplexe.
\newblock {\em Mathematische Annalen}, 114:570--590, 1937.
\newblock \href {https://doi.org/10.1007/BF01594196}
  {\path{doi:10.1007/BF01594196}}.

\bibitem{Wollan15Thestructure}
Paul Wollan.
\newblock The structure of graphs not admitting a fixed immersion.
\newblock {\em J. Comb. Theory, Ser. {B}}, 110:47--66, 2015.
\newblock \href {https://doi.org/10.1016/j.jctb.2014.07.003}
  {\path{doi:10.1016/j.jctb.2014.07.003}}.

\bibitem{Wollan16Finding}
Paul Wollan.
\newblock Finding topological subgraphs.
\newblock In {\em Encyclopedia of Algorithms}, pages 749--752. Springer, 2016.
\newblock \href {https://doi.org/10.1007/978-1-4939-2864-4_695}
  {\path{doi:10.1007/978-1-4939-2864-4_695}}.

\bibitem{Wood09ontre}
David~R. Wood.
\newblock On tree-partition-width.
\newblock {\em European Journal of Combinatorics}, 30(5):1245--1253, 2009.
\newblock \href {https://doi.org/10.1016/j.ejc.2008.11.010}
  {\path{doi:10.1016/j.ejc.2008.11.010}}.

\bibitem{Yu06more}
Yaming Yu.
\newblock More forbidden minors for {W}ye-{D}elta-{W}ye reducibility.
\newblock {\em Electronic Journal of Combinatorics}, 13(1):7:15, 2006.
\newblock \href {https://doi.org/10.37236/1033} {\path{doi:10.37236/1033}}.

\end{thebibliography}

\end{document}